\documentclass[10pt,journal,compsoc]{IEEEtran}

\usepackage[nocompress]{cite}

\usepackage{amsmath,amssymb,amsfonts}
\usepackage{amsthm}
\usepackage{mathtools}

\usepackage{algorithmic}

\usepackage{graphicx}
\usepackage{tikz}
\usetikzlibrary{arrows.meta,positioning,fit}
\usepackage{textcomp}
\usepackage{xcolor}
\usepackage{pifont}   % \ding{51}/\ding{55} used in the matrix table
\usepackage{enumitem} % [leftmargin=...] used on itemize environments

\usepackage{booktabs}
\usepackage{tabularx}
\usepackage{makecell}

\usepackage{listings}

\usepackage[hyphens]{url}
\usepackage{seqsplit}

\usepackage{microtype}
\usepackage[colorlinks=true,linkcolor=black,citecolor=black,urlcolor=black]{hyperref}
\hypersetup{
  pdftitle={Stop Means Stop: Measuring and Repairing the Enforcement Gap in Agent-Framework Control Primitives},
  pdfauthor={Sajjad Khan},
  pdfkeywords={LLM agents, multi-agent systems, human-in-the-loop, control plane, idempotency, cancellation, agent frameworks},
}

\newtheorem{property}{Property}
\newtheorem{definition}{Definition}
\newtheorem{lemma}{Lemma}
\newtheorem{remark}{Remark}
\newtheorem{proposition}{Proposition}

\newcommand{\CPUMODEL}{AMD Ryzen~7 PRO 6850U}
\newcommand{\LOADENV}{\CPUMODEL{} workstation (8C/16T), Rust 1.95, Linux}% concurrent-load primary environment; single-vCPU container floor in the Sec. 7.3 footnote
\newcommand{\SOCKETMED}{53\,\mu\mathrm{s}}% measured: bench_socket.py, 20k round-trips, Ryzen 7 PRO 6850U / Ubuntu 24.04
\newcommand{\sys}{\textsc{SoundGate}}

\begin{document}

\title{Stop Means Stop: Measuring and Repairing the\\ Enforcement Gap in Agent-Framework Control Primitives}

\author{Sajjad~Khan% <-this % stops a space
\thanks{S. Khan is an independent researcher, London, UK
(e-mail: sajjadanwar200@gmail.com).}%
\thanks{Artifact (probes, harness, formal models, and \sys{} reference
implementation), with a single-command audit (\texttt{reproduce.sh}) that
re-derives every headline number from committed data, available for review
at: \url{https://anonymous.4open.science/r/soundgate-paper-B576}. The
repository will be made public at
\url{https://github.com/sajjadanwar0/soundgate-paper} upon acceptance. The
gate is installable from PyPI: \texttt{pip install soundgate} (v0.1.0).}%
}

\IEEEtitleabstractindextext{%
\begin{abstract}
Production LLM-agent frameworks ship \emph{control primitives}---human-in-the-loop
approval gates, run cancellation, and execution timeouts---whose names and
documentation imply \emph{barrier} semantics: while a run is paused, cancelled,
or timed out, no gated side effect executes. This contract holds on none of six
widely used open-source frameworks. Model-free differential probes isolate a
recurring \emph{sibling leak}---an approval gate suspends its own branch while
a sibling's effect executes \emph{during the pause}, defeating rejection---in
every framework shipping a pre-execution gate (five of six, four execution
models, two language runtimes), and confirm replay double-execution,
cancellation orphans, and timeout zombies. The hazard is reachable: frontier
models emit the leak-triggering plan shape at rates up to 14\%, and live
models driving unmodified frameworks leak $215$ of $1{,}200$ runs
($P(\mathrm{leak}\mid\mathrm{emitted}){=}1.00$); on naturalistic $\tau$-bench
episodes models serialize writes---the everyday gap is \emph{latent}---while
injection induces it deterministically and a 13-incident public corpus
corroborates the replay and cancellation failures. We repair the gaps with
\sys{}, an \emph{environment-external} Rust gate through which every side
effect must be admitted, enforcing hold-until-decided, reject-cancels,
dedup-on-replay, and fence-on-cancel under a stated
\emph{complete-mediation} contract, discharged for network egress by two
kernel-enforced routes. The admission core is mechanically verified (Verus;
TLA+/TLC to $7.5{\times}10^{7}$ states; TLAPS; Loom on the deployed Rust) and
bridged to code by differential conformance over $1.2{\times}10^{7}$
operations with zero divergences. Under that contract \sys{} blocks every
measured violation on all six frameworks while releasing legitimate effects:
gated $\tau$-bench episodes complete with zero refusals at ${\sim}1$\,ms per
write, and durable admission sustains ${\sim}12$k admissions per second.
\end{abstract}

}

\maketitle
\IEEEdisplaynontitleabstractindextext
\IEEEpeerreviewmaketitle

% =====================================================================
\section{Introduction}
\label{sec:intro}
\IEEEPARstart{L}{arge}-language-model (LLM) agents increasingly act on external
systems: they send e-mail, open tickets, modify databases, deploy code, and
authorize payments. Because these actions are irreversible and occasionally
catastrophic, every major agent framework ships \emph{control primitives} that
insert human authority or hard limits into an otherwise autonomous loop: a
human-in-the-loop (HITL) \emph{approval gate} that pauses before a sensitive
action, run \emph{cancellation}, and execution \emph{timeouts}. Practitioner
guidance treats these as the load-bearing safety mechanism for agents with real
tool access, and emerging regulation is beginning to treat demonstrable human
oversight as a compliance obligation rather than a best practice.

The value of a stop primitive rests on an implicit \emph{barrier} assumption:
while a run is paused for approval, cancelled, or timed out, no gated side effect
executes; and the human's decision---approve or reject---governs what actually
happens to the world. If that assumption holds, an operator who clicks
``reject'' can be confident the action did not occur. We are precise about the
assumption's provenance and about what a violation of it means. It is an
\emph{implied} contract---implied by the primitives' names, by framework
documentation, and by practitioner guidance (Section~\ref{sec:background}
grounds it in vendor wording and a maintainer-filed issue), not promised as a
formal specification---and our claim throughout is a \emph{contract mismatch}:
the semantics operators are led to assume are stronger than the semantics the
primitives deliver. Whether any individual gap is a bug or a deliberate design
is for maintainers to say; the operator-facing outcome is the same either way,
and it is that outcome we measure and repair. Figure~\ref{fig:motivation}
previews the central failure and the repair.

\begin{figure}[t]
\centering
\begin{tikzpicture}[
  font=\scriptsize,
  box/.style={draw, rounded corners=1pt, align=center, inner sep=2.5pt, minimum height=6.5mm},
  lbl/.style={font=\scriptsize\bfseries, anchor=west},
  arr/.style={-{Stealth[length=1.6mm]}, semithick},
  node distance=2.5mm and 3mm]
% --- Unmediated row ---
\node[lbl] (l1) at (0,0) {Unmediated (measured, Secs.~3--4):};
\node[box, below=of l1.west, anchor=north west] (a1) {approval gate\\pauses its branch};
\node[box, right=of a1, fill=black!8] (a2) {\textbf{sibling effect}\\\textbf{executes} during pause};
\node[box, right=of a2] (a3) {operator rejects\\$\Rightarrow$ too late};
\draw[arr] (a1) -- (a2);
\draw[arr] (a2) -- (a3);
% --- Gated row ---
\node[lbl, below=9mm of a1.south west, anchor=west] (l2) {Through \sys{} (Secs.~5--6):};
\node[box, below=of l2.west, anchor=north west] (b1) {approval gate\\pauses its branch};
\node[box, right=of b1] (b2) {sibling effect\\submitted $\to$ \textbf{held}};
\node[box, right=of b2] (b3) {reject $\Rightarrow$ refused;\\\textbf{zero effects}};
\draw[arr] (b1) -- (b2);
\draw[arr] (b2) -- (b3);
\end{tikzpicture}
\caption{The sibling leak and its repair. Top: an approval pause suspends only
its own branch; a sibling's effect commits during the pause, so the rejection
cannot prevent it---reproduced on every evaluated framework that ships a
pre-execution gate (Table~\ref{tab:matrix}). Bottom: with every effect
admitted by an external gate, the sibling's submission is held during the
pause and a rejection yields zero effects.}
\label{fig:motivation}
\end{figure}

\textbf{This paper shows the barrier assumption does not hold in the
evaluated frameworks---six, spanning four execution models and two language
runtimes---and that the failures recur across independently designed
frameworks and execution models rather than being one implementation's bug.}
Using small,
\emph{model-free} probes---the primitives under test are properties of the
\emph{framework}, not of any LLM, so no model or API key is required---we
characterize what the shipping control primitives actually guarantee. Our
central finding is a \emph{sibling leak}: when an approval gate and a
side-effecting action are siblings within the same execution step, the gate
suspends its own branch but the sibling's effect executes anyway, \emph{while the
run is paused awaiting approval}. For effects that are irreversible or
already externally committed---an e-mail delivered, a payment captured---a
subsequent human rejection is powerless: the action already happened.
(Compensable effects trade this for compensation machinery; see the
transactional treatments in Section~\ref{sec:related}.) We observe this leak in five frameworks spanning four distinct execution
models---Pregel/BSP supersteps, an event bus, message-passing fan-out, and
parallel tool calls within a single model turn---and two language runtimes,
which rules out a single implementation bug. We claim the recurrence
itself---the failure is not one implementation's accident---and leave its
origin (shared design pattern versus shared async-substrate constraints)
explicitly unresolved.

We then confirm three further enforcement gaps on current releases:
\emph{replay double-execution} (an effect before a resume point executes twice),
\emph{cancellation orphaning} (a cancelled run's tool effect lands after the
caller observes cancellation, on the standard worker-thread path), and
\emph{timeout zombies} (an effect completes after a timeout is reported). On the
timeout axis the frameworks \emph{disagree}---one refuses at construction to
attach a native timeout to a non-cancellable tool, one blocks the caller past
its deadline while the effect lands anyway, others leak zombies under host-level
deadlines---which shows sound behavior is achievable and the divergence
avoidable, by design or by fix.

Finally, we argue that because the framework's own control flow is the very thing
that fails, a repair whose guarantee is independent of each framework's fix
must live \emph{outside} it---upstream patches are complementary, and two
frameworks' clean cells show sound native designs exist
(Section~\ref{sec:related}). We present \sys{}, an
environment-external effect gate: a small, language-independent arbiter---%
implemented in Rust and exposed over a line-delimited protocol---through which
every side effect must be admitted before it touches the world. Under
\emph{complete mediation}---every side-effecting path submits to the gate, an
integration contract we state explicitly rather than assume silently
(Section~\ref{sec:design})---\sys{} enforces
four properties, one addressed to each measured violation class, and in
end-to-end replay it blocks every violation class while still
releasing legitimate approved effects, demonstrated across all six
frameworks---five with a pre-execution gate, repaired on their full violated
set, and a sixth post-hoc-review framework on its cancellation and timeout
axes---spanning all four execution models and both language runtimes.

\noindent\textbf{Contribution boundaries.} The \emph{primary} contribution is
the measurement: a characterized, cross-framework enforcement gap in shipped
control primitives. The repair and its verification are supporting
contributions that show the measured gap is closable at one
framework-independent point; the exposure and corroboration studies bound its
practical reach. Because the paper combines
measurement, a repair, and verification, we state each claim's scope once,
up front, and hold to it throughout: the \emph{measurement} (C1) stands on
its own and is unconditional over the evaluated frameworks and releases; the
\emph{repair} (C3) is conditional on the complete-mediation contract, and
every ``closes'' in this paper means \emph{suppresses the mediated effect}
under that condition; the \emph{verification} covers a \emph{model} of the
admission core (Verus, TLA+/TLC, TLAPS), connected to the deployed Rust by
differential conformance testing---refinement evidence, not a mechanized
refinement proof (Section~\ref{sec:failure-model}).

\noindent\textbf{Contributions.} The paper makes four separable
contributions---\textbf{(C1)} measurement, \textbf{(C2)} qualitative
corroboration that the failures occur in the wild,
\textbf{(C3)} an external repair mechanism whose admission core is
model-verified and differentially tested against the deployed code, and \textbf{(C4)} its evaluation
under benign and adversarial input---each of which stands on its own
evidence and is tagged below.
\begin{itemize}[leftmargin=1.2em]
\item \textbf{(C1)} A model-free differential \emph{measurement} isolating a
\emph{sibling leak} in approval gates and reproducing it across five frameworks
(four independent designs plus a cross-language port), two language runtimes,
and at least four independent environments; plus confirmation on current
releases of replay, cancellation-orphan, and timeout-zombie gaps, including a
cross-framework timeout \emph{disagreement} that maps the design space
(Section~\ref{sec:measurement}); plus an executed durable-execution contrast
(Temporal) separating what re-architecture buys by construction (replay)
from what it does not (the pause barrier), with the gate composing at that
engine's activity boundary unchanged (Section~\ref{sec:temporal}).
\item \textbf{(C2)} A five-model \emph{exposure} measurement under an a-priori
protocol---native-API anchors for GPT-4o and Claude, native replications for
Gemini and DeepSeek, provider-direct for Llama---showing real models emit the
plan shape at non-negligible, task-dependent, model-disjoint rates
(Section~\ref{sec:exposure}); and a live \emph{end-to-end} measurement
(Experiment~A) driving real models through the real \textsf{FW-A} runtime that
turns emission into a measured leak ($P(\mathrm{leak}\mid\mathrm{emitted}){=}1.00$,
$P(\mathrm{leak})$ up to $0.44$, $0$ once mediated; Section~\ref{sec:expA}),
with a verified 13-incident public corpus across three trackers as an
occurrence lower bound (Section~\ref{sec:prevalence}).
\item \textbf{(C3)} \sys{}, an environment-external effect gate whose four
properties map onto the measured violations under a stated complete-mediation
contract, with end-to-end repair on all six frameworks (four independent
designs and the JavaScript port on their full violated set, the post-hoc-review
framework on its cancellation/timeout axes), all four execution models, both
runtimes, plus a completion A/B on a third-party benchmark ($\tau$-bench
retail) in which the gated agent completes real episodes with zero
fail-closed refusals of legitimate writes, and a single-node concurrent-load
evaluation including durable (WAL)
mode (Sections~\ref{sec:design}--\ref{sec:eval}).
\item \textbf{(C4)} A demonstration on the real \textsf{FW-A} runtime that the
leak \emph{composes} with prompt injection and the gate holds the same
invariant when the plan is adversarially induced---an intent-agnostic barrier,
not an injection defense (Section~\ref{sec:injection}).
\item \textbf{(C1--C4)} A reproducible artifact: all probes with per-framework
transcripts, the measurement harness, the concurrent-load benchmark, a static
mediation linter, the formal models with checker logs, and the \sys{} reference
implementation, runnable without API keys; the gate also ships as an installable
Python package (\texttt{pip install soundgate}) with native (PyO3) and
pure-Python clients, so the repair is adoptable without a framework rewrite.
\end{itemize}

% =====================================================================
\section{Background and Threat Model}
\label{sec:background}

\subsection{Control primitives in agent frameworks}
Contemporary frameworks expose HITL approval through a pause/resume mechanism:
a node or step requests human input, the runtime suspends the run and persists
its state, and execution resumes when a decision arrives. Cancellation and
timeouts are exposed either as native runtime features or via the host
language's concurrency facilities. Across frameworks these primitives share a
common promise---\emph{pause here, and nothing sensitive proceeds until the human
(or the deadline) decides}.

\begin{definition}[Barrier contract, operator sense]
\label{def:barrier}
We use ``barrier semantics'' throughout in the operator's sense---a stop is a
stop for the run's gated effects---made precise as four per-primitive safety
clauses: \textbf{B1} (approval) no gated effect of the run executes between the
pause and the decision; \textbf{B2} (rejection) a rejected effect never
executes; \textbf{B3} (resume) each logical effect executes at most once;
\textbf{B4} (cancel/timeout) after the caller observes the stop, no further
effect of the run lands. This is a \emph{fence} at the effect boundary, not a
synchronization barrier in the parallel-computing sense
(all-arrive-before-any-proceed). B1--B4 are exactly the violation predicates of
Section~\ref{sec:measurement} and the enforcement properties P1--P4 of
Section~\ref{sec:design}; they are safety clauses only---\emph{liveness} (every
held effect is eventually decided) is deliberately not part of the contract
(Section~\ref{sec:design}). The clause set is exactly one per shipped stop
primitive (approval, rejection, resume, cancellation/timeout), not an open
axiom list: candidate additions such as rollback guarantees, notification
ordering, or read-set freshness govern what happens \emph{after} a decision
or outside the stop window, and are named as non-goals in
Sections~\ref{sec:design} and~\ref{sec:limitations} rather than folded into
the contract.
\end{definition}

\noindent\emph{What ``executes'' means.} An effect \emph{executes} when its
externally visible action is performed---the commit point past which the tool
cannot unilaterally revoke it (the e-mail handed off for delivery, the request
that captures the payment). Submitting a request to a gate, or computation
preceding the commit point, is not execution. B1--B4 and the violation
predicates of Section~\ref{sec:measurement} are evaluated at the commit
point, which the probes instrument directly---each tool's single effect-log
append sits exactly there---so every predicate is a falsifiable single bit
per executed trace.

Before measuring, we make the implied contract explicit, and we do not rest
it on intuition alone: the expectation is the vendors' own framing.
\textsf{FW-A}'s documentation describes the primitive as letting you ``pause
graph execution at specific points and wait for external input''---graph
execution, not one branch of it---with the canonical example gating exactly
the irreversible actions we study (tool calls, e-mail sends) behind the
pause~\cite{lgdocs26}; practitioner guidance built on the same surface
presents the pause as the mechanism that prevents the action from firing
until a human decides~\cite{abstractalg26}; and the corpus
(Section~\ref{sec:prevalence}) shows a \emph{maintainer-filed} issue
tracking the missing multi-interrupt barrier as a defect---maintainer-side
evidence that the barrier reading is the intended one, not a contract we
invented to violate.
Table~\ref{tab:audit} contrasts the barrier semantics an integrator reasonably
assumes with the measured reality; only one framework's documentation states the
non-barrier resume behavior~\cite{lgdocs26}, and we found no framework
documentation stating that sibling effects proceed during an approval pause.

\begin{table}[!t]
\centering
\caption{Expectation audit: the implied barrier contract versus measured
behavior (details: Table~\ref{tab:matrix} and Section~\ref{sec:occurrence}).}
\label{tab:audit}
\renewcommand{\arraystretch}{1.2}
\footnotesize
\begin{tabularx}{\columnwidth}{@{}p{1.55cm}XX@{}}
\toprule
\textbf{Primitive} & \textbf{Documentation-implied} & \textbf{Measured}\\
\midrule
Approval gate & Pausing halts the run's effects until a decision & Sibling
effects execute during the pause in every evaluated framework that ships a
pre-execution gate and can express parallelism\\
Rejection & A rejected action does not happen & The action may already have
completed before the decision point\\
Resume & Completed work is not redone & Pre-gate effects re-execute
(1$\to$2) in both runtimes of one design; documented for that framework
only~\cite{lgdocs26}; three frameworks retain completed work across resume
instead\\
Cancel & In-flight work stops & Thread-backed tools orphan wherever
constructible; on Node even pure-async orphans on the native surface\\
Timeout & After the deadline the action will not happen & Zombie effects
(host-level deadlines), a blocking overrun where the effect lands anyway, and
one refusal-by-construction\\
\bottomrule
\end{tabularx}
\end{table}

\subsection{Why the primitives may leak: parallelism and re-entry}
Two structural features of modern frameworks put pressure on the barrier
assumption. First, many frameworks execute independent branches
\emph{concurrently} within a single logical step; an approval gate that suspends
one branch does not necessarily suspend its siblings. Second, several frameworks
achieve durability by \emph{re-executing} a node or step from its start on
resume; any side effect performed before the suspension point can then run more
than once. Both features are desirable for throughput and reliability, but
neither, by itself, preserves stop semantics for side effects.

\subsection{Threat model and scope}
\label{sec:scope}
We consider a benign but realistic operator who uses a framework's documented
control primitives to gate irreversible or externally committed actions;
compensable effects admit transactional treatments
(Section~\ref{sec:related}) and are not our focus. We do \emph{not} defend against prompt
injection~\cite{greshake23} or a compromised tool; those are orthogonal, well
studied, and out of scope. We do, however, show in
Section~\ref{sec:injection} that \emph{when} injected content induces the
leak plan shape, the barrier holds the adversary's effect exactly as it holds
a benign one---an intent-agnostic guarantee, not injection detection. The operator is human: decisions may take seconds to
minutes, and automation bias is real~\cite{parasuraman97,cummings04}; B1--B4 are
therefore stated so that the \emph{duration} of a pause is irrelevant---nothing
gated may land during it, however long it lasts. Decision authenticity
(that a recorded approval really came from the operator) is a deployment
obligation, not part of this model (Section~\ref{sec:failure-model}).
Our question is narrower and prior:
\emph{do the stop primitives themselves provide the guarantee their names
imply?} across configurations of released framework versions. We do
not claim these behaviors are undocumented in every case---some are noted in
issue trackers or engineering blogs---but that their \emph{scope, cross-framework
recurrence, and repair} have not been characterized. All of
Section~\ref{sec:measurement}'s probes are model-free by construction; the
phenomena they measure are properties of framework control flow, and
introducing an LLM would only add nondeterminism to an otherwise
deterministic measurement. The exposure study (Section~\ref{sec:exposure})
necessarily queries real models, and the injection demonstration
(Section~\ref{sec:injection}) is scripted by default with an optional
live-model mode; neither is part of the model-free measurement, and we label
each accordingly.

\smallskip\noindent\textbf{Non-goals, stated once.} So the repair's scope
cannot be over-read, we enumerate here what \sys{} does \emph{not} provide,
each treated in full where cited: injection \emph{detection}
(Section~\ref{sec:injection} demonstrates composition with injected input,
not a defense against it); confinement of a \emph{malicious} tool's
non-network channels---shared filesystem, local IPC, shared memory---beyond
the placement contract or an analogous seccomp/LSM policy
(Section~\ref{sec:design}); atomicity or compensation \emph{across} phases of
a multi-phase API (Section~\ref{sec:limitations}); the quality or timeliness of the human decision itself---whether an operator can decide correctly in the available window, and whether an approval step degrades vigilance, are the classical human-factors concerns of situation awareness and supervisory control~\cite{parasuraman97,bainbridge83,endsley95,cummings04}, which the barrier makes \emph{meaningful} (a rejection now provably prevents the effect) but cannot answer; and key-distribution infrastructure for
the decision secret, which is ordinary secret management on an authenticated
channel (Section~\ref{sec:failure-model}).

% =====================================================================
\section{Measurement}
\label{sec:measurement}

\subsection{Method}
Each probe instantiates a minimal workflow in the target framework in which
the control primitive under test is the \emph{only} varying factor and no
nondeterministic input exists.
Side effects are represented by appends to an in-process event log, so that we
can observe precisely \emph{when}, and \emph{how many times}, an effect occurs
relative to a pause, a rejection, a cancellation, or a timeout. For each probe we
fix a \emph{violation predicate} before running (pre-registration), so the
outcome is a single bit: does the primitive provide its implied guarantee. The
probes use plain Python functions as nodes/steps; no model is invoked. The
probes are thus \emph{minimal witness programs}---each the smallest workflow
exhibiting one predicate---built for internal validity rather than ecological
realism: this section establishes \emph{existence and mechanism};
\emph{frequency} is deliberately a separate question, answered by
Section~\ref{sec:occurrence}. The measurement is descriptive throughout: we
classify whether a primitive satisfies B1--B4 as shipped, not whether any
implementation is ``incorrect.''

We report results for six frameworks---five that ship a pre-execution
approval primitive plus \textsf{FW-E}, which ships only post-hoc review and
is therefore counted as ``no primitive to violate'' throughout, not as a
sixth violation (footnote~\ref{fn:fwe-posthoc}). The five are a Pregel/BSP-style graph runtime
(hereafter \textsf{FW-A}); an event-driven workflow runtime (\textsf{FW-B}); a
message-passing multi-agent runtime with fan-out edges (\textsf{FW-C}); an
agent SDK whose model turns may contain parallel tool calls (\textsf{FW-D}); a
role/task crew orchestrator (\textsf{FW-E}); and the JavaScript port of
\textsf{FW-A}'s runtime on Node (\textsf{FW-F}). Identities and pinned
versions: \textsf{FW-A}${=}$LangGraph (Python) 1.2.7;
\textsf{FW-B}${=}$LlamaIndex Workflows (\texttt{llama-index-core} 0.14.23);
\textsf{FW-C}${=}$Microsoft Agent Framework (\texttt{agent-framework-core}
1.10.0); \textsf{FW-D}${=}$OpenAI Agents SDK 0.17.7;
\textsf{FW-E}${=}$CrewAI 1.15.1\footnote{\label{fn:fwe-posthoc}\textsf{FW-E}
ships no pre-execution approval primitive: its human-input hook requests
feedback \emph{after} a tool has run, so it cannot exhibit a sibling
\emph{leak}---there is no gate to leak past. We exclude it from every
recurrence count and mark its cells ``not comparable'' ($\ominus$) in
Table~\ref{tab:matrix}, but retain the row: a framework in wide use offering
only post-hoc review is itself evidence that pre-execution barriers are not a
solved, ubiquitous feature. Calling this a ``violation'' would be a category
error. \textsf{FW-E} thus functions as a \emph{design contrast} rather than a
count-padding sixth violation: our headline recurrence numbers are stated
over the four pre-execution-gate frameworks, and \textsf{FW-E}'s value is
showing that a widely used framework made a different (post-hoc) design
choice---and even so exhibits the cancellation and timeout violations that
are independent of the approval axis. The four frameworks that \emph{do} provide pre-execution gates
(\textsf{FW-A}--\textsf{D}, plus \textsf{FW-F}) are the basis for every
recurrence claim.}; \textsf{FW-F}${=}$LangGraph.js
(\texttt{@langchain/langgraph} 1.4.7, Node 22/23). \textsf{FW-C} is the
direct successor to AutoGen and Semantic Kernel (GA April 2026; AutoGen is
in maintenance)~\cite{msaf26}, so the AutoGen lineage is measured here in
its current generation rather than its deprecated predecessor. The inclusion criteria: an
open-source runtime instrumentable locally, shipping at least one of the
audited primitives, chosen to cover the four execution-model families
production agent runtimes are built on (Pregel/BSP, event bus,
message-passing fan-out, parallel tool calls) plus one cross-language port of
the most-used design---coverage by architecture, not a popularity ranking
(Section~\ref{sec:threats-validity}). Hosted,
closed agent platforms (cloud-vendor agent services) expose no local runtime
to instrument and are outside this open-runtime scope. We keep the neutral
\textsf{FW-$x$} labels in the running text only so the emphasis stays on the
shared \emph{pattern}. \textsf{FW-A} is additionally
checked across two adjacent releases for version stability, and every verdict
in Table~\ref{tab:matrix} reproduced identically in at least three independent
environments---four for \textsf{FW-A}--\textsf{D} and \textsf{F}, whose full
suites were re-executed verdict-identically on a fresh single-vCPU container
during revision (artifact, \texttt{evidence/}). Exact commands are in the
artifact.

\subsection{Probes and violation predicates}
\begin{itemize}[leftmargin=1.2em]
\item \textbf{Sibling leak (approval coverage under parallelism).} An approval
gate and a side-effecting action are siblings in one step. \emph{Violation:} the
effect executes while the run is paused awaiting approval.
\item \textbf{Reject-after-effect.} Continuing the above, the human rejects.
\emph{Violation:} the sibling effect has already occurred and rejection cannot
prevent it.
\item \textbf{Replay double-execution.} A node performs an effect \emph{before}
requesting approval; the human then approves. \emph{Violation:} the effect
executes twice (re-execution from the step start on resume).
\item \textbf{Cancellation orphan.} An asynchronous run is cancelled while a node
executes a blocking tool on a worker thread. \emph{Violation:} the caller
observes cancellation, yet the effect lands afterward.
\item \textbf{Timeout zombie.} A deadline fires while a tool is in flight.
\emph{Violation:} a timeout is reported to the caller, yet the effect lands
afterward. Because not every framework ships a native run timeout, this axis
compares \emph{deadline-enforcement behavior}, not identical APIs: where no
native parameter exists we probe the documented host-level deadline pattern
and label the cell as such (Table~\ref{tab:matrix}, note~a).
\end{itemize}

\noindent\emph{Severity ordering.} The axes are not equally severe, and we
state the ranking we use rather than imply uniformity: the sibling leak and
its entailed reject-after-effect defeat the primitive's core purpose on
irreversible actions and rank highest; replay double-execution is a
correctness (often financial) hazard bounded by whether the tool is
idempotent; cancellation orphans and timeout zombies are
correctness-and-availability hazards whose blast radius is one in-flight
effect. The repair treats all four uniformly only because the gate's
admission cost is identical per class; the measurement's headline is the
approval axis.

\subsection{Results}
Table~\ref{tab:matrix} summarizes the executed outcomes. The sibling leak and
reject-after-effect violations reproduce in \emph{every} framework that can
express the configuration, despite four dissimilar execution models, and
\textsf{FW-A}'s behavior is identical across two releases. Stated both ways
so neither framing inflates: all four gate-shipping framework \emph{designs}
leak (five implementations, counting the cross-language port), and
\textsf{FW-E}, which ships no pre-execution gate, is outside the sibling
denominator entirely rather than a sixth data point. Reject-after-effect is
listed as the operator-facing consequence \emph{entailed by} the sibling
leak---once the sibling has executed during the pause, rejection is
necessarily powerless---not as an independent mechanism; the matrix rows
count observed predicate outcomes, not distinct root causes. The
reject-after-effect window is structural rather than a race: in every
violating trace the sibling's effect completes within the same execution step
that raised the pause, before control returns to any caller that could issue a
resume---no human reaction time, however fast, closes it. (This is a
mechanism-level reading of the executed traces, uniform across the violating
frameworks, not a formal impossibility proof.) For \textsf{FW-D},
whose steps are model turns, the probe scripts the turn deterministically, so
the verdict concerns the framework's \emph{execution} of a given parallel
plan---whichever party proposed it---not any model's propensity to propose one
(that propensity is Section~\ref{sec:exposure}'s subject). Replay is a design
property rather than an implementation accident: the same double-execution
reproduces in both language runtimes of the same design (\textsf{FW-A},
\textsf{FW-F}), while two independently designed frameworks (\textsf{FW-C},
\textsf{FW-D}) are clean because they cache completed work in the resume
token. Cancellation soundness is host-language-dependent: identical pure-async
node logic cancels cleanly under Python's asyncio yet orphans its effect under
Node's native \textsf{AbortSignal} surface, because promises cannot be
interrupted. On the timeout axis the frameworks \emph{differ} in four distinct
ways (Table~\ref{tab:matrix}, notes a--c). These disagreements are important:
they show the behaviors are not inherent to agent frameworks and can be
avoided; whether each gap reflects a deliberate design choice or an unsolved
problem is for maintainers to clarify---the claim remains the contract
mismatch of Section~\ref{sec:intro}, and only one framework's documentation
states the divergence~\cite{lgdocs26}. The contract
we hold frameworks to is not a claim that async branches ought to be
mutually blocking in general---they should not be---but the far narrower
operator expectation that \emph{attaching a human-approval gate to an
irreversible action causes that action to wait for the human}. An operator
who writes ``require approval before issuing a refund'' is not reasoning
about event-loop scheduling; they are asserting that the refund does not
happen until they say so. That expectation is what the sibling leak
violates, and it is orthogonal to whether unrelated branches run
concurrently---which, as the repair shows (Section~\ref{sec:design}), they
still do. Because the frameworks share async
ecosystems (\texttt{asyncio} event loops, JavaScript promises, worker
threads), the recurrence is consistent with both design-level causes and
shared implementation constraints; our claim is the weaker, well-supported
one---the barrier does not hold as measured, across dissimilar execution
models.

\begin{table}[!t]
\centering
\caption{Executed control-primitive outcomes across six frameworks (five
Python, one JavaScript). Identities:
\textsf{A}${=}$LangGraph, \textsf{B}${=}$LlamaIndex Workflows,
\textsf{C}${=}$Microsoft Agent Framework, \textsf{D}${=}$OpenAI Agents SDK,
\textsf{E}${=}$CrewAI, \textsf{F}${=}$LangGraph.js (versions in
Section~\ref{sec:measurement}).
\checkmark${=}$violation reproduced;
\ding{55}${=}$clean/contrast; R${=}$unsound configuration refused at
construction; $\ominus$${=}$not comparable (primitive absent by design,
footnote~\ref{fn:fwe-posthoc}); ``--''${=}$not applicable or not probed
(per-cell distinction in
the artifact's matrix). Every verdict reproduced identically in $\geq$3
independent environments ($\geq$4 for \textsf{A}--\textsf{D},\textsf{F});
\textsf{FW-A} additionally across two releases.}
\label{tab:matrix}
\renewcommand{\arraystretch}{1.2}
\footnotesize
\begin{tabularx}{\columnwidth}{@{}Xcccccc@{}}
\toprule
\textbf{Axis} & \textbf{A} & \textbf{B} & \textbf{C} & \textbf{D} & \textbf{E} & \textbf{F}\\
\midrule
Sibling approval leak       & \checkmark & \checkmark & \checkmark & \checkmark & $\ominus^{e}$    & \checkmark \\
Reject-after-effect         & \checkmark & \checkmark & \checkmark & \checkmark & $\ominus^{e}$ & \checkmark \\
Replay double-execution     & \checkmark & \ding{55}$^{g}$ & \ding{55}$^{d}$ & \ding{55} & --          & \checkmark \\
Cancel: worker/async        & \checkmark & \checkmark$^{g}$ & \checkmark & \checkmark & \checkmark & --$^{f}$   \\
Cancel: pure async          & \ding{55}  & --         & \ding{55}  & \ding{55}  & --          & \checkmark$^{f}$ \\
Timeout zombie              & \checkmark$^{a}$ & \ding{55} & \checkmark$^{a}$ & \ding{55}/R$^{b}$ & \ding{55}$^{c}$ & \checkmark \\
Timeout blocks, effect lands & --        & --         & --         & --         & \checkmark$^{c}$ & --        \\
\bottomrule
\end{tabularx}

\smallskip
\raggedright\footnotesize
$^{a}$No native run-timeout parameter exists in the pinned release; probed via
a host-level deadline and labeled as such.
$^{b}$Pure-async tools cancel cleanly; attaching the native per-tool timeout to
a synchronous tool is refused at construction (sound by refusal).
$^{c}$Strict zombie predicate is clean, but the timeout error surfaces only
\emph{after} the timed-out work completes and its effect lands, blocking the
caller past the deadline (distinct class, last row).
$^{d}$Also clean across a fresh-process checkpoint restore.
$^{e}$\textsf{FW-E} exposes no pre-execution approval primitive; its human-input
hook is by-design post-hoc review (the effect precedes the review), so neither
cell states a barrier to violate. Marked $\ominus$ and excluded from every
recurrence count (footnote~\ref{fn:fwe-posthoc}).
$^{f}$No sync/async split exists in JavaScript; the single documented
cancellation surface (\textsf{AbortSignal}) acknowledges the abort while the
effect lands afterward.
$^{g}$\textsf{FW-B}: the native \texttt{cancel\_run} raises inside the
workflow but does not cover a worker thread, whose effect lands after
cancellation (violation); resume via the documented
\texttt{Context.from\_dict} path does not re-execute the completed step
(clean)---a design contrast with \textsf{FW-A}/\textsf{F}.
\end{table}

\noindent\textbf{Interpretation.} The contrast cells are as informative as the
violations. Cancellation is \emph{clean} when the tool is a pure-async coroutine
under Python's asyncio but \emph{orphans} when the identical logical tool runs
on a worker thread---and orphans even for pure-async code under Node, where the
framework's own \textsf{AbortSignal} surface cannot interrupt a promise chain.
Cancellation soundness thus depends invisibly on a tool's execution mode
\emph{and} on the host language. Replay is clean exactly in the three
frameworks that retain completed work across resume---two cache turn results
in the resume token, one restores step progress from its serialized
context---and double-executes, unsafely for irreversible effects, in both runtimes of the design that re-executes from
the checkpoint. \textsf{FW-D} shows that
refusing to construct an unsound configuration is a viable design point.
Enforcement that depends on such incidental factors is exactly what an external
gate removes.

\noindent\textbf{Randomized structural sweep.} To test whether the sibling
leak is an artifact of authored probes, we swept 1{,}000 seeded random
workflows through the real \textsf{FW-A} runtime (artifact,
\texttt{randgraph/}): out-trees of 3--8 nodes (in-degree ${\leq}1$, so join
semantics are not a confound), exactly one approval gate placed uniformly at
random, one to three effect nodes, the remainder reads, each effect
pre-classified by its structural relation to the gate. The outcome is deterministic
across all 1{,}000 generated workflows (Table~\ref{tab:randgraph}): every effect concurrent with the
gate's superstep executed during the pause (577/577); every gate-descendant
effect was withheld until the decision (0/363); and concurrent effects
scheduled in \emph{later} supersteps never executed during the pause
(0/331)---the interrupt halts the scheduling loop after the superstep that
raises it, so the leak window is exactly the superstep that raises the pause,
sharpening the mechanism-level reading above into a measured boundary.
Effects on the gate's own path or in earlier supersteps ran strictly before
the gate node was entered (median ${\sim}1$\,ms prior) and are outside the
B1 window by construction. Two incidental confirmations at scale: the gate
node body re-executed on resume in all 1{,}000 graphs (the documented
resume-from-node-start behavior~\cite{lgdocs26}), and each completed effect
executed exactly once. The sweep targets the runtime whose topology is
user-specified; compiling arbitrary graphs onto an event bus or a model
turn's tool batch would interpose a dispatcher of our own construction, so
the other execution models remain covered by the minimal witnesses above.

\begin{table}[!t]
\centering
\caption{Randomized structural sweep: leak rate by an effect's relation to
the gate, over 1{,}000 seeded random workflows executed on \textsf{FW-A}
(Wilson 95\% intervals). The leak is exactly the schedulability predicate of
the pausing superstep, independent of topology.}
\label{tab:randgraph}
\renewcommand{\arraystretch}{1.15}
\footnotesize
\begin{tabular}{@{}lccc@{}}
\toprule
\textbf{Relation to gate} & \textbf{Leak / $n$} & \textbf{Rate} & \textbf{95\% CI}\\
\midrule
Concurrent, same superstep & 577/577 & 1.00 & [0.99, 1.00]\\
Concurrent, later superstep & 0/331 & 0.00 & [0.00, 0.01]\\
Gate-descendant & 0/363 & 0.00 & [0.00, 0.01]\\
\bottomrule
\end{tabular}
\end{table}

\subsection{Contrast: a durable-execution engine}
\label{sec:temporal}
Because the natural response to Table~\ref{tab:matrix} is ``adopt a
durable-execution engine,'' we execute the same four predicates on the engine
most often named: Temporal~\cite{temporal} (server 1.31.2 via CLI 1.7.3,
Python SDK \texttt{temporalio} 1.30.0, local dev server; every activity
pinned to \texttt{RetryPolicy(maximum\_attempts=1)} so no verdict can be
manufactured by platform retries; verdicts identical across repeated runs;
artifact \texttt{probes-temporal/}, receipts
\texttt{evidence/temporal\_probes.txt}). Probe battery and gate composition reproduced verdict-identically across two environments spanning server versions 1.31.1--1.31.2, the analog of the two-release stability check applied to \textsf{FW-A}.
Temporal is a \emph{contrast arm} in the \textsf{FW-E} sense---excluded from
every recurrence denominator---and one verdict must be framed before it is
stated: Temporal's documentation nowhere implies that awaiting a Signal in
one branch pauses its siblings, so the sibling bit below is \emph{behavioral
only}; no contract-mismatch claim attaches to it. The results split exactly
where the trichotomy (Proposition~\ref{prop:trichotomy}) predicts.
\emph{Replay is closed by construction}: an activity's effect committed
before the approval wait commits exactly once ($1{\to}1$) across a forced
full-history replay---worker shut down with the workflow-task cache
disabled, decision delivered, fresh worker replays the history---because
journaled activities return recorded results rather than re-executing, the
re-architecture benefit Temporal exists to provide. \emph{The pause barrier
is absent there too}: with an approval wait (the documented Signal-based
HITL pattern) and a side-effecting activity as siblings under
\texttt{asyncio.gather}, the sibling's effect commits while the workflow's
queried state is \texttt{awaiting\_decision}, before any decision is
delivered, and a subsequent rejection is powerless---case~(ii) of the
trichotomy as substrate behavior rather than broken promise.
\emph{Cancellation is cooperative, as documented}~\cite{temporal-docs}: a
blocking activity on the worker thread pool with no heartbeat---heartbeats
being the only cancellation-delivery channel Temporal documents---commits
its effect after the caller has observed \texttt{CancelledError} (the
orphan), while the identical logical activity, heartbeating, observes the
cancel and commits nothing (clean). \emph{Timeouts zombie, as
documented}~\cite{temporal-docs}: an activity outliving its
\texttt{start\_to\_close} deadline commits after the caller observed the
timeout, and the SDK's own runtime log records the completion of the
already-timed-out activity. Finally, the repair travels: the identical
${\sim}$20-line wrapper (\texttt{pip install soundgate}) at Temporal's
activity boundary holds the sibling's mediated effect during a live Signal
pause (one pending identity, zero effects), the operator's rejection is
sticky, a zombie resubmission meets \texttt{refused\_rejected}, and a
legitimate mediated effect on a fresh identity releases and executes
(receipt \texttt{evidence/temporal\_gated.txt}; hold/reject path---the
approve choreography and full verdict set are Table~\ref{tab:e2e},
unchanged by the host engine). The reading for the re-platforming question
of Section~\ref{sec:related} is thereby executed rather than argued:
re-architecture buys exactly one measured axis by construction, the pause
barrier must be an admission point in any engine, and the same
environment-external gate supplies it on this one unchanged.

% =====================================================================
\section{Does the Hazard Occur in Practice?}
\label{sec:occurrence}
The matrix establishes what frameworks \emph{permit}. This section asks
whether the permitted hazard is \emph{reachable} in practice, and answers with
two independent measurements: real models emit the triggering plan shape at
non-negligible rates---and, driven end-to-end, execute the leak on live
runtimes---and real users hit the replay and cancellation failures and file
them.

\subsection{Model exposure}
\label{sec:exposure}
The sibling leak requires a specific plan shape: the consequential
(approval-gated in deployment) tool call sharing an assistant turn with at
least one benign sibling call. We measured how often five models emit that
shape under a protocol fixed a priori: ten authored tasks (five
single-outcome, five compound ``look up X and then do Y''; wording rules
fixed a priori---no concurrency vocabulary, no tool names in prompts), each
pairing one consequential tool with two benign read-only tools; temperature
1.0; $N{=}100$ runs per task per model (up from the $N{=}25$ pilot), sized
for \emph{estimation} rather than hypothesis testing---the power arithmetic
closes this section; the run stops at the first turn containing the
consequential call, which is \emph{never executed}. All five
models---GPT-4o, Claude Sonnet~4.6, Gemini~2.5~Flash, DeepSeek~V3.2, and
Llama~3.3-70B---were queried through a single OpenRouter integration with
\texttt{require\_parameters} routing, so a request is sent only to a backend
that honors the tool schema, never silently degraded; the $N{=}25$ pilot
queried GPT-4o and Claude on their \emph{native} APIs, giving those two
models a second, independent serving path against which the $N{=}100$
results are cross-checked below. Retry, per-$(\mathit{task},\mathit{run})$
deduplication, and decoding-seed handling are stated in
Appendix~\ref{app:exposure}. The metric, exclusion criteria, and task
wording are unchanged from the pilot and are committed to the artifact; we
make no external-registry (``pre-registered'') claim.

\begin{table}[!t]
\centering
\caption{E-EXPOSURE: parallel gated$+$ungated emission, $N{=}100$/task, five
models. \dag${=}$near-zero rate through the shared OpenRouter integration is
\emph{pathway-confounded} (a positive control shows these models never
bundle even with zero hazard on the table) and is not interpretable as a
model disposition. Indented rows are second-serving-path replications:
native APIs for Gemini (which \emph{does} emit the shape on Google's own
surface) and DeepSeek (near-zero reproduces), provider-direct (Together) for
Llama, whose $0.90$ on \texttt{compound\_cleanup} is the highest single-task
rate in the study; its Called column excludes 117/1000 schema-nonadherent
runs (analysis in text).}
\label{tab:exposure}
\renewcommand{\arraystretch}{1.15}
\setlength{\tabcolsep}{3pt}
\scriptsize
\begin{tabular}{@{}lcccc@{}}
\toprule
\textbf{Model} & \textbf{Called} & \textbf{Exposure [95\% CI]} & \textbf{Sgl./Cmp.} & \textbf{Worst task}\\
\midrule
GPT-4o & 1000/1000 & 0.14 [0.12, 0.17] & 0.13 / 0.16 & 0.75 (cleanup)\\
Claude Sonnet 4.6 & 980/1000 & 0.04 [0.03, 0.05] & 0.00 / 0.07 & 0.34 (transfer)\\
Gemini 2.5 Flash\dag & 942/1000 & 0.00 [0.00, 0.00] & 0.00 / 0.00 & ---\\
\quad\emph{3.5 (native)} & 1000/1000 & 0.02 [0.01, 0.03] & 0.00 / 0.03 & 0.17 (invoice)\\
DeepSeek V3.2\dag & 987/1000 & 0.00 [0.00, 0.01] & 0.00 / 0.00 & 0.01 (refund)\\
\quad\emph{V4-flash (native)} & 975/1000 & 0.00 [0.00, 0.00] & 0.00 / 0.00 & ---\\
Llama 3.3 70B\dag & 948/1000 & 0.00 [0.00, 0.01] & 0.00 / 0.01 & 0.02 (transfer)\\
\quad\emph{3.3 70B (Together)} & 769/883 & 0.08 [0.07, 0.10] & 0.00 / 0.23 & 0.90 (cleanup)\\
\bottomrule
\end{tabular}
\end{table}

Table~\ref{tab:exposure} carries two distinct findings, and conflating them
would overclaim. \emph{First}, for the two models with a native-API anchor,
the $N{=}100$ OpenRouter measurements replicate the pilot across a change of
serving path \emph{and} a $4\times$ sample-size change: GPT-4o's pooled rate
moves $0.15{\to}0.14$ and its worst task (\texttt{compound\_cleanup})
$0.84{\to}0.75$; Claude's pooled rate moves $0.03{\to}0.04$ and its worst
task (\texttt{compound\_transfer}) $0.24{\to}0.34$; every interval overlaps
its pilot counterpart, so we read these rates as properties of the models.
The triggering tasks are disjoint (GPT-4o leaks on
\texttt{compound\_cleanup}, never on \texttt{compound\_transfer}; Claude the
reverse), and because the consequential tool is called in 94--100\% of runs
across all five models, conditioning is not concealing a small base rate.
The ``worst task'' column is a post-hoc descriptive statistic over the ten
pre-fixed tasks, not a primary endpoint; a multiplicity analysis over the
most conservative family the study admits ($m{=}80$ cells,
Bonferroni-adjusted Wilson intervals) reclassifies no nonzero finding
(Appendix~\ref{app:exposure}).
\emph{Second}, Gemini, DeepSeek, and Llama show a near-total absence of the
shape through OpenRouter (0--0.3\% of called instances)---but a
positive-control diagnostic shows this is not evidence of caution: with the
consequential tool stripped so zero hazard is on the table, GPT-4o still
bundles the two benign reads on $12/15$ trials, while the three models
bundle on $0/90$ combined trials. Their OpenRouter rates are therefore
\emph{pathway-confounded and not interpretable as model
dispositions}---neither as exposure nor as safety. A second serving path
then resolves each case (full forensics in Appendix~\ref{app:exposure}).
\textsf{gemini-3.5-flash} on Google's native surface, same $N{=}100$ over
the same ten tasks, called the gated tool in $1000/1000$ runs and emitted
the shape in $17/1000$ (pooled $0.02$, concentrated on
\texttt{compound\_invoice} at $17/100$): the family \emph{does} reach the
shape on a faithful path, so the hazard is not GPT-4o-specific.
\textsf{deepseek-v4-flash} on DeepSeek's own API called the gated tool in
$975/1000$ runs and emitted the shape in $0/975$---a genuinely sequential
disposition the vendor's own surface reproduces, the model-disjointness the
headline claims. \textsf{Llama-3.3-70B} provider-direct (Together) emits in
$64/769$ classifiable runs (pooled $0.08$; $0$ on single-outcome, $0.23$ on
compound, $0.90$ on \texttt{compound\_cleanup}---\emph{higher} than
GPT-4o's $0.75$): the OpenRouter near-zero was a pathway artifact concealing
the most exposed model in the study on the leak-driving task. One honesty
note travels with that row: $117/1000$ Llama runs are excluded as
unclassifiable because the model invents tool names on the hardest compound
tasks---a schema-adherence weakness, not a serving-path effect, whose
mediation implication (unknown-tool handlers must be wrapped or refused like
any consequential tool) is drawn in Appendix~\ref{app:exposure}; the
exclusions concentrate away from the signal task ($57$ of $64$ exposures sit
on \texttt{compound\_cleanup} at $99/100$ valid), so the de-confound does
not rest on error-thinned tasks. Family versus checkpoint, stated once so
the de-confound is not over-read: only GPT-4o and Claude carry the
\emph{same} checkpoint on two independent serving paths; the Gemini and
DeepSeek second paths are successor checkpoints on the vendors' own surfaces
(family-level de-confounds), and the Llama second path replicates the same
checkpoint through a different provider (Appendix~\ref{app:exposure} states
the boundary in full). Third, model-level caution, where validated, is not a
stop primitive: on one cancellation task, Claude completed both lookups and
then replied without calling the consequential tool in 20 of 100
runs---useful behavior, but a disposition, not an enforced guarantee.
Scope: these are single-integration measurements for all five models,
anchored to native-API pilots for GPT-4o and Claude, with full second-path
replications for the other three, over ten authored tasks with frictionless
canned tool results. Wilson intervals are tighter than the pilot's (GPT-4o's
worst task narrows from roughly $[0.65,0.94]$ to $[0.66,0.82]$), and we
treat the study as estimation, not hypothesis confirmation. ($N{=}100$ per
task is sized for that purpose: a true per-task rate of 3\% evades detection
entirely with probability $0.97^{100}{\approx}4.8\%$, so per-task zeros are
reported with their intervals rather than as absences, while a pooled zero
over ${\sim}1000$ called runs bounds the underlying rate below ${\sim}0.4\%$
at 95\% confidence---the DeepSeek-native row's $[0.00,0.00]$ is that bound
at two-decimal display, not a claim of impossibility.) Emission is measured
on the raw, ungated API: it is the plan shape reaching the framework, the
leak's \emph{precondition}, not its end-to-end frequency. The full causal
chain---real model, real runtime, effect landing during the pause---is
\emph{executed}, not inferred, in Sections~\ref{sec:expA}
and~\ref{sec:injection}; what remains unmeasured is the chain's frequency in
production deployments, which no component of this study claims. Longer
conversations, richer toolsets, naturalistic (non-authored) tasks, and
adversarial prompting rates remain out of scope, as does strategic
adaptation: whether a model that \emph{observes} gating shifts toward or
away from parallel emission is a question this ungated-API design cannot
see, and we flag it as open rather than assume either direction.

\smallskip\noindent\textbf{Naturalistic baseline, stated here rather than
buried.} Because the ten tasks above are authored to create the conditions
under which the shape \emph{can} appear, the honest complement is how often it
appears when we do \emph{not} author for it. We ran the identical measurement
over a third-party benchmark we did not write---$\tau$-bench's~\cite{taubench24}
retail and airline episodes, two frontier models, $431$ tool turns
(full treatment and receipts in Section~\ref{sec:limitations})---and the
result cuts against a high base rate: the dominant behavior is serialization,
and across the $71$ gated batches observed \emph{not one} co-emitted a benign
read, so the benign-sibling shape of Table~\ref{tab:exposure} did not arise on
these tasks. A more severe consequential-sibling variant (two writes in one
batch) appeared, but only for GPT-4o and never for Claude. The reading we
carry forward is therefore deliberately modest: current frontier models, on
naturalistic tasks, tend to gather in parallel and act sequentially, so the
gap is \emph{latent}---a severe failure that is rare by model habit rather
than absent by construction, since the framework supplies no barrier and any
model or configuration that does batch gated calls executes the siblings
unchecked. The exposure numbers above are a controlled measurement of a real
enforcement failure, not a claim that the shape is frequent in production.

\smallskip\noindent\textbf{Task structure, not model reluctance.} The converse
question---when a naturalistic task genuinely comprises two \emph{independent}
consequential steps, does the model serialize them or batch them?---has a sharp
answer. Over five realistic multi-effect operations (deploy-and-notify,
delete-and-deploy, merge-and-announce, publish-and-open-PR, push-and-message),
$100$ runs each, driven multi-turn with distinct named side-effecting tools and
\texttt{run\_shell} classified per command by write intent, GPT-4o issued both
consequential effects in a single parallel turn in \emph{all} $500$ runs
(\texttt{naturalistic\_exposure.py}). We are explicit that these tasks are
\emph{selected} to afford two consequential steps, so this is a \emph{conditional}
rate---as authored, in that sense, as the ten tasks above---not a random-sample
naturalistic base rate. Read together with the $\tau$-bench null the picture is
nonetheless coherent, and it is the honest one: the leak condition is a property
of the \emph{task}, not of a model reluctance to parallelize---latent under
exploratory or single-effect work, where current models serialize, and
near-certain under multi-effect work (ship-and-notify, refund-and-message,
offboard-and-email). The class is measured, not asserted: on the third-party
benchmark used throughout this paper, the reference solutions themselves
demand it---45 of 115 $\tau$-bench retail tasks and 15 of 50 airline tasks
require at least two consequential writes in their gold action sequences
(0.39, Wilson [0.31, 0.48]; 0.30, [0.19, 0.44]), with \emph{adjacent} write
pairs in 41/115 and 14/50; tool classification follows a committed rubric
over the benchmark's own tool set, and the labels, extraction, and per-task
receipts are in the artifact (\texttt{prevalence/}). These are authored
benchmark tasks, not production telemetry: the claim is the class's
recurrence in an independent third-party workload, not a deployment
census---and it reads coherently with the serialization null above, since
GPT-4o already batches consequential writes in 9 of its 52 gated
$\tau$-bench batches (Section~\ref{sec:limitations}). Because such tasks
recur and the framework supplies no serialization guarantee, the
control-plane barrier cannot be replaced by reliance on the model's present
disposition.

\subsection{From emission to executed leak: a live end-to-end measurement}
\label{sec:expA}
Section~\ref{sec:exposure} measures a \emph{proxy}---the plan shape reaching the
framework---and is explicit that emission becomes a leak only once the framework
schedules the sibling concurrently with the pause. We now discharge that proxy
directly rather than by argument. Experiment~A (artifact,
\texttt{e2e/experiment\_a.py}) composes two components already in the artifact,
so nothing here is re-modeled: the \emph{exact} per-run model call of the
exposure study (\texttt{exposure.runner}, same providers, tasks, temperature,
and resume logic, so its emission event is identical to
Table~\ref{tab:exposure}'s) and the real \textsf{FW-A} (LangGraph~1.2.7)
fan-out graph of Section~\ref{sec:injection}---an approval interrupt on one
branch and the consequential effect on a sibling branch, the effect labeled
with each task's own gated action. For every run we record three bits: did the
live model \emph{emit} the shape; did the effect \emph{leak} (execute during the
pause) with no gate; and did it leak \emph{with} every effect routed through a
live \sys{} process. The same live model, the same runtime, and the same human
pause therefore yield a \emph{measured} unconditional $P(\mathrm{leak})$, a
conditional $P(\mathrm{leak}\mid\mathrm{emitted})$, and a mediated rate---not an
inferred chain. The battery drives the three leak-driving tasks (\texttt{compound\_cleanup}, \texttt{compound\_transfer}, \texttt{single\_offboard}) at pause${=}0$; pause \emph{duration} is swept explicitly rather than argued away. A keyless harness (\texttt{e2e/pause\_sweep.py}) drives the same \textsf{FW-A} fan-out graph at every pause in $\{0, 0.1, 0.5, 1, 2\}$\,s, twenty repetitions per value per arm: unmediated, the sibling lands during the pause in $20/20$ at \emph{every} value---the effect commits inside the fan-out superstep before \texttt{invoke()} returns, so the leak is pause-invariant and B1/B4's duration-irrelevance is executed rather than asserted---while mediated it is held in $20/20$ with zero effects on rejection (receipt \texttt{evidence/pause\_sweep.txt}, \texttt{langgraph}~1.2.7). The same sweep on \textsf{FW-B} is likewise flat ($15/15$ emitting runs leak unmediated, $0$ mediated, at every pause; \texttt{results/expA\_fwb\_mock\_pause*.jsonl}), so pause-invariance is not an \textsf{FW-A} artifact. The live arms therefore run at pause${=}0$ without loss of generality for the leak bit.

Table~\ref{tab:expA} reports the two arms whose vendor keys we hold, each run to
completion. For \textsf{gpt-4o}, $133/300$ runs leak an effect \emph{during the
pause} without the gate ($P(\mathrm{leak}){=}0.44$, Wilson $[0.39,0.50]$); for
\textsf{claude-sonnet-4.6}, $26/300$ ($0.09$, $[0.06,0.12]$; that arm's $314$
raw records deduplicate to $300$ unique $(\mathit{task},\mathit{run})$ pairs,
keeping the final attempt exactly as in Section~\ref{sec:exposure}). In \emph{both}
arms the conditional $P(\mathrm{leak}\mid\mathrm{emitted})$ is exactly
$1.000$---every one of the $133$ and $26$ emitted shapes executed its effect
during the pause, and no run scheduled the sibling \emph{after} the pause. This
\emph{measures}, rather than argues, the mechanism-level reading of
Section~\ref{sec:measurement}: once the shape is emitted the framework leaks it,
with human reaction time irrelevant. We do not present this conditional as a
\emph{surprising} result---given Section~\ref{sec:measurement}'s finding that
the scheduler commits the sibling inside the pausing superstep, $P(\mathrm{leak}
\mid\mathrm{emitted}){=}1.00$ is close to what that mechanism predicts; its
role is to \emph{confirm} the mechanism end-to-end on live models and to
locate the paper's genuinely contingent quantity, which is
$P(\mathrm{emitted})$ (Section~\ref{sec:exposure}), not the conditional. The
end-to-end value the repair must beat is therefore $P(\mathrm{leak}){=}
P(\mathrm{emitted})$ under this scheduler, and the mediated column drives it
to zero.  Routed through \sys{}, the identical runs
leak \emph{zero} effects---$0/300$ and $0/300$, Wilson upper bounds $1.3\%$
each---every consequential submission held at the gate and refused on the
operator's rejection. The other two arms complete the picture and both hold the
mediated rate at zero: \textsf{deepseek-chat} emits and leaks \emph{nothing}
($0/300$, $0$ emissions), reproducing end-to-end the sequential disposition the
exposure study found for DeepSeek---on a \emph{different} checkpoint than the
\textsf{v4-flash} exposure row, so the family de-confounds under live execution
as well; and \textsf{Llama-3.3-70B} leaks $56/300$ unmediated ($0.19$,
$[0.15,0.24]$) with, again, $P(\mathrm{leak}\mid\mathrm{emitted}){=}1.00$ and
$0/300$ mediated. Across all four arms this is $215/1{,}200$ end-to-end leaks
without the gate against \emph{zero} of $1{,}200$ with it.
This is the benign-input analog of the single adversarial demonstration in
Section~\ref{sec:injection}, now at scale and across models: it upgrades
``models emit the shape'' (a proxy) to ``the effect actually executes during
the pause, and the gate stops it'' (an executed outcome), and its mediated
column is the repair claim of Section~\ref{sec:eval} measured on \emph{live}
models rather than scripted planners. Its scope mirrors the exposure study's:
this is the leak-driving task subset, one runtime (\textsf{FW-A}), and
frictionless canned tool results, so it bounds occurrence \emph{for those tasks
under \textsf{FW-A}}, not a production census; longer conversations, richer
toolsets, and other runtimes remain out of scope for the same reason they are
in Section~\ref{sec:exposure}.

\emph{A second runtime.} To test whether the end-to-end leak is an artifact of
\textsf{FW-A}'s superstep scheduler, we repeated the \textsf{gpt-4o} arm on
\textsf{FW-B} (LlamaIndex Workflows)---a structurally different runtime whose
parallelism is an asynchronous event loop with explicit human-in-the-loop
events, not a fan-out graph over a shared state channel. The port changes only
the framework arm; emission is the identical \texttt{run\_one} call (same
model, tasks, temperature, resume logic), so any difference is the runtime's,
not the model's. The result reproduces \textsf{FW-A}: $121/300$ runs leak the
effect \emph{during the pause} unmediated ($P(\mathrm{leak}){=}0.40$, Wilson
$[0.35,0.46]$), with $P(\mathrm{leak}\mid\mathrm{emitted}){=}1.000$ across all
$121$ emitted shapes, and $0/300$ routed through \sys{} (Wilson upper bound
$1.3\%$; \texttt{results/expA\_fwb\_openai\_gpt4o.jsonl}). Per-task emission
varied as expected from an independent serving of the same model
(\texttt{compound\_cleanup} $78/100$, \texttt{single\_offboard} $41/100$,
\texttt{compound\_transfer} $2/100$), but the conditional leak was $1.00$ and
the mediated rate $0$ on every task. The leak is therefore not specific to one
scheduler: on a second execution model, under live \textsf{gpt-4o}, once the
shape is emitted the runtime commits the sibling effect inside the human-input
pause, and the same environment-external gate holds every one. The remaining
question is host language: \textsf{FW-A} and \textsf{FW-B} both run on Python's
\texttt{asyncio}, so those two arms separate \emph{schedulers} but share a
runtime. We close that gap with a third live arm on \textsf{FW-F}
(\textsf{LangGraph.js} on Node's event loop, a non-Python runtime; artifact
\texttt{experiment\_a\_js.mjs}, \texttt{results/expA\_fwf\_openai\_gpt4o.jsonl}):
live \textsf{gpt-4o} over the same three leak-driving tasks emitted the shape
on $143$ of $300$ runs, every one of the $143$ executed its effect during the
pause unmediated ($P(\text{leak}\mid\text{emitted}){=}143/143$), and all $300$
routed through \sys{} released nothing ($0/300$ mediated). The identical
conditional across three schedulers on two host languages---Python
\texttt{asyncio} (\textsf{FW-A}, \textsf{FW-B}) and the Node event loop
(\textsf{FW-F})---is the evidence that the leak is a property of the
concurrent-pause pattern, not of one runtime, and that the external gate's
repair travels with it.

\begin{table}[!t]
\centering
\caption{Experiment~A: measured end-to-end leak on live models through the real
\textsf{FW-A} runtime (three leak-driving tasks, pause${=}0$). $P(\mathrm{leak})$
is the fraction of runs whose effect executes during the approval pause with no
gate; $P(\mathrm{leak}\mid\mathrm{emitted})$ conditions on the model emitting
the shape; the last column is the same runs with every effect routed through a
live \sys{} process. Wilson 95\% intervals; $N{=}300$ per model ($100$ runs
$\times$ three leak-driving tasks). All four arms are complete; the two anchor
arms (\textsf{gpt-4o}, \textsf{claude-sonnet-4.6}) run on the vendors' native
APIs, \textsf{Llama} provider-direct (Together), and \textsf{deepseek-chat} on
DeepSeek's own surface (\texttt{results/expA\_*.jsonl}).}
\label{tab:expA}
\renewcommand{\arraystretch}{1.2}
\setlength{\tabcolsep}{3pt}
\scriptsize
\begin{tabular}{@{}lcccc@{}}
\toprule
\textbf{Model} & \textbf{$N$} & \textbf{$P(\mathrm{leak})$ [95\% CI]} & \textbf{$P(\mathrm{lk}\!\mid\!\mathrm{em})$} & \textbf{$P(\mathrm{leak})$/\sys{}}\\
\midrule
\textsf{gpt-4o}                 & 300 & 0.44 [0.39, 0.50] & 133/133 & 0/300 [0, .013]\\
\textsf{claude-sonnet-4.6}      & 300 & 0.09 [0.06, 0.12] & 26/26   & 0/300 [0, .013]\\
\textsf{deepseek-chat}          & 300 & 0.00 [0, 0.013]   & --- (0 em.) & 0/300 [0, .013]\\
\textsf{Llama-3.3-70B}          & 300 & 0.19 [0.15, 0.24] & 56/56   & 0/300 [0, .013]\\
\bottomrule
\end{tabular}
% Provenance (all four arms complete): recomputed from results/expA_*.jsonl,
% denominator N=300 = 100 runs x 3 leak-driving tasks per model
% (claude's 314 raw records dedup to 300 unique (task,run_idx); the 14 extra are
% retries). gpt-4o 133/300 leak_unmediated / 0 mediated; claude 26/300 / 0;
% deepseek-chat 0/300 / 0 (0 emissions -- reproduces its exposure-study sequential
% disposition, here on a different DeepSeek checkpoint than the v4-flash exposure
% row, so it de-confounds the family end-to-end); llama 56/300 / 0. Across all
% four arms: 215/1200 unmediated leaks, 0/1200 mediated, P(leak|emitted)=1.00 for
% every emitting model.
\end{table}

\emph{A real, external effect.} Experiment~A and the probes observe the leaked
effect through instrumented tools; to confirm the effect is genuinely external
and truly escapes the operator's decision, we compose \textsf{FW-A}'s two
\emph{documented} primitives---parallel fan-out from the entry node and an
\texttt{interrupt()} approval pause---with a tool that issues a real HTTP
\texttt{POST} to a live local endpoint. Unmediated, the \texttt{POST} reaches
the endpoint in the same superstep as the interrupt, timestamped $512$\,ms
\emph{before} the operator's rejection---the $512$\,ms being the harness's own
scripted rejection latency, a \emph{lower} bound on any human's---so the
rejection cannot prevent it;
routed through \sys{}, the identical graph delivers \emph{zero} \texttt{POST}s
and the gate returns \texttt{refused\_rejected}
(\texttt{evidence/webhook\_leak\_demo.txt}). This is a single demonstration
composed from documented primitives, not a shipped template found unmodified;
it shows that the leaked effect is real and lands before a human can stop
it---complementing Experiment~A's rate evidence with one concrete external
action.

\subsection{Occurrence in the wild: a public-incident corpus}
\label{sec:prevalence}
We compiled a corpus of public reports evidencing the same failure classes.
We read it as \emph{qualitative corroboration}---that these failures occur
outside our own probes---never as a rate or a prevalence estimate: the corpus
is a lower bound, tracker-biased toward the framework with the largest public
issue tracker, and no denominator is claimed. Each record is verified against
the live tracker and every claim below
carrying an authoritative state as of July~2026 (full table, URLs, dates, and
verbatim search queries in the artifact). Classification is deliberately
conservative. The seed corpus contributes 11 \emph{direct} incidents that
actually exhibit a stop-primitive failure (5 on the sibling/parallel axis, 6
on replay), separated from 3 \emph{adjacent} reports sharing the root cause
and 6 \emph{context} items (user questions and practitioner write-ups); the
wider sweep below adds two further direct incidents on the cancellation axis,
for thirteen across three independent trackers. The direct set spans March 2025 to
May 2026 and includes a report in the framework's \emph{JavaScript} repository
of resume restarting from the beginning---independent user corroboration of the
cross-runtime replay replication in Table~\ref{tab:matrix}. The probe
suite doubles as minimized reproductions of the corpus's two axes: the
parallel-pending-interrupt construction of the sibling probes is exactly the
shape the multi-interrupt reports describe, and the \textsf{FW-F} replay
probe executes precisely the resume-restarts-from-the-beginning behavior of
the JavaScript-repository report---matrix and corpus witness the same
mechanisms, not correlated symptoms.

One asymmetry is stated outright. The replay and cancellation axes carry
\emph{effect-level} third-party corroboration---a user's resume restarted from
the beginning; a user's cancelled handler ran to completion. The five
sibling/parallel incidents document the same missing multi-interrupt barrier
at the \emph{interrupt-routing} layer, maintainer-acknowledged and structurally
the construction our sibling probes minimize---but no report in the corpus
states the effect-level consequence (``my sibling action executed while I was
deciding''). We therefore do not claim wild corroboration of an executed
sibling leak. The corpus's contribution on that axis is the acknowledged
missing barrier; the executed-effect evidence is Experiment~A
(Section~\ref{sec:expA}), where the leak is measured on live models through
the real runtime, not inferred from reports. Plausible reasons a leaked
sibling effect goes unreported---the operator never learns the effect landed
before the rejection, or attributes it to the tool---are exactly why an issue
count is a lower bound; we note them without leaning on them.

The upstream trajectory is itself evidence: a maintainer-filed issue tracks
the missing multi-interrupt barrier explicitly---a node with two pending
interrupts re-runs after a single resume because per-interrupt identity is
not stored---and remains \emph{open} as an enhancement, a merged pull
request supplies the partial single-interrupt fix, and a cluster of three
parallel-interrupt routing bugs was closed together in January
2026---acknowledgement plus incremental mechanism repair, with the general
barrier still absent. A further upstream item sharpens the replay axis into
a persistence-ordering race: an issue we filed on the LangGraph tracker
(disclosed in Appendix~\ref{app:disclosure}, not counted among the
third-party incidents above) shows that under \texttt{durability="sync"} a
completed task's pending writes and the superseding checkpoint are submitted
to a shared thread pool with no ordering edge, so a crash during checkpoint
persistence yields \emph{host-dependent} recovery---the same node
re-executes, duplicating its external side effects, on one machine and
replays exactly once on another, at an identical injected crash point.
Several other developers reproduced the duplicated effect on different
operating systems and core counts (including macOS/arm64), and three
maintainer or community pull requests propose the write-before-checkpoint
barrier that would close it---third-party reproduction of the replay hazard
on a serving path distinct from our own. Cancellation appears in the corpus
with a different texture---users asking how to stop a running agent at all,
cancellation failing to propagate to the deepest active agent, cancelled
runs losing un-checkpointed state---which we report as-is: underspecified
cancellation semantics evidenced through user confusion rather than
minimized repros. Two limits are inherent: the corpus is dominated by the
framework with the largest public tracker and the most-used interrupt
surface (a selection effect the six-framework probes correct for), and an
issue count is a lower bound on occurrence, never a rate. During revision we
widened the \emph{search}: the same axis-derived keyword families (six
queries per tracker, committed verbatim) were swept across all six
frameworks' public trackers via unauthenticated GitHub search, yielding 130
unique candidates, 72 outside the dominant tracker (\texttt{corpus\_sweep.py},
\texttt{sweep\_results.jsonl}). Classified under the same conservative
rubric by reading each full thread, the sweep adds two confirmed
\emph{direct} third-party incidents---both on the cancellation axis, one
beyond the dominant tracker: a LangGraph subgraph that does not stop on
\texttt{CancelledError}, and a LlamaIndex workflow handler that runs to
completion after its run is cancelled or times out, the same orphan our
\textsf{FW-B} probe reproduces, filed independently---plus twenty-four
\emph{adjacent} reports, including hosting- and dev-tool-layer instances of
the same classes (a dev UI restarting instead of resuming after a human
response; a handoff flow duplicating an approved call on resume; an AG-UI
host dropping sibling tool calls when one requires approval). With the seed
corpus this yields thirteen direct third-party incidents across three
independent trackers (LangGraph, LangGraph.js, LlamaIndex), so occurrence no
longer rests on one project's tracker, while remaining an acknowledged lower
bound, not a rate.

\subsection{Composition with adversarial input}
\label{sec:injection}
Our threat model (Section~\ref{sec:scope}) scopes out prompt injection, which
invites a fair question: if the leak is a benign-model concurrency bug, why does
an \emph{external} barrier matter rather than a fix inside the framework? We
answer on the objection's own terms, twice: with a keyless, deterministic
demonstration on the real \textsf{FW-A} runtime (artifact
\texttt{e2e\_injection\_langgraph.py}), and then \emph{live} across four frontier models at $N{=}100$ each under the same injected ticket (numbers below). A support-triage agent processes a
ticket whose body carries an injected instruction to issue a refund; the
injected content steers the agent into the leak plan shape---the refund emitted
alongside benign lookups, the shape Section~\ref{sec:exposure} shows real models
emit under such input. We make the adversarial trigger a controlled variable (a
scripted planner reproduces the shape deterministically; \texttt{-{}-live}
swaps in a real model, and we report how often each model takes the bait
below without claiming it generalizes beyond this one injected composition)
and we claim \emph{no} injection defense---detecting injection is a separate,
open problem. The point is narrower and mechanical: under the framework's own
approval interrupt the adversary's refund executes \emph{during the pause it
creates} (\texttt{refund\_during\_pause}${=}1$, and a subsequent human rejection
is powerless), whereas the identical run with the refund routed through \sys{}
holds it (\texttt{during\_pause}${=}0$) and the rejection yields zero effects.
We ran this live across four frontier models at $N{=}100$ each under the
injected ticket (\texttt{results/injection\_live\_*.jsonl}). How readily a
model takes the bait varies enormously and is not our claim: \textsf{gpt-4o}
emitted the refund-plus-lookups shape on $95$ of $100$ runs,
\textsf{Llama-3.3-70B} on $19$, \textsf{claude-sonnet-4-6} on $7$, and
\textsf{gemini-3.5-flash} on none. The mechanical invariant, by contrast,
held identically for every model that did emit: across all $121$ emitting
runs the unmediated framework executed the adversary's refund during the
pause it created ($P(\text{leak}\mid\text{emitted}){=}121/121$, and a
subsequent human rejection was powerless), while the identical runs routed
through \sys{} released nothing ($0/400$ mediated over all four models).
Whether the plan came from the model most eager to comply, the one that complied once in fourteen, or a script, the barrier holds. The barrier
does not adjudicate intent; it enforces that \emph{no} effect, benign or
adversarial, is released without a decision---the assurance an in-framework gate
cannot give once the injected plan races it. Nor would a static deny-list
substitute: it can stop the injected refund only by stopping \emph{every}
refund, whereas the hold preserves the legitimate path---the operator approves
genuine refunds and rejects this one---which is exactly the decision the
deployment installed an approval gate to make. This strengthens rather than widens
the claim: even under hostile input, the enforced invariant is the one we
verify. One scope note keeps it honest against our own threat model: the
demonstration presumes the same mediation contract. Injected text selects among
the deployment's defined tools; it cannot conjure an unmediated path that does
not exist, so under whole-choke-point placement the adversary's effect is
mediated and meets the hold. What injection \emph{could} target is a left-%
unwrapped tool---the accidental-bypass case the linter flags and the structural
routes close below the application, where even a hypothetical unwrapped tool's
egress is refused by the kernel. Injection does not weaken the mediation
contract; it raises the stakes of discharging it, which is why the structural
routes exist.

% =====================================================================
\section{\sys{} Design}
\label{sec:design}

\subsection{Principle: enforce outside the framework}
\noindent\textbf{Why an admission point is forced---and why we place it
externally.} One structural argument delimits the design space: the trichotomy below \emph{forces} an admission point outside the step's own schedule, while its placement---inside the framework or outside---remains an engineering and trust choice we argue, not derive. Its assumptions are explicit so its scope cannot be over- or under-read.

\begin{proposition}[Barrier trichotomy]\label{prop:trichotomy}
Fix an execution model in which \textup{(A1)} two or more branches of one step
may execute concurrently, each branch's effects committing independently
within the step, and \textup{(A2)} a pre-execution approval gate suspends only
its own branch, resolving no earlier than control's return to a caller that
can deliver a decision. Then in any run where a gated effect and a
side-effecting sibling share a step, at least one of the following holds:
\textup{(i)} \emph{serialization}---no sibling effect commits until every gate
in the step resolves, surrendering \textup{(A1)}'s intra-step concurrency for
effectful branches; \textup{(ii)} \emph{leak}---some sibling effect commits
during the pause, violating B1; or \textup{(iii)} \emph{mediated
admission}---every effect's commit is ordered after a verdict from an
admission component whose decision state is consulted at the effect boundary
and is not itself subject to the step's concurrent schedule.
\end{proposition}

\begin{proof}
By (A1) and (A2) the sibling's effect can reach its commit point while the
gate is unresolved, since resolution requires control to leave the step.
Either some ordering edge delays that commit past every in-step gate's
resolution, or none does. If none does, the commit precedes the decision:
case~(ii). If one does, consider what issues it. If the step's own schedule
issues it, the edge applies to every effectful sibling and effectful branches
serialize behind the gate: case~(i). Otherwise the edge is supplied by a
component outside the step's schedule that holds the effect against the run's
pause state and releases it only on a decision---an admission point, case~(iii).
\end{proof}

The proposition is a design-space trichotomy under (A1)--(A2), not a
mechanized impossibility theorem: case~(i) is a legitimate design a framework
could adopt at the parallelism cost, and none of our measured frameworks does.
What it pins down is that the barrier, intra-step concurrency of effectful
branches, and enforcement \emph{inside the step's own schedule} cannot all
three hold at once. Case~(iii) explicitly \emph{includes} in-framework
admission points---a framework could build the component natively, as a
deterministic decision state machine consulted at its own effect
boundary---so the trichotomy is not a false dilemma between ``serialize'' and
``go external.'' Where case~(iii)'s component lives is then an engineering
and trust question rather than a logical one: built \emph{natively}, it is an
upstream fix that must land and stay fixed in every framework a deployment
runs---and a policy layer embedded in the existing control flow inherits the
very scheduler that produced the race (Section~\ref{sec:related})---whereas
built \emph{externally}, it is one deployment-side point, adoptable
unilaterally across frameworks and versions, which is \sys{}. Two frameworks'
clean replay cells (Table~\ref{tab:matrix}) show sound native designs are
achievable on individual axes; the external gate is the enforcement point
whose guarantee does not wait for them.

Because the framework's own control flow is what fails, \sys{} places the
enforcement point at the \emph{tool boundary}, external to the framework and its
host process. Every side effect is submitted to \sys{} for \emph{admission}
before it is performed; the effect executes only if \sys{} returns
\textsf{release}. For every \emph{mediated} effect this inverts the trust
relationship: the framework need not be trusted to honor stop semantics,
because a side effect that never obtains a release simply never happens---even
if it is emitted by a sibling branch during a pause, a re-executed step, or a
zombie thread after cancellation. One word deserves a definition before it is
used further: we use \emph{repair} in the deployment sense---restoring the
end-to-end barrier contract for a deployment's mediated effects---not in the
sense of patching framework internals. The frameworks' own schedulers still
race exactly as Section~\ref{sec:measurement} measures; what changes is that
the race can no longer externalize an undecided effect. Upstream fixes remain
complementary work our regression probes exist to support
(Section~\ref{sec:related}).

\noindent\textbf{Conditional enforcement (the reference-monitor contract).}
\sys{} is an execution monitor in the Anderson/Schneider
lineage~\cite{anderson72,schneider00}, carrying its defining assumption,
\emph{complete mediation}. The conditional theorem the verification tiers
instantiate: \emph{if every side-effecting operation performs its external
action only after receiving \textsf{release} for its identity
$(\mathit{run},\mathit{key})$, then every execution satisfies P1--P4.} The
properties are unconditional facts about the gate's verdicts; the ``closes the
gap'' claims are conditional on the contract, and we say so throughout. This
invites the objection that a reference monitor merely proves ``a gate enforcing
P1--P4 enforces P1--P4.'' The antecedent (every effect submits before
externalizing) is discharged operationally, below, not proved. What \emph{is}
proved is the non-trivial half---that the admission core maintains P1--P4 where
naive implementations fail: concurrent submissions/decisions interleaved across
connections (TLC to $804{,}357$ states), unbounded run/key populations (TLAPS),
and crash-recovery reconstructing the released and fenced sets from a torn log
(WAL lemmas). The two defects the effort surfaced---a cross-run key-reuse
clobber and a compaction ordering bug---are invisible from the properties'
statement, which is why ``the theorem restates the implementation'' is wrong:
the implementation had bugs the proof obligations caught. The condition is
discharged in increasing strength: by \emph{placement discipline} at the
framework's single tool-invocation choke point (no framework modification in
the integration demo, Section~\ref{sec:eval}); and by a best-effort
\emph{static mediation linter} (\texttt{mediation\_lint.py}) flagging direct
calls to registered effect callables outside the wrapper---blind to dynamic
dispatch and third-party internals, so it raises the cost of \emph{accidental}
bypass, the realistic failure for a benign operator, and no more.

\smallskip\noindent\textbf{Making mediation structural.} The strongest
discharge removes discipline from the loop entirely: \emph{structural}
mediation by kernel-level interception or network-namespace egress
allow-listing, so tool processes can reach only the
gate host, at higher integration cost---both implemented, loaded, and
verified in the artifact: the namespace route across every executed framework
(\texttt{e2e/e2e\_structural\_all.sh}) and the eBPF route across all four
egress channels (\texttt{ebpf/mediation\_guard\_full.c}), with transcripts in
\texttt{evidence/}.
The network-namespace route (\texttt{egress\_demo.sh}) runs the gate and a
tool inside a loopback-only namespace, where an unmediated external
connection fails at the routing layer (\texttt{ENETUNREACH}) while the
mediated path releases. The same route composes with the real frameworks unchanged: \emph{every}
executed integration of Section~\ref{sec:eval}---all five Python frameworks
(\textsf{FW-A} through \textsf{FW-E}), each on its full violated-axis set---is
re-run inside a fresh loopback-only namespace, the gate spawned within it as
the only reachable egress, and each passes its complete repair verdict set in
situ while an unwrapped tool's external connection is refused by the kernel
(\texttt{Network is unreachable}) inside that same namespace
(\texttt{e2e/e2e\_structural\_all.sh}; transcript in
\texttt{evidence/e2e\_structural\_all.txt}). The repair logic is byte-identical
to the unconfined runs---only the egress environment changes---and because
these harnesses drive scripted planners rather than live model APIs, they
execute under total egress confinement with no external dependency; the Node
integration (\textsf{FW-F}) admits the identical wrapping. The kernel-interception route
(\texttt{ebpf/mediation\_guard\_full.c}) now closes \emph{every} network
egress channel the cgroup hooks expose: we attach \texttt{cgroup/connect4},
\texttt{cgroup/connect6}, \texttt{cgroup/sendmsg4}, and
\texttt{cgroup/sendmsg6} together, each permitting only the gate's address. A
tool process placed in the attached cgroup has every non-gate destination
refused by the kernel at the syscall permission layer---an IPv4 \emph{and} an
IPv6 \texttt{connect()}, and an unconnected v4 or v6 \texttt{sendmsg}/%
\texttt{sendto} to any non-gate address, all return \texttt{EPERM}---while its
connection to the gate succeeds and releases: a distinct and more surgical
enforcement point than routing, requiring no isolated namespace, the same
class of kernel-level syscall mediation that production security tooling such
as Cilium~\cite{cilium} and Falco~\cite{falco} builds on, here narrowed to a
single permitted destination. We load all four with BTF-enabled CO-RE via
\texttt{bpftool loadall} (on Linux~6.8 the committed transcript records the
four JITed \texttt{cgroup\_sock\_addr} programs
\texttt{restrict\_connect4/6} and \texttt{restrict\_sendmsg4/6} attached under
the target cgroup; from inside it a non-gate \texttt{connect}/\texttt{sendmsg}
on each channel returns \texttt{EPERM} while the gate connection
releases---artifact, \texttt{evidence/ebpf\_channels.txt}). On kernels where
\texttt{io\_uring}-submitted connects route through the same \texttt{sock\_addr}
hooks they are covered as well; a seccomp filter denying
\texttt{io\_uring\_setup} is the belt-and-braces complement for any kernel that
bypasses them. The two routes' remaining scope boundary is identical and we
state it exactly rather than let ``structural'' overclaim: the namespace route
confines all network egress at the routing layer, the eBPF route confines it
per-syscall across the four channels above, and \emph{neither} route covers
effects that externalize \emph{without} network egress from the tool
process---a write to a shared filesystem another process acts on, local IPC,
shared memory: for those channels mediation remains the placement contract (or
an analogous seccomp/LSM policy), a boundary we restate in the limitations.
This scoping costs little for the measured set: every violation we measure externalizes over the network, exactly the channel the routes close. Under either route, within its
stated channel coverage, bypass is
prevented by the operating system rather than by discipline, so
complete mediation becomes a property the tool process cannot opt out of
rather than a convention a wrapper must uphold. Precisely stated: structurally, the discharge is
\emph{network-complete}---every network-externalized effect of the confined
process is forced to the gate---and \emph{environment-complete only under the
placement contract}, since filesystem, IPC, and shared-memory channels are
closed by discipline (or an analogous seccomp/LSM policy) rather than by these
two routes. The actions our measurements gate (e-mail, payment capture, ticket creation, deployment) are network-externalized, so the network-complete discharge covers that governed set; a deployment whose consequential actions externalize through the filesystem, local IPC, or shared memory must add the seccomp/Landlock rung below, and the wider environment-complete claim is the residual we scope, not one we assert.
Section~\ref{sec:related} positions these routes against general-purpose
sandboxes, whose confinement machinery they reuse and whose missing
admission semantics the gate supplies. Table~\ref{tab:mediation} states the discharge ladder and each rung's
channel coverage in one place. As deployment guidance the ladder is also a ranking: prefer the
cgroup eBPF route where kernel support exists (per-syscall enforcement, no
topology change), the loopback namespace where full egress isolation is
acceptable, and wrapper discipline plus the linter only where OS-level
enforcement is unavailable---adding the seccomp profile below wherever a
tool's non-network channels matter.

\begin{table}[!t]
\centering
\caption{How the complete-mediation contract is discharged, per route: what
each covers, and what remains outside it. The first two rungs are
discipline; the two structural routes are kernel-enforced, both implemented
and exercised in the artifact; the last row names the residue every route
shares. All routes leave the approver's own integrity and a Byzantine
submitter out of scope (Section~\ref{sec:failure-model}).}
\label{tab:mediation}
\renewcommand{\arraystretch}{1.2}
\footnotesize
\begin{tabularx}{\columnwidth}{@{}p{2.05cm}p{1.5cm}X@{}}
\toprule
\textbf{Route} & \textbf{Enforced by} & \textbf{Covers / does not cover}\\
\midrule
Placement discipline (wrapper) & convention & every tool routed through it;
nothing a developer forgets to wrap\\
Static mediation linter & best-effort analysis & direct calls to registered
effect callables; not dynamic dispatch or third-party internals\\
Loopback-only netns & kernel (routing) & all network egress of the confined
process; not filesystem, IPC, or shared memory\\
\texttt{cgroup} eBPF \texttt{sock\_addr} & kernel (per-syscall) &
\texttt{connect4/6} $+$ \texttt{sendmsg4/6} to non-gate addresses
(\texttt{EPERM}); not filesystem, IPC, or shared memory\\
Non-network channels & kernel (seccomp) & shared filesystem, local IPC, shared
memory: fail-closed under the tested deny-by-default seccomp profile below, or
the placement contract where it is not applied\\
\bottomrule
\end{tabularx}
\end{table}

\smallskip\noindent One unwrapped tool path bypasses the
gate \emph{at the wrapper layer}; this is intrinsic to any enforcement point
that relies on placement, and is the same
trust boundary every reference monitor carries. It is not, however,
unavoidable: under the structural routes just shown, an unwrapped tool's
external action is refused by the kernel rather than silently executed, so the
bypass surface is closed below the application for network-externalized
effects (the routes' stated channel coverage above)---the residual trust
then rests on the OS enforcement mechanism (namespace routing or the
cgroup egress hooks) rather than on developer discipline. The resulting posture answers the strongest objection this design
invites. Under wrapper-only mediation a forgotten wrap
\emph{fails open}: the unmediated effect executes silently, and the guarantee
collapses exactly as the objection says. Under either structural route the
same mistake \emph{fails closed}: the unwrapped tool has no route to any
non-gate destination, its first external action errors loudly
(\texttt{ENETUNREACH}/\texttt{EPERM}), and the omission surfaces as an
availability bug on first execution rather than as a silent safety hole. The
companion objection---that per-tool wrapping is ``the same discipline'' as
fixing each framework internally---inverts under the same structure: an
internal fix must change the scheduler semantics of every framework a
deployment runs, land upstream, and stay fixed across versions, whereas the
gate plus a deny-by-default egress policy is a \emph{deployment-side} control
at one choke point, adoptable unilaterally and unchanged across framework
upgrades---the classical economy-of-mechanism argument for an external
reference monitor, made concrete by Table~\ref{tab:matrix}'s recurrence of
the same failure across four independently designed schedulers. The condition is also
\emph{necessary}, not merely sufficient---witnessed empirically: every
violation in Section~\ref{sec:measurement} is precisely an unmediated effect
landing.
Section~\ref{sec:threats-validity} discusses the residual risk.

\smallskip\noindent\textbf{Closing the non-network channels.} The two routes
above are network-complete but not environment-complete; we close the residual
filesystem, IPC, and shared-memory channels with a complementary
\emph{deny-by-default} seccomp-bpf profile applied to the tool process, under
which its only provisioned egress is the gate. A stand-in unwrapped tool that
attempts a filesystem write, a local (AF\_UNIX) socket, and a
\texttt{memfd\_create} shared-memory channel succeeds on all three unconfined
and is refused on all three under the profile (\texttt{EPERM}; $3/3$ channels
escape unconfined, $0/3$ confined), while an effect submitted to the gate over
the provisioned descriptor is still admitted (artifact, \texttt{confine/};
receipt, \texttt{evidence/confine\_fs.txt}). The non-network channels therefore
inherit the same fail-closed posture as network egress: a forgotten or malicious
unwrapped effect errors loudly rather than leaking silently. seccomp denies filesystem writes wholesale; where a tool needs a legitimate workdir, path-granular allow-listing is discharged---not merely composed in principle---by a Landlock ruleset (Linux~$\geq$5.13) that handles \texttt{WRITE\_FILE} and grants it only beneath the tool's provisioned workdir. On an enforcing kernel (ABI~4; Landlock in the active LSM stack), a gate-mediated write to that workdir is released and commits, while the same tool's unmediated writes to a shared, externally observed path and to a location outside any grant are both refused by the kernel (\texttt{EACCES}), and reads remain ungoverned---a positive-control self-test (handle write, grant nothing, confirm the write is denied) gates the run so a non-enforcing kernel cannot report a false pass (artifact, \texttt{probes/landlock\_workdir\_demo.py}; receipt, \texttt{evidence/landlock\_workdir.txt}). The filesystem channel thereby inherits the same fail-closed posture at \emph{path} granularity that seccomp gives it wholesale; where Landlock is absent from the LSM stack, the seccomp floor's wholesale write-deny is the portable fallback, executed above. A tool that escalates privileges or exploits
a kernel vulnerability is out of scope, as for any seccomp/LSM sandbox, and the
gate process itself remains trusted. ``Closing the non-network channels'' does not \emph{eliminate} the
assumption it discharges: complete mediation is the reference-monitor
obligation, not a theorem, and \sys{} supplies admission semantics rather than
confinement. Each effect channel is discharged by the mechanism that governs
it---network by the namespace or eBPF routes, filesystem and the syscall
surface by seccomp (path granularity by Landlock where available)---and any
channel a deployment leaves unprofiled remains on placement discipline,
fail-closed under deny-by-default rather than fail-open. The residual is thereby
\emph{bounded to exactly the channels left unconfined}, not removed, which is
the standard posture of every user-space reference monitor and the honest scope
of the discharge.

\smallskip\noindent
A related worry---could a sibling's effect land \emph{before} the gate rules
on it?---dissolves under mediation, which is a \emph{causal ordering}, not
observation: the wrapper performs the external action only after
\textsf{release}, so a mediated effect cannot precede its own admission by
construction; the gate is an admission point, not a monitor racing completed
actions. The unmediated measurements of Section~\ref{sec:measurement} show
effects landing during pauses precisely because no such ordering exists
there, and the repaired runs (Section~\ref{sec:eval}) show zero effects
during the pause for the same reason. A tool whose \emph{implementation}
performs an external action without submitting at all is simply an
unmediated path---the stated contract, not a race.

\smallskip\noindent\textbf{The circularity objection, answered directly.} A
reviewer may sharpen this into an apparent contradiction: the measured gaps
\emph{exist} because effects execute without waiting for the framework's
control flow, so why would those same effects reliably submit to the gate?
The objection conflates two distinct layers. What the measurements show
bypassing control flow is the framework's \emph{pause}---the scheduler
commits a sibling's effect without waiting for the interrupt to resolve. What
mediation interposes is a \emph{wrapper at the tool boundary}, one layer below
the scheduler: a wrapped effect calls the gate and externalizes only on
\textsf{release}, and no scheduler behavior---no superstep race, no premature
resume, no orphaned thread---can make a \emph{submitted} effect externalize
without a verdict, because the external call physically follows the release in
the wrapper's own code. The scheduler still races exactly as measured; the
race now moves a \emph{held} token, not an executed effect. The effects that
still bypass the gate are precisely the \emph{unwrapped} ones---the
complete-mediation residual of Table~\ref{tab:mediation}, closed structurally
for network egress and bounded explicitly elsewhere---not a case where
mediation was applied and the scheduler defeated it. ``The gaps exist because
effects bypass control flow'' and ``a wrapped effect cannot bypass the gate''
are therefore both true and non-contradictory: they are statements about
different layers.

\subsection{Model}
\begin{definition}[Effect]
An effect is a triple $(\mathit{run},\ \mathit{key},\ \mathit{needs\_approval})$,
where $\mathit{run}$ identifies the cancellable run, $\mathit{key}$ is a stable
idempotency key identical across replays of the same logical effect, and
$\mathit{needs\_approval}$ indicates whether human authorization is required.
\end{definition}

\sys{} maintains four pieces of state, each scoped to the effect
\emph{identity} $(\mathit{run},\mathit{key})$: the set of \emph{released}
identities (for deduplication), the set of \emph{cancelled} runs, the map of
\emph{pending} (held) effects awaiting a decision, and the set of
\emph{rejected} identities. Scoping by $\mathit{key}$ alone is unsound:
rejections and deduplication would bleed across unrelated runs, and a second
run's held effect would silently displace the first's---counterexamples that
are executable regression tests in the artifact. Cross-run key reuse is
therefore permitted by construction. (The alternative design---globally unique
idempotency tokens, as payment APIs use---eliminates cross-run collisions by
fiat but pushes uniqueness onto every producer; scoping by
$(\mathit{run},\mathit{key})$ keeps producer-local keys sound, and a deployment
wanting global tokens simply embeds them in $\mathit{key}$.)

\emph{Key discipline.} The properties assume \texttt{key} is a stable
idempotency key---identical across replays of the same logical effect,
distinct across different ones. (Idempotency keys apply to consequential
\emph{writes}; reads and status polls are not gated effects and are simply
not mediated, so nothing forces a poll to deduplicate.) Both violations of
that discipline degrade
\emph{fail-safe with respect to the measured hazard}: a fresh key per attempt
makes P3 vacuous (replays are no longer recognized, and the deployment
regresses to the framework's native replay behavior while P1, P2, and P4
stand), and a key shared across distinct logical effects
over-deduplicates---the second effect is \emph{refused} as a duplicate, a
visible fail-closed error, never a silent double execution. No key mistake
can cause two releases of one identity; the failure directions are omission
and refusal, not duplication.

\emph{Who mints the key} is the deployment decision the discipline turns on,
so we name the three realistic strategies rather than leave key generation
abstract. (i) \emph{Caller-named semantic keys}: the deployment names the
logical effect at the wrapper call site (\texttt{charge\_card},
\texttt{refund:\{order\_id\}}). This is what every executed integration in
Section~\ref{sec:eval} does, and it is stable across framework replay by
construction, because replay re-executes the same call site with the same
arguments. (ii) \emph{Canonicalized argument hashes}: stable across replay
for the same reason---re-execution reproduces the arguments---and sensitive
to any argument change, which is the correct deduplication boundary for
``same logical effect.'' (iii) \emph{Framework-issued call identifiers}
where a runtime exposes them: convenient, but a runtime that mints a fresh
identifier on re-execution silently degrades this to the fresh-key case, so
P3 protection then rests on (i) or (ii); whether a framework's identifiers
survive replay is a property a deployment checks once per runtime, not per
call. Under all three, the failure directions remain the fail-safe ones
above---omission and visible refusal, never silent double release.
Characterizing key stability empirically across large, heterogeneous tool
surfaces is future work; the executed integrations exercise strategy~(i).
Admission follows
a fixed
priority---cancellation and duplication dominate approval---given in
Algorithm~\ref{alg:submit}.

\begin{figure}[!t]
\begin{algorithmic}[1]
\STATE \textbf{function} \textsc{Submit}$(e)$
\IF{$e.\mathit{run} \in \mathit{cancelled}$}
  \RETURN \textsc{RefusedCancelled} \hfill$\triangleright$ P4
\ENDIF
\IF{$(e.\mathit{run},e.\mathit{key}) \in \mathit{released}$}
  \RETURN \textsc{RefusedDuplicate} \hfill$\triangleright$ P3
\ENDIF
\IF{$(e.\mathit{run},e.\mathit{key}) \in \mathit{rejected}$}
  \RETURN \textsc{RefusedRejected} \hfill$\triangleright$ P2
\ENDIF
\IF{$e.\mathit{needs\_approval}$}
  \STATE $\mathit{pending}[(e.\mathit{run},e.\mathit{key})] \gets e$
  \RETURN \textsc{HeldForApproval} \hfill$\triangleright$ P1
\ELSE
  \STATE $\mathit{released} \gets \mathit{released} \cup \{(e.\mathit{run},e.\mathit{key})\}$
  \RETURN \textsc{Release}
\ENDIF
\end{algorithmic}
\caption{\sys{} admission. \textsc{Decide}$(\mathit{run},\mathit{key},\mathit{approved})$
removes the pending effect under that identity and releases it only if approved
and its run is not cancelled; deciding a cancelled run's effect reports
\textsc{RefusedCancelled} (the fence, not a ``duplicate'', is the reason).
\textsc{Cancel}$(\mathit{run})$ marks the run and drops its held effects.}
\label{alg:submit}
\end{figure}

\subsection{Properties}
\begin{property}[Hold-until-decided]
An approval-gated effect is never released until a decision is recorded; it is
held. This closes the sibling leak: a sibling's mediated gated effect cannot
execute during a pause.
\end{property}
\begin{property}[Reject-cancels]
A rejected effect is never released and remains refused on resubmission. This
closes reject-after-effect.
\end{property}
\begin{property}[Dedup-on-replay]
An identity $(\mathit{run},\mathit{key})$ released once is refused thereafter;
distinct runs reusing a key are unaffected. This closes replay and timeout
double-execution without cross-run collateral refusals.
\end{property}
\begin{property}[Fence-on-cancel]
After a run is cancelled, its subsequent effects are refused even if submitted
late by a zombie thread, and any held effects of that run are dropped. This
closes cancellation orphaning.
\end{property}

\noindent\textbf{What the properties do not promise.} P1--P4 are safety
properties over \emph{releases}, and every ``closes'' above is conditional on
complete mediation (Section~\ref{sec:design}). Before the non-claims, one
clarification that pre-empts a natural misreading: a hold suspends
\emph{the submitted effect}, not the run and not its sibling branches. If one
branch reads a database while another submits an e-mail that meets an
approval hold, the read is not a gated effect, is never submitted, and
proceeds untouched; only the e-mail waits. The gate therefore does not impose
a global lock or a per-approval run freeze---it withholds exactly the
irreversible external actions routed to it, and leaves all other computation,
including sibling reads, status polls, and internal control flow, running.
The serialization the gate does impose is per effect at its single admission
mutex, measured in Section~\ref{sec:load}; it is not a serialization of the
agent's branches. Four non-claims: the gate does not prevent an
effect's \emph{submission} during a pause (submission is how a hold is
created); it does not make in-flight \emph{ungated} work atomic or revocable;
it does not roll back an effect already released and in flight at an
external service---an approved payment capture awaiting its HTTP response is
beyond any admission point's reach, the gate is an admission barrier, not a
transaction manager (compensation is the Atomix/saga territory of
Section~\ref{sec:related});
and it bounds no decision latency---liveness (every held effect is eventually
decided or fenced) belongs to the approver and the framework, not the
admission core. Likewise on the timeout axis the claim is effect-level only:
the gate guarantees a zombie's \emph{effect} cannot land after the stop is
observed (P4)---an in-flight \emph{mediated} tool may keep computing after
the cancel, but its submission meets the fence and its external action never
happens; it does not unblock a caller stuck past its own deadline and
cannot make a synchronous tool cancellable---those are the framework-design
divergences Table~\ref{tab:matrix} documents, not gaps the gate claims to
close. This effect-level reading is the measured predicate itself: B4 and the zombie probes are evaluated at the commit point (Section~\ref{sec:background}), which is exactly what the fence refuses---a zombie's computation may run to completion; its effect does not land. Fail-closed pausing under an unavailable approver is the
intended behavior for irreversible actions, not an oversight. A deployment
wanting bounded holds needs no new mechanism: because a hold is resolved by
an ordinary authenticated \textsc{Decide}, a decide-by-deadline policy is a
client-side watchdog that issues
\textsc{Decide}$(\mathit{run},\mathit{key},\mathit{approved}{=}\mathit{false})$
when a configured time-to-live expires---the held effect is then refused and
stays refused (P2), converting silent starvation into a visible, bounded
rejection. The watchdog composes above the verified core and changes none of
its state machine; safety is indifferent to \emph{who} issues a rejection,
only to its authenticity (HMAC, Section~\ref{sec:failure-model}). This
policy is executed, not sketched (artifact, \texttt{e2e\_ttl.py}, run with
the HMAC channel enabled): a held effect starves past a $0.5$\,s TTL and the
watchdog's authenticated rejection lands; the rejection is sticky against
both resubmission and a \emph{late valid approval}; an unauthenticated late
approval is refused at the channel; and a control effect approved inside its
TTL releases normally while the watchdog's late firing is a no-op
refusal---eight checks, all passing, with the starved effect never
executing.

\begin{remark}[The watchdog is inside the verified envelope]
\label{rem:watchdog}
Deploying the watchdog creates no new proof obligation, by the specification
rather than by fiat. In the shared TLA\textsuperscript{+} model
(\texttt{formal/tla/SoundGate.tla}), decisions are environment actions:
\textsf{Next} contains $\exists r,k,\mathit{ap}\in\mathrm{BOOLEAN}:
\textsc{Decide}(r,k,\mathit{ap})$, so the proved behaviors already include
\emph{every} pattern of decisions any client could issue. A watchdog firing
\textsc{Decide}$(\cdot,\cdot,\mathit{false})$ at a TTL selects a subset, and
safety proved over all behaviors holds over any subset; the TLC exhaustion and
TLAPS induction cover the watchdog-augmented deployment as-is. What it adds is a
liveness \emph{policy} (bounded decision latency), outside the verified safety
scope, and \texttt{e2e\_ttl.py} is executed evidence of it. Its two failure
modes degrade to the documented posture: if the watchdog crashes or partitions,
its rejection never arrives and the held effect starves---fail-closed pausing,
the stated liveness non-goal, safety untouched because a hold releases nothing;
and it cannot corrupt state, since its only capability is an authenticated
\textsc{Decide}$(\cdot,\cdot,\mathit{false})$, idempotent and sticky (P2), so a
delayed or replayed watchdog rejection is absorbed as any client's, and a fresh
post-TTL effect is a new identity with its own hold. The watchdog thus inherits
the gate's safety without new obligations; its own \emph{liveness} is a
deployment obligation like any approver's availability.
\end{remark}

\subsection{Implementation}
\label{sec:impl}
The gate core is a small, dependency-light Rust module; its admission logic is
deterministic and unit-tested (one test per property, regression tests
encoding the cross-run counterexamples above, and a randomized
invariant harness over interleaved submit/decide/cancel/close sequences).
The randomized harness is what surfaced a subtle admission bug during
development---re-submitting a held identity with approval disabled could
release it while leaving a stale pending entry that a later decision would
honor as a second release; the fix makes holds idempotent, and the property
``no identity is released twice'' now withstands thousands of randomized
sequences. A thin server exposes the
core over a line-delimited JSON protocol
(\texttt{submit}/\texttt{decide}/\texttt{cancel}), so any language can drive it:
the framework's tool-invocation wrapper issues a \texttt{submit} and performs the
effect only on \textsf{release}. Rust is used for the enforcement substrate---not
as a contribution in itself---because the gate must be a small, robust,
framework-independent process outside the untrusted control flow.

\subsection{Failure model and deployment obligations}
\label{sec:failure-model}
\noindent\textbf{Mechanized verification.} Because the admission core is small
and side-effect-free, we mechanically verify its four safety properties across
four tiers---three sharing one specification over a model of the core, and a
Loom tier on the deployed concurrent Rust (Table~\ref{tab:formal}; artifact,
\texttt{formal/}, \texttt{tests/loom\_gate\_test.rs}). On the word ``tiers'':
three are the \emph{same} abstract model at increasing scope (sequential proof,
finite-concurrent exhaustion, unbounded induction) and the fourth touches the
deployed code at bounded interleavings---a scope ladder over one model plus one
bounded check of the binary, not four independent verifications. In \emph{Verus}
we prove the properties over an abstract model whose transitions mirror the Rust
line-by-line (no rejected effect releases, no identity releases twice, a fenced
run never releases, a gated effect holds until approval), discharging every
obligation; since Verus establishes this for the sequential logic (which the
gate guarantees by serializing decisions under a mutex) and does not verify the
std-library collections, the model-to-code bridge is the differential harness
below plus the unit/regression/randomized suites. We then model-check the
\emph{concurrent} protocol in \emph{TLA+}/TLC over all interleavings of submit,
decide, cancel, close, exhausting the finite state space with no violation
($2{\times}2$: 729 states; $3{\times}3$: 804{,}357; $4{\times}3$: 74{,}805{,}201
from $4.2{\times}10^{9}$ generated, depth 17, ${\sim}29$~min on 16 threads; all
complete searches, the largest with a TLC-estimated $1.8{\times}10^{-3}$
fingerprint-collision probability). In \emph{TLAPS} we prove the invariants
\emph{inductive}, holding for unbounded runs and keys. TLC and TLAPS share one
specification: \texttt{Runs}/\texttt{Keys} are constants TLC instantiates
finitely and TLAPS leaves arbitrary; the three-conjunct invariant I1--I3
\emph{is} the inductive strengthening---I3 (fence compaction) exists so the
induction closes over \textsc{Close}---and P1--P4 are theorems from it, one
preservation lemma per action, so the finite check and the unbounded proof are
of literally the same \textsf{Next} and invariant. The load-bearing conjunct is
fence compaction---\emph{a closed run retains no per-identity state}, the formal
answer to the zombie-after-close race: closure is monotone and checked before
any per-identity lookup, so a submission bearing a closed run's identity refuses
via the fence regardless of arrival order, at any scale.

We claim no kernel-grade end-to-end refinement in the style of verified
OS kernels~\cite{sel4} or of end-to-end verified distributed systems and
file systems (IronFleet~\cite{ironfleet15}, Verdi~\cite{verdi15},
FSCQ~\cite{fscq15}), which would require Verus to model the standard
library; the verified surface is deliberately the small admission logic. (A bounded-model-checking attempt on the deployed \texttt{Gate} with Kani/CBMC is committed---\texttt{evidence/kani.txt}---and illustrates the cost: the search drowns in symbolic unwinding of the standard library's collections before the gate's own logic is exercised; the conformance harnesses below are the bridge we use instead.) We
close the model-to-code gap by execution instead: a differential conformance
harness (\texttt{tests/conformance.rs}) transcribes the Verus transitions into
executable form and drives \emph{both} the deployed \texttt{Gate} and that model
with identical randomized sequences, asserting identical verdict and derived
state (released, rejected, pending, cancelled, closed) after \emph{every}
operation. Over $1.2\times10^{7}$ operations across $2\times10^{5}$ traces on an
$8{\times}6$ identity domain---dense with the aliasing, re-submission, and
cross-run interleaving where a \texttt{HashMap} or ownership edge case would
surface---code and model never diverged; zero divergences over
$n{=}1.2\times10^{7}$ bounds the per-operation divergence probability
\emph{under the harness's distribution} below $3/n\approx2.5\times10^{-7}$ at
95\% confidence (rule of three)---the epistemic grade of testing, which is why
we do not call it a proof. We complement this with a \emph{bounded-exhaustive}
check on the same deployed \texttt{Gate}: over the two-run, two-key domain we
enumerate the entire reachable transition relation---all $729$ states (BFS-%
confirmed, the count TLC reaches independently, diameter $6$) and all $20$
operations from each, $14{,}580$ transitions---asserting identical verdict and
state on every one, zero divergences (\texttt{tests/exhaustive\_conformance.rs});
enumerating every sequence to depth five ($3.2\times10^{6}$ sequences) agrees
identically. This is the finite exhaustiveness TLC applies to the \emph{model},
here executed against the \emph{binary}: exhaustive at small scale, extended
(sampled) over the larger $8{\times}6$ domain by the randomized run---refinement
evidence by differential testing, the discipline long used to validate compilers
against oracles~\cite{mckeeman98,csmith11}, not a mechanized refinement proof.
The verification is not ceremony: escalating rigor found two real defects
testing had missed, and the same harness would catch future model/code drift. A
randomized invariant harness first exposed a double-release---resubmitting a
held identity with approval disabled released it while leaving a stale pending
entry a later decision would honor---fixed by making holds idempotent. Model
checking then found a second: a late rejection of an already-released identity
recorded a contradictory rejected entry (violating released/rejected
disjointness), which the Verus tier independently flagged at the same
transition. We fixed the gate---a released identity now refuses late decisions
of either polarity as duplicates---and pinned both fixes with regression tests
and harness invariants.

\smallskip\noindent\textbf{Mutation adequacy of the conformance harness.}
A harness that passes is reassuring only if it could have failed on a wrong
gate. We check this directly by applying five standard mutation operators to the
admission core---one per enforced property and code path: swapping the
release/reject verdict, negating the duplicate check, and weakening the
cancel/close fence (\texttt{||}\,$\to$\,\texttt{\&\&}) in each of \texttt{submit}
and \texttt{decide}. The bounded-exhaustive conformance harness catches all four
mutations whose effect is observable in the sequential semantics. The
fifth---weakening the fence on the \emph{decide-after-hold} path---is unreachable
sequentially, because both \texttt{cancel} and \texttt{close\_run} drop the run's
pending effects, so a later \texttt{decide} always takes the (unmutated)
no-pending fence rather than the post-hold one; it is reachable only when a
\texttt{decide} races a \texttt{cancel}, and the Loom tier below catches it.
Every property-violating mutation is thus caught by one of the two harnesses
(five of five across the two tiers---four by exhaustive conformance, the
concurrency-only fence by Loom)
(artifact, \texttt{scripts/mutation\_score.py},
\texttt{evidence/mutation\_score.txt}), evidence that the conformance suite
exercises the properties it checks rather than passing vacuously---bounded, as
any mutation analysis is, by the operator set: the five operators target the
enforced properties' decision points, not every conceivable semantic
deviation---and that the
sequential and concurrent tiers cover complementary reachable states.

\smallskip\noindent\textbf{Model-checking the concurrent Rust (Loom).}
The differential harness cross-checks code against model over random
\emph{sequential} sequences; to check the actual \emph{concurrent} Rust we add
a fourth tier using Loom~\cite{loom}, which exhaustively explores the legal
thread interleavings of a small program under the C11 memory model. We use
Loom's exhaustive search rather than a randomized concurrency tester such as
Shuttle~\cite{shuttle} precisely because the admission core's shared-state
surface is small enough to enumerate---two threads contending on one
mutex---so exhaustiveness is affordable and strictly stronger than sampling
here; a randomized explorer would be the tool of choice only at a larger
interleaving space than the gate presents. Because
the admission core is plain \texttt{\&mut self} over standard collections with
no internal locks, its only shared-state access is through the mutex the
server serializes on; Loom drives the \emph{deployed} \texttt{Gate} under two
threads contending for that mutex and asserts, after every interleaving, the
same three-part invariant the abstract tiers target plus the verdict-level
facts (exactly one of two racing decisions of a held identity returns
\textsf{release}, the other \textsf{refused-duplicate}; a cancel concurrent
with a zombie resubmission always fences; cross-run key reuse never
collaterally refuses). Three models pass with zero failures
(artifact, \texttt{tests/loom\_gate\_test.rs}, \texttt{evidence/loom.txt}),
raising the concurrent evidence from the abstract \textsf{Next} to the running
mutex-guarded type. This is bounded model checking of the real code, not a
mechanized refinement proof of it---the gap the abstract names---but it closes
the specific concern that TLC checks a model while the Rust concurrency goes
unchecked.
We treat liveness (every held effect is eventually
decided or fenced) as out of scope for the gate, since it depends on the
approver and framework rather than the admission core.

\smallskip\noindent\textbf{Residual trusted computing base.} What remains
trusted, because unverified, is enumerated: the Rust compiler and standard
library (the \texttt{HashMap}/\texttt{HashSet} collections and the serializing
mutex), the JSON layer, the TCP stack and OS scheduler, the filesystem's
\texttt{fsync} contract in WAL mode (and the device write-cache flush beneath
it---the durability contract every WAL store inherits), the HMAC
implementation (RFC~4231-pinned), and the conformance harness. The
differential harness targets the first two (model/code divergence through a
collection or ownership edge case) and the torn-tail discipline the
\texttt{fsync} item; the rest we inherit as every user-space reference monitor
does. What matters is not that these are correct but that their failures are
\emph{safety-transparent}: safety is over \textsf{release} verdicts, and a
fault in JSON decoding, the socket, or the scheduler yields a malformed,
dropped, or delayed request---no \textsf{release}, not a spurious one. Memory
corruption inside the gate's own code is absent by construction: the source
tree contains no \texttt{unsafe} block, so it would require a compiler or
std-library fault, already enumerated. Two silent-corruption exceptions are named rather than averaged over: a serialization bug corrupting a verdict field (why the reply path is a fixed, checked enumeration), and a collection returning wrong membership---a wrong \emph{verdict} the dense-aliasing harness and bounded-exhaustive enumeration are built to surface, since any divergence is verdict- or state-visible at the next compared step. That narrows, not eliminates, the risk of a fault the harness
distribution never exercises---the honest boundary of testing. The transport
half is tested too: a generative fuzz harness
(\texttt{scripts/fuzz\_boundary.py}) drives the live server with
$1.8{\times}10^{5}$ malformed inputs across eight classes---random bytes,
invalid UTF-8, non-JSON text, wrong-shape JSON, wrong-typed/unknown fields,
hostile \texttt{decide} values (forged MACs, oversized and control-character
identities, duplicate keys), protocol-state abuse (decide-before-submit,
approve-after-cancel, double decide), and framing abuse (half, unterminated,
oversized lines)---asserting continuously that no input elicits a
\textsf{release} (zero fail-open observed), planted canary state (held,
released, fenced) survives every batch, and a valid round-trip answers after
every class. All pass with the server alive throughout
(\texttt{evidence/fuzz\_boundary.txt}). The untrusted transport can thus make
the gate fail \emph{closed} (its designed posture) but cannot, without a
harness-visible divergence, make it fail open; a refinement proof down to the
socket would upgrade this from fail-closed-by-construction to
proven-equivalent, the natural next step (Section~\ref{sec:limitations}).

\begin{figure}[!t]
\centering
\footnotesize
\begin{tikzpicture}[
  node distance=2.2mm,
  box/.style={draw, rounded corners=1.5pt, align=center, inner sep=3pt,
              font=\footnotesize},
  tier/.style={draw, rounded corners=1.5pt, align=center, inner sep=3pt,
               font=\scriptsize, minimum height=8mm},
  >={Stealth[length=2mm]}]
  \node[box, fill=black!4] (rust) {Deployed Rust gate\\\texttt{lib.rs} (in the TCB)};
  \node[tier, right=13mm of rust] (verus) {Verus\\\emph{sequential}\\11 items};
  \node[tier, right=of verus] (tlc) {TLA+/TLC\\\emph{concurrent,}\\\emph{finite}\\$7.5{\times}10^{7}$ st.};
  \node[tier, right=of tlc] (tlaps) {TLAPS\\\emph{unbounded}\\68/68};
  \draw[<->, thick] (rust) -- node[above, font=\scriptsize, align=center, inner sep=1pt]
       {diff.\ conf.} node[below, font=\scriptsize, align=center, inner sep=1pt]
       {$1.2{\times}10^{7}$~ops, 0 div.} (verus);
  \draw[->] (verus) -- (tlc);
  \draw[->] (tlc) -- (tlaps);
  \node[draw, dashed, rounded corners, fit=(verus)(tlc)(tlaps),
        inner sep=4pt, label={[font=\scriptsize]below:one shared \textsf{Next}
        + one invariant $I1{-}I3$}] (model) {};
\end{tikzpicture}
\caption{What is proved, and how it connects to what runs. The three model
tiers verify a \emph{model} of the admission core (shared specification and
invariant, increasing scope left to right); a differential conformance harness
bridges that model to the deployed Rust, and a Loom tier (not shown)
model-checks the deployed concurrent code directly---all \emph{evidence}, not
a mechanized refinement proof. The transport, standard library, and OS remain in the
trusted computing base.}
\label{fig:verif-pipeline}
\end{figure}
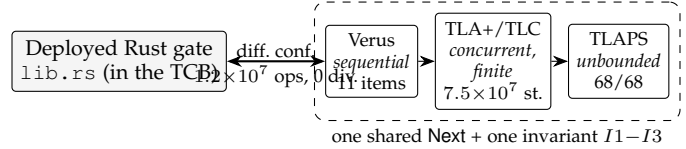

\begin{table}[!t]
\centering
\caption{Mechanized verification of the admission core. Each tier
proves what the one above cannot; the first three target the same three-part safety
invariant (released/rejected disjoint; a pending identity is undecided; a
closed run retains no per-identity state) over a shared model, and Loom checks
the same invariant on the deployed concurrent Rust.}
\label{tab:formal}
\renewcommand{\arraystretch}{1.2}
\footnotesize
\begin{tabularx}{\columnwidth}{@{}l X p{0.30\columnwidth}@{}}
\toprule
\textbf{Tier} & \textbf{Establishes} & \textbf{Result}\\
\midrule
Verus & safety properties over the sequential admission model mirroring
\texttt{lib.rs} & 11 items verified, 0 errors\\% receipt committed: formal/verus.txt (11 verified, 0 errors)
TLA+/TLC & same invariants under all concurrent interleavings (finite,
exhaustive) & $2{\times}2$: 729; $3{\times}3$: 804{,}357; $4{\times}3$:
74{,}805{,}201 distinct states; 0 violations\\% receipts: formal/tla/tlc_{2x2,3x3,4x3}.txt
TLAPS & invariants inductive, hence at unbounded scale &
68/68 obligations proved\\% receipt: formal/tlapm.txt (All 68 obligations proved)
Loom & same invariant on the \emph{deployed} concurrent Rust (bounded
interleaving search) & 3 models, 0 failures\\% receipt: soundgate/evidence/loom.txt
\bottomrule
\end{tabularx}
\end{table}

\medskip
\noindent The reference gate holds its state in memory. This suffices to
demonstrate the semantics and drive the end-to-end replay, but a deployment
inherits four obligations, enumerated here rather than hidden. \emph{Durability
(implemented):} the reference server optionally appends state-changing
verdicts to a write-ahead log and fsyncs \emph{before} acknowledging, so no
acknowledged release can be forgotten; recovery replays the log before the
listener opens (fail-closed: nothing is admitted against unrestored state),
and a crash--restart scenario in the evaluation shows the replay and
cancellation fences surviving an uncleanly killed process. Recovery
distinguishes the two corruption cases: a torn \emph{final} record---the only
artifact a crash mid-append can leave, since the fsync precedes the
acknowledgement---is skipped with a warning, while an unparsable record
anywhere earlier is mid-log corruption and aborts startup rather than opening
with partial fences. The operational cost of that choice is stated rather
than hidden: mid-log corruption halts admission---fail-closed---until the
operator repairs or truncates the log, in preference to a gate that has
silently forgotten releases or fences; the WAL is per-instance and bounded
by compaction, so the blast radius is one gate's active runs, and per-run
sharding (below) confines it further. State-changing throughput in WAL mode is consequently
fsync-bound (Section~\ref{sec:load}); the reference server's WAL writer implements \emph{group
commit}---the classical WAL batching of the database
literature~\cite{dewitt84}---by default (a single writer thread batches up to 512
events behind one fsync; each reply is written only after the fsync covering
its event, preserving the discipline exactly), evaluated in
Section~\ref{sec:load}; the unbatched per-operation variant that produced
Table~\ref{tab:load}'s WAL rows is pinned in the artifact so both modes'
provenance is exact.

\begin{lemma}[Group-commit crash safety]
\label{lem:gc}
Under the WAL writer's discipline---each reply is written only after the
\texttt{fsync} whose batch covers that event's record
(\texttt{src/main.rs}, \texttt{wal\_writer})---a crash at any point yields
a durable log of which the acknowledged event set is a prefix; recovery
therefore never forgets an acknowledged release or fence, and batching
introduces no new corruption case beyond the unbatched server's torn final
record.
\end{lemma}
\begin{IEEEproof}[Proof sketch]
Acknowledged $\Rightarrow$ durable: an event is acknowledged only after the
\texttt{fsync} covering its record returns, so every acknowledged event is
on stable storage at acknowledgement time, and the acknowledged set is a
prefix of the durable log (events are appended and fsynced in arrival
order). A crash between an \texttt{fsync} and some of its batch's replies
loses only \emph{acknowledgements}, not durability: the events are on disk,
their clients time out undecided and resubmit, and the protocol's
idempotence makes the resubmission re-hold or refuse as a
duplicate---never a second release. A crash mid-append can tear at most
the final record (the only record not yet covered by a returned
\texttt{fsync}), which recovery already skips; a torn record is by
construction unacknowledged, so skipping it forgets nothing a client was
told. No interleaving fabricates an acknowledged event that is not
durable, which is the failure that would break P2--P4 across restart.
\end{IEEEproof}

\noindent Crash atomicity under batching is thus the same discipline as the
unbatched server's, stated once as the lemma above rather than re-argued
per mode. Recovery is linear in the surviving log: replaying the
accounting WAL of Section~\ref{sec:load} completes in under two seconds
before the listener opens, and \textsc{CloseRun} compaction bounds the log in
steady state.
Linearity holds at ten times that scale: a synthesized $1{,}706{,}000$-event,
$107$\,MB log in the exact durable-event format ($1.7$\,M releases across
$2{,}000$ runs plus $6{,}000$ fences) replays on a second container in
$1.7$\,s ($3.3$\,s cold; page-cache dominates), after which a replayed identity
refuses as duplicate, a fenced run's late submission refuses, and fresh work
releases (\texttt{e2e\_recovery10x.py}). Held-but-undecided effects are
deliberately not durable: losing a hold is conservative, and a resubmission
re-holds.
\emph{Reachability:} if the gate is unreachable the wrapper must refuse rather
than proceed---failing open would convert an availability incident into a
bypass; gate downtime pausing effects is the correct trade for irreversible
actions. Executed (\texttt{e2e\_partition.py}): with the gate never started the
submission refuses and the effect body does not run; a \texttt{SIGKILL}
mid-session turns the next submission into a refusal; and on WAL restart fresh
work releases while the pre-kill release's replay still refuses as a
duplicate---four checks, all from the client's side of the partition.
\emph{Bounded state (implemented):} released and rejected identities would
otherwise grow monotonically; a \textsc{CloseRun} operation marks a run
terminal and compacts its per-identity state, so steady-state memory is bounded
by \emph{active} runs rather than total historical effects. Compaction is
sound because the run-level fence then refuses any late submission from the
closed run, so a replayed effect cannot slip through as a fresh release.
Because every piece of gate state is keyed by run---identities by
$(\mathit{run},\mathit{key})$, fences by run---runs never share state, so a
deployment can shard runs across independent gate instances with zero
coordination while preserving all four properties; single-gate serialization
is a per-run necessity, not a scaling ceiling. \emph{Decision authenticity:}
\textsc{Decide} must bind to an authenticated approver channel or a network
attacker reaching the gate could forge an approval. The reference gate requires
each decision to carry
$\mathrm{HMAC\text{-}SHA256}(\mathit{secret},\,\mathit{run}\,\|\,\mathit{key}\,\|\,\mathit{approved})$,
verified in constant time \emph{before} any state change, so a forged or absent
tag refuses without touching state; an attack test (\texttt{e2e\_auth.py})
confirms forged approve and forged reject both refuse while the secret-holder's
valid tag releases, and the HMAC passes the RFC~4231 vector. The channel
authenticates the \emph{decision}, not the approver's integrity: a compromised
approver (subverted UI, stolen secret, coerced operator) can issue a
\emph{valid} approval, as in any human-in-the-loop system---the same residual
trust every approval mechanism carries, out of scope here alongside the
Byzantine submitter. Key distribution and secret rotation are ordinary secret
management; multi-tenant isolation reduces to a tenant identifier in
$\mathit{run}$, which per-run sharding already respects.
\emph{Rejection absorption:} a
mediated wrapper returns the gate's verdict as an ordinary tool result and acts
only on \textsf{release}, so a \textsc{Refused-\{Duplicate,Cancelled,Rejected\}}
is a normal ``effect not performed'' value, not a raised exception; the
executed integrations (Section~\ref{sec:eval}, \texttt{e2e/e2e\_langgraph.py})
resume as if the tool were a no-op, none crash-looping across five
verdict-identical repetitions per framework. A deployment preferring a hard
failure can raise inside its own wrapper; either way the effect stays
un-externalized. The gate does not repair a framework whose \emph{own} error
model mishandles a tool exception---orthogonal to admission.
\emph{Protocol assumptions:} the reference transport is
per-connection TCP (ordered, non-duplicating per stream); across connections
the gate's mutex serializes operations, exactly the interleavings the TLC and
TLAPS tiers model. Idempotence makes message-level duplication or replay safe
(a resubmission re-holds or refuses; a repeated approval finds the identity
released and refuses; a repeated rejection is sticky). Byzantine
\emph{submitters} stay out of scope (a compromised wrapper can decline to
submit---the complete-mediation contract); decision forgery is defeated by the
authenticated channel above. The core holds no lock while waiting on a client
and a hold is pure state, so the gate cannot deadlock; held-effect starvation
is the stated liveness non-goal. The admission log doubles as a causally
ordered audit trail. These obligations are tracked
in the artifact and are orthogonal to the semantics evaluated in
Section~\ref{sec:eval}.

% =====================================================================
\begin{table*}[!tp]
\centering
\caption{Executed-experiment inventory. Every row corresponds to a committed
artifact log; \sys{} results are reproduced from those logs. Model-exposure
rows report exposures over runs that called the gated tool
(Table~\ref{tab:exposure}); framework rows report violation classes repaired
in situ. ``Native/direct'' marks a second serving path on the vendor's own
API. FW-A: LangGraph; FW-C: Microsoft Agent Framework; FW-D: OpenAI Agents
SDK.}
\label{tab:inventory}
\renewcommand{\arraystretch}{1.15}
\footnotesize
\begin{tabularx}{\textwidth}{@{}p{0.25\textwidth}p{0.095\textwidth}X>{\scriptsize}p{0.235\textwidth}@{}}
\toprule
\textbf{Experiment} & \textbf{Scale} & \textbf{Result} & \textbf{Receipt}\\
\midrule
\multicolumn{4}{@{}l}{\emph{Framework control-primitive measurement and repair (Sec.~\ref{sec:measurement},~\ref{sec:eval})}}\\
Cross-framework probe matrix & 6 frameworks & sibling leak + replay/timeout divergences reproduced & \texttt{probes/}, Table~\ref{tab:matrix}\\
\quad Randomized structural sweep (FW-A) & 1{,}000 graphs & conc-same 577/577 leak; later-wave 0/331; descendant 0/363 & \seqsplit{\texttt{randgraph/results\_fwa.jsonl}}\\
\quad FW-A end-to-end repair & 3 classes & 3/3 repaired in situ (real \textsf{langgraph}) & \texttt{e2e\_langgraph.py}\\
\quad FW-B end-to-end repair & 2 classes & 2/2 repaired (replay natively clean) & \texttt{e2e\_llamaindex.py}\\
\quad FW-C end-to-end repair & 3 classes & 3/3 repaired (replay natively clean) & \texttt{e2e\_msaf.txt}\\
\quad FW-D end-to-end repair & 2 classes & 2/2 repaired (scripted model) & \texttt{e2e\_openai\_agents.txt}\\
\quad FW-F end-to-end repair & 3 classes & 3/3 repaired in situ (real \textsf{@langchain/langgraph}, Node) & \seqsplit{\texttt{e2e\_langgraph\_js.txt}}\\
\quad FW-E cancel/timeout repair & 2 classes & 2/2 violated axes fenced in situ (real \textsf{crewai}) & \texttt{e2e\_crewai.txt}\\
Decision-channel authentication & --- & forged decisions refused; only secret-holder releases & \texttt{e2e\_auth.txt}\\
Bounded-hold TTL watchdog & 8 checks & 8/8 passed (HMAC watchdog, $\text{TTL}{=}0.5$\,s) & \texttt{e2e\_ttl.txt}\\
Structural mediation $\times$ frameworks & 5 frameworks & full violated-axis sets pass inside loopback-only netns; unwrapped egress refused by kernel & \makecell[tl]{\texttt{e2e\_structural\_all.txt}\\\texttt{e2e\_structural\_langgraph.txt}}\\
Non-network confinement (seccomp $+$ Landlock) & 3 channels $+$ path arm & unconfined $3/3$ escape, confined $0/3$ (fs/AF\_UNIX/\texttt{memfd}); Landlock ABI~4 workdir write allowed, shared/outside writes kernel-refused & \makecell[tl]{\texttt{confine\_fs.txt}\\\texttt{landlock\_workdir.txt}}\\
Adversarial injection (4 models, \textsf{FW-A}) & 400 live runs & 121/121 emitting runs leak unmediated; 0/400 mediated & \makecell[tl]{\texttt{e2e\_injection\_langgraph.txt}\\\texttt{injection\_live\_runs.txt}}\\
Experiment A: live end-to-end leak (Sec.~\ref{sec:expA}) & 4 arms, $2{\times}1{,}200$ runs & 215/1{,}200 unmediated leaks; 0/1{,}200 mediated; $P(\mathrm{leak}{\mid}\mathrm{em.}){=}1.00$ per emitting arm & \texttt{expA\_*.jsonl}\\
Real-effect demonstration (Sec.~\ref{sec:expA}) & \textsf{FW-A} fan-out $+$ \texttt{interrupt()}, real \texttt{POST} & 1 \texttt{POST} at endpoint $512$\,ms pre-reject; mediated $0$, \texttt{refused\_rejected} & \texttt{webhook\_leak\_demo.txt}\\
Experiment A on \textsf{FW-B} (second runtime) & \textsf{gpt-4o}, 300 runs & 121/300 unmediated leaks; 0/300 mediated; $P(\mathrm{leak}{\mid}\mathrm{em.}){=}1.00$; different scheduler & \texttt{expA\_fwb\_*.jsonl}\\
Experiment A on \textsf{FW-F} (non-Python: \textsf{LangGraph.js}/Node) & \textsf{gpt-4o}, 300 runs & 143/300 unmediated leaks; 0/300 mediated; $P(\mathrm{leak}{\mid}\mathrm{em.}){=}143/143$; live cross-language arm & \seqsplit{\texttt{expA\_fwf\_openai\_gpt4o.jsonl}}\\
Approval-pause sweep (keyless, \textsf{FW-A}$+$\textsf{FW-B}) & 5 pauses $\times$ 20 & unmediated leak pause-invariant on both; mediated 0 at every pause & \makecell[tl]{\texttt{pause\_sweep.txt}\\\texttt{expA\_fwb\_mock\_pause*.jsonl}}\\
\midrule
\multicolumn{4}{@{}l}{\emph{Durable-execution contrast (Sec.~\ref{sec:temporal})}}\\
Temporal probe battery (server 1.31.2) & 4 predicates & replay clean ($1{\to}1$ across forced full-history replay); sibling, no-heartbeat cancel, and timeout behavioral bits reproduce; heartbeating cancel clean & \seqsplit{\texttt{temporal\_probes.txt}}\\
Gate composition on Temporal & repaired T1 twin & held during live Signal pause (0 effects); reject sticky; zombie resubmit refused; control released & \seqsplit{\texttt{temporal\_gated.txt}}\\
\midrule
\multicolumn{4}{@{}l}{\emph{Model-exposure study, 10 tasks (Sec.~\ref{sec:exposure}, Table~\ref{tab:exposure})}}\\
GPT-4o (OpenRouter) & 1000 & exposure 0.14 [0.12, 0.17] & \texttt{or\_gpt4o.jsonl}\\
Claude Sonnet 4.6 (OpenRouter) & 1000 & exposure 0.04 [0.03, 0.05] & \texttt{or\_claude.jsonl}\\
GPT-4o, Claude (native pilot) & $2{\times}250$ & anchors the OpenRouter rates & \seqsplit{\texttt{exposure\_\{openai,anthropic\}\_*.jsonl}}\\
Gemini 2.5 Flash (OpenRouter) & 1000 & exposure 0.00 [0.00, 0.00] (confounded) & \texttt{or\_gemini.jsonl}\\
\quad Gemini 3.5 Flash (native) & 1000 & exposure 0.02 [0.01, 0.03]; de-confounds & \seqsplit{\texttt{exposure\_gemini\_native\_*\_n100.jsonl}}\\
DeepSeek V3.2 (OpenRouter) & 1000 & exposure 0.00 [0.00, 0.01] (confounded) & \texttt{or\_deepseek.jsonl}\\
\quad DeepSeek V4-flash (native) & 1000 & exposure 0.00 [0.00, 0.00]; sequential disposition & \seqsplit{\texttt{exposure\_deepseek\_\-v4flash\_native\_n100.jsonl}}\\
Llama 3.3 70B (OpenRouter) & 1000 & exposure 0.00 [0.00, 0.01] (confounded) & \texttt{or\_llama.jsonl}\\
\quad Llama 3.3 70B (Together, provider-direct) & 883 & exposure 0.08 [0.07, 0.10]; 0.90 on cleanup; de-confounds & \seqsplit{\texttt{exposure\_llama\_\-together\_3.3-70b\_n1000.jsonl}}\\
Naturalistic tasks ($\tau$-bench retail/airline; Sec.~\ref{sec:limitations}) & 431 tool turns, 2 models & 0/71 gated batches co-emit a benign read; GPT-4o batches consequential writes in 9/52 & \seqsplit{\texttt{taubench\_exposure\_*.jsonl}}\\
\quad Multi-effect prevalence ($\tau$-bench gold solutions) & 165 tasks & retail 45/115, airline 15/50 require $\geq$2 consequential writes & \texttt{prevalence/}\\
\midrule
\multicolumn{4}{@{}l}{\emph{Formal verification (Sec.~\ref{sec:impl}, Table~\ref{tab:formal})}}\\
Verus (sequential model) & 11 items & 11 verified, 0 errors & \texttt{formal/verus.txt}\\
TLA+/TLC (concurrent, exhaustive) & $2{\times}2$, $3{\times}3$, $4{\times}3$ & 729 / 804{,}357 / 74{,}805{,}201 distinct states, 0 violations & \seqsplit{\texttt{tlc\_\{2x2,3x3,4x3\}.txt}}\\
TLAPS (unbounded induction) & 68 obligations & 68/68 proved & \texttt{tlapm.txt}\\
Loom (concurrent Rust) & 3 models & 0 failures on the deployed \texttt{Gate} & \texttt{evidence/loom.txt}\\
Differential conformance & $1.2{\times}10^{7}$ ops & model $\equiv$ code; surfaced 2 defects & \texttt{conformance.txt}\\
Exhaustive conformance ($2{\times}2$) & 729 states $\times$ 20 & complete reachable transition relation; model $\equiv$ code, 0 divergences & \seqsplit{\texttt{exhaustive\_conformance.txt}}\\
Mutation adequacy (Sec.~\ref{sec:failure-model}) & 5 operators, 2 harnesses & 4/4 sequential mutants caught by exhaustive conformance; decide-after-hold fence (concurrency-only) caught by Loom & \seqsplit{\texttt{mutation\_score.txt}}\\
Protocol-boundary fuzz (Sec.~\ref{sec:failure-model}) & $1.8{\times}10^{5}$ inputs, 8 classes & 0 fail-open verdicts; canary state intact; server alive throughout & \texttt{fuzz\_boundary.txt}\\
\midrule
\multicolumn{4}{@{}l}{\emph{Durability, recovery, and load (Sec.~\ref{sec:load})}}\\
Crash-recovery replay & $1.71{\times}10^{6}$ events & 107\,MB replayed in 1.7\,s; 3/3 state checks & \texttt{e2e\_recovery10x.txt}\\
Fail-closed under partition/kill & 4 checks & 4/4 (down, SIGKILL mid-session, WAL restart) & \texttt{e2e\_partition.txt}\\
Replicated failover (3-voter Raft) & 1 leader kill & new leader in 1.9\,s; pre-crash fence refuses duplicate on new leader; fresh admits & \texttt{raft\_failover.txt}\\
Completion A/B ($\tau$-bench retail) & 20 episodes $\times$ 2 arms & 0/20 fail-closed on gated writes; ${\sim}1$\,ms/write; latency unchanged (37.9$\to$37.3\,s) & \seqsplit{\texttt{p6\_completion\_retail\_gpt4o.jsonl}}\\
Replicated throughput (3-voter Raft) & $C{=}1$--$32$ & 1{,}019--1{,}829 adm/s; exceeds WAL-GC at low $C$; followers no-lag & \texttt{concurrent\_raft.txt}, \texttt{pipelined\_raft.txt}\\
Emulated-WAN Raft sweep (\texttt{netem}) & RTT $0$/$10$\,ms, $C{=}1$--$8$ & $2{,}244{\to}26$ adm/s at $C{=}1$; p50 $0.32{\to}21.5$\,ms ${\approx}2{\times}$RTT per commit & \texttt{netem\_raft.txt}\\
Throughput (mem / WAL / WAL-GC) & $C{=}1$--$128$ & Table~\ref{tab:load}; WAL-GC $11{,}984$ adm/s at $C{=}128$ & \seqsplit{\texttt{concurrent\_mem.txt}}\\
\bottomrule
\end{tabularx}
\end{table*}

\section{Evaluation}
\label{sec:eval}

Before the detailed results, Table~\ref{tab:inventory} inventories every
experiment executed for this paper, each with the artifact receipt that backs
it, so a reader can map any quantitative claim to the log that produced it.
All numbers below are recomputed from these raw logs, not transcribed from
prior runs.

\subsection{Does the gate suppress the violations?}
\noindent When we say the gate
``closes'' or ``repairs'' a violation we mean it \emph{suppresses the mediated
effect}---holds, fences, or dedups it so it never externalizes undecided---not
that it eliminates the framework's underlying race, which continues exactly as
Section~\ref{sec:measurement} measures, now moving a held token rather than an
executed effect. The claim is always over mediated effects under the stated
contract.
We replay the \textsf{FW-A} violation scenarios---plus the cross-run key-reuse
counterexample---routing every side effect through the live \sys{} process over
its socket interface. Table~\ref{tab:e2e} reports the executed result: every
violation class is blocked, cross-run key reuse is correctly scoped rather than
collaterally refused, and a legitimate, unique, approved effect on a live run
is released---so the gate is not simply refusing everything. The
false-positive question---does the gate ever hold or refuse a legitimate
effect it should release?---is answered by the same executed evidence rather
than asserted: every approve-path scenario across the framework integrations
released its effect in five verdict-identical repetitions each
(\texttt{evidence/}); the concurrent-load sweeps of Section~\ref{sec:load}
run entirely on the legitimate unique-release path, where every submission
released and the durable runs' WAL accounting confirms all
acknowledged admissions persisted with none spuriously refused; P3's
cross-run scoping is checked explicitly (a reused key on a \emph{different}
run releases), and the Loom tier asserts under every explored interleaving
that cross-run key reuse never collaterally refuses
(Section~\ref{sec:failure-model}). A refusal, when it does occur, is by
construction one of the three named verdicts a mediated wrapper receives as
an ordinary result (Section~\ref{sec:failure-model}), not a silent drop. A sixth scenario
kills the WAL-backed gate mid-run (\texttt{SIGKILL}) and restarts it: the
replay and cancellation fences survive while fresh work still flows.
Finally, an integration demonstration wires the gate into a \emph{real}
\textsf{FW-A} agent through a twenty-line tool wrapper: during a genuine
\texttt{interrupt()} pause the sibling's mediated effect is held and a
rejection prevents it entirely; on resume the node body re-executes exactly
as documented, but the wrapper's resubmission is refused as a duplicate, so
the effect runs exactly once; and a worker thread surviving cancellation is
\begin{table}[!t]
\centering
\caption{Integration cost per framework, measured from the executed harnesses.
The mediation \emph{wrapper} (a socket \texttt{GateClient}: \texttt{submit},
perform-on-\texttt{release}, \texttt{decide}, \texttt{cancel}) is
framework-invariant---22 lines of Python, byte-identical for \textsf{FW-B}/%
\textsf{FW-D} and cosmetically different (type hints, a default argument) for
\textsf{FW-A}; the \textsf{FW-F} wrapper exposes the identical four-call
surface in 42 lines of JavaScript, the delta being the explicit socket
line-buffering Node requires. What varies per framework is only the
\emph{injection point}: where a consequential tool call is routed through
\texttt{gate.mediated\_effect(\dots)} inside that runtime's tool model. No
framework's source is modified.}
\label{tab:integration}
\renewcommand{\arraystretch}{1.2}
\footnotesize
\begin{tabularx}{\columnwidth}{@{}llcX@{}}
\toprule
\textbf{Fw.} & \textbf{Execution model} & \textbf{Wrapper} & \textbf{Injection point}\\
 & & \textbf{LOC} & \\
\midrule
FW-A & Pregel/BSP nodes    & 22 & effect call inside a graph node\\
FW-B & event workflow      & 22 & tool function body\\
FW-C & message-passing     & 22 & message handler\\
FW-D & parallel tool calls & 22 & SDK tool callable\\
FW-F & Pregel/BSP (Node)   & 42 & effect call inside a graph node\\
\bottomrule
\end{tabularx}
\end{table}

fenced. Complete mediation required no framework-source modification---the
change is the deployment's tool definitions at the injection point,
Table~\ref{tab:integration} (artifact:
\texttt{e2e\_langgraph.py}). A second integration
(\texttt{e2e\_llamaindex.py}) repairs \textsf{FW-B}'s independently
designed \emph{event-driven} runtime through the same ${\sim}$20-line
wrapper shape: the mediated sibling is held during the
\texttt{InputRequiredEvent} pause and a rejection yields zero effects, and
a worker thread surviving the native \texttt{cancel\_run()} is fenced
(2/2 of \textsf{FW-B}'s violated axes; replay is natively clean there, so
there is nothing to repair on that axis). Two further integrations complete the execution-model set.
\texttt{e2e\_msaf.py} repairs \textsf{FW-C}'s message-passing fan-out
runtime (3/3 violated axes: the mediated sibling is held during the
\texttt{request\_info} pause and a rejection yields zero effects; workers
surviving a host-level cancel and a host-level deadline are both fenced;
replay is natively clean, including across a checkpoint restore).
\texttt{e2e\_openai\_agents.py} repairs \textsf{FW-D}'s parallel tool calls
within a single model turn, driven by the probe suite's scripted model
(keyless; 2/2 violated axes: the mediated sibling emitted alongside the
approval-gated call is held while the run pauses with
\texttt{interruptions}, and rejecting both yields zero effects; a sync tool
on the SDK's default worker-thread path surviving cancellation is fenced;
replay is natively clean and the native sync-tool timeout is refused at
construction). A fifth integration carries the repair across the language-runtime
boundary. \texttt{e2e\_langgraph\_js.mjs} (artifact, \texttt{probes-js/})
drives the real \textsf{FW-F} runtime (\texttt{@langchain/langgraph} 1.4.7
on Node 22) with the repaired twins of its three measured violations: the
mediated sibling emitted alongside a genuine \texttt{interrupt()} pause is
held and a rejection yields zero effects; on resume the node body
re-executes---the measured replay---and the wrapper's resubmission refuses
as a duplicate, so the effect runs exactly once; and the \textsf{AbortSignal}
orphan, the one cancellation violation JavaScript's uninterruptible promises
make unavoidable at the runtime level, is fenced---the orphaned promise's
late submission meets \texttt{refused\_cancelled} with zero effects (3/3,
five repetitions verdict-identical; transcript
\texttt{evidence/e2e\_langgraph\_js.txt}). The gate binary, protocol, and
wrapper surface are byte-for-byte the ones the Python integrations use; only
the client language changes, which is the point---the same external barrier
repairs the same design's leak in both of its runtimes.
The four Python integrations use the same ${\sim}$20-line wrapper
shape, required no framework-source modification (the change is confined to
the deployment's tool definitions at each runtime's injection point,
Table~\ref{tab:integration}), and reproduced verdict-identically
across five repetitions each in that environment
(Table~\ref{tab:integration} quantifies the per-framework cost: the Python wrapper is
22 lines and framework-invariant, and only the injection point differs; the
JavaScript wrapper exposes the identical four-call surface in 42 lines, the
difference being the explicit socket line-buffering Node requires). So that
the 22-line figure is not read as the whole integration, we account for the
rest from the committed harnesses themselves: each executed integration
totals 198--236 lines \emph{including} its scenario drivers and assertions,
and within each, the lines that name the framework's own surface---imports
plus the node, handler, or tool definitions that place
\texttt{mediated\_effect} at that runtime's injection point---number roughly
17--33; the remainder is scenario setup shared in shape across frameworks.
Refusal handling costs no code beyond the wrapper's perform-on-release check
(a refusal returns as an ordinary ``effect not performed'' result;
Section~\ref{sec:failure-model}), and decision authentication is one HMAC
argument on the decide call. What no table amortizes is the mediation
contract itself: a production deployment must route \emph{every}
side-effecting tool through the wrapper, and locating those tools across a
real codebase is the true integration cost---the reason the mediation linter
and the structural routes exist (Section~\ref{sec:design}). The repair is thus
demonstrated end-to-end on all six frameworks, covering all four execution
models and both language runtimes: the five with a pre-execution gate on
their full violated-axis set, and \textsf{FW-E}---which ships no such gate,
so has no approval barrier to repair---on its two approval-independent
violations. \texttt{e2e\_crewai.py} drives the real \textsf{FW-E} crew
(\texttt{crewai} 1.15.1, the probe's own scripted LLM) with a gate-mediated
effect tool: the async cancellation orphan (probe C3) and the native
\texttt{max\_execution\_time} zombie (probe C4) are both fenced---the tool's
late effect submission meets \texttt{refused\_cancelled} and never
executes---with the timeout arm additionally reproducing \textsf{FW-E}'s
note-c behavior, the caller blocked past its deadline while the fenced
effect is nonetheless prevented (2/2 violated axes, five repetitions
verdict-identical; transcript \texttt{evidence/e2e\_crewai.txt}).
\textsf{FW-E} has no sibling-leak or
replay cell to repair, so ``all six'' refers to every framework's every
\emph{violated} axis, not to a uniform three-class repair. The six-scenario replay and the first two
integration demonstrations executed identically on the additional container
environment during revision (artifact, \texttt{evidence/}). These
demonstrations drive the frameworks with deterministic and keyless probes
\emph{by design}: the property under test is the gate's enforcement given the
leak-triggering plan shape, and a deterministic driver isolates that
enforcement from model stochasticity, so a failure is unambiguously the
gate's rather than a sampling artifact. That the same barrier holds when the
plan shape instead comes from a \emph{real} model reading hostile content is
shown separately in Section~\ref{sec:injection}, where \textsf{gpt-4o} under
an injected ticket emits the gated effect and \sys{} holds and refuses it for
zero effects---the enforcement path is identical, only the source of the plan
shape differs.

\noindent This integration cost is reduced further, not merely reported: the
socket \texttt{GateClient} ships as an installable package
(\texttt{pip install soundgate}) with a native (PyO3) binding and a
byte-compatible pure-Python client, so a deployment adds a dependency and the
per-effect wrapper above rather than vendoring any code. The external-client
mode keeps the gate a separate process, preserving the environment-external
reference-monitor property of Section~\ref{sec:design}; an in-process embedding
is also provided, trading that isolation for lower latency, and is documented as
such.

\begin{table}[!t]
\centering
\caption{End-to-end replay through the live \sys{} process. Every prior
violation is blocked \emph{for mediated effects}---the complete-mediation
contract of Section~\ref{sec:design}, in force here as everywhere---and a
legitimate approved effect is released.}
\label{tab:e2e}
\renewcommand{\arraystretch}{1.25}
\begin{tabularx}{\columnwidth}{@{}Xl@{}}
\toprule
\textbf{Scenario} & \textbf{Outcome via \sys{}}\\
\midrule
Parallel approval + reject       & held, held $\rightarrow$ both refused (\emph{blocked})\\
Resume replay                    & release then refused-duplicate (\emph{blocked})\\
Cancel/timeout zombie            & refused-cancelled (\emph{blocked})\\
Cross-run key reuse              & other run released; own replay refused (\emph{scoped})\\
Crash $+$ restart (WAL)          & dedup and cancel fences survive (\emph{durable})\\
Legitimate approved effect       & held $\rightarrow$ release (\emph{correct})\\
\bottomrule
\end{tabularx}
\end{table}

\noindent\textbf{Does the gate break legitimate completions?}
The scenarios above show the gate blocks what it should; the complementary
question is whether inserting it costs legitimate task success on work it was
never meant to stop. We measure this on a third-party benchmark we did not
author: $\tau$-bench retail~\cite{taubench24}, run to completion under two arms
that differ only in the gate. In the \emph{nogate} arm tools execute under
$\tau$-bench's own semantics; in the \emph{gate} arm every consequential (write)
tool is routed through the live \sys{} with immediate approval---%
$\text{submit}\rightarrow\text{approve}\rightarrow\text{release}$---before it
executes, so the gate is transparent to a legitimate write by construction.
Across twenty episodes per arm at temperature~$0$ (\textsf{gpt-4o} agent,
\textsf{gpt-4o-mini} user simulator; \texttt{results/p6\_completion\_\-retail\_\-gpt4o.jsonl}),
all twenty writes the gate arm submitted released and executed, with
\emph{zero} fail-closed events: the gate never once refused a legitimate
effect. The admission cost was a median $1.1$\,ms round-trip per write
(maximum $3.3$\,ms), invisible against $30$--$40$\,s episodes---median episode
latency was statistically unchanged ($37.9$\,s nogate versus $37.3$\,s gate).
Task reward differed on eight of twenty episodes, but this is model
nondeterminism, not gate damage: three of the eight flips went the \emph{other}
way (the gate arm solved a task nogate failed), which an approve-only gate
cannot cause, and \texttt{fail\_closed}${=}0$ means every gated write executed
identically to nogate---so the variance is \textsf{gpt-4o}'s own
temperature-$0$ nondeterminism compounded by the stochastic user simulator, not
a refusal. The robust findings---no legitimate effect blocked, ${\sim}1$\,ms
over multi-second episodes---answer deployability directly; we report $n{=}20$
as a first signal, noting a precise parity figure needs more episodes to
average out that upstream nondeterminism.

\subsection{Overhead}
\sys{} adds one local round-trip per effect. Isolating the gate decision
itself (Criterion, in-process, on the reference workstation%
\footnote{Reported figures are from a single workstation
(\texttt{\CPUMODEL}, Rust~1.95, Linux); absolute latencies are
hardware-dependent, but the \emph{ordering} below is the stable, portable
claim and reproduces on a second, single-vCPU container environment
(artifact, \texttt{evidence/}).}): a unique release
costs ${\sim}290$\,ns (${\sim}480$\,ns against a gate already holding $10^5$
identities), a full held-then-approve round ${\sim}380$\,ns, and the refusal
paths are the cheapest of all---duplicate hit ${\sim}80$\,ns, cancellation
fence ${\sim}36$\,ns. The safety checks (fence, dedup) are thus the fastest
operations the gate performs, an order of magnitude below a release. Driven end-to-end from a single Python client over loopback TCP (artifact:
\texttt{bench\_socket.py}, $2{\times}10^4$ sequential round-trips), a full
admission---request encoding, round-trip, and decode---costs ${\sim}\SOCKETMED$
at the median (mean $65$, p99 $202\,\mu\mathrm{s}$), two to three orders of
magnitude above the in-process decision. Admission latency is therefore set by
transport and client marshalling, not by the admission logic; a lower-overhead
client or in-process embedding would approach the Criterion figures above. The
one added round-trip is acceptable for the irreversible actions the gate exists
to govern (network calls, payment captures, writes) that already cost
milliseconds, against which ${\sim}53\,\mu\mathrm{s}$ is well under one percent;
we do not claim it negligible for a sub-millisecond local write, which is rarely
the kind of action a human-approval barrier gates. As arithmetic rather than
adjectives: in-memory the added share is
${\sim}53\,\mu\mathrm{s}/(53\,\mu\mathrm{s}+L)$---about $0.03\%$ of a
$200$\,ms payment capture, $1\%$ of a $5$\,ms write, material only when $L$
approaches the transport cost. Durable mode replaces the constant with the
group-commit flush window (p50 ${\sim}7$--$11$\,ms across the sweep,
Table~\ref{tab:load}), negligible against a seconds-to-minutes human decision
and material against sub-$100$\,ms effects---the same crossover one order
higher, set by the storage device. Behavior under concurrent
load, including durable (WAL) mode, is measured next
(Section~\ref{sec:load}).

\subsection{Throughput under concurrent load}
\label{sec:load}
The deployment-relevant mode is
\textsf{WAL-GC}---durable, group-committed---whose rows anchor every
operational claim below; the \textsf{Mem} rows bound the admission logic
itself. The gate serializes admissions per instance (a mutex around the core; the
server is thread-per-connection), so we measure it as a serialization point:
$C$ closed-loop clients, each on its own connection, each issuing sequential
unique-key \texttt{submit} requests on the release path---the state-growing
common case. Table~\ref{tab:load} reports the sweep on the reference
workstation (\LOADENV{}); the same sweep on a deliberately adversarial
single-vCPU container---server and all $C$ client threads timeslicing one
core---establishes a floor.\footnote{Artifact: \texttt{concurrent\_bench.rs};
workstation outputs \texttt{concurrent\_\{mem,wal,walgc\}.txt}. On the adversarial
single-vCPU floor, absolute rates fall (durable throughput tracks the storage
device's flush latency, which differs ${\sim}10{\times}$ across our
environments), while the orderings and the saturation argument in the text are
unchanged---in-memory admission stays far above durable, and group commit
scales monotonically with concurrency in both settings.}

\begin{table}[!t]
\centering
\caption{Closed-loop concurrent admission, unique-release path, on the
reference workstation (\CPUMODEL{}). Mem${=}$in-memory gate;
WAL${=}$fsync-before-acknowledge per state-changing op;
WAL-GC${=}$the same discipline with group commit (batched fsync). Latencies
are per-admission, microseconds; single-vCPU container floor in the
footnote.}
\label{tab:load}
\renewcommand{\arraystretch}{1.2}
\setlength{\tabcolsep}{3.2pt}
\scriptsize
\begin{tabular}{@{}lrrrrrr@{}}
\toprule
\textbf{Mode} & \textbf{C} & \textbf{adm/s} & \textbf{p50} & \textbf{p95} & \textbf{p99} & \textbf{p99.9}\\
\midrule
Mem & 1   & 16{,}889  & 50.9   & 84.5   & 163.4   & 237\\
Mem & 8   & 204{,}940 & 33.4   & 46.3   & 62.0    & 107\\
Mem & 32  & 374{,}929 & 67.5   & 149.5  & 259.1   & 610\\
Mem & 128 & 325{,}548 & 329.1  & 579.2  & 1{,}139.9 & 3{,}914\\
WAL & 1   & 130       & 6{,}264.4 & 10{,}266.8 & 10{,}606.0 & 15{,}646\\
WAL & 32  & 1{,}295   & 19{,}654.8 & 20{,}255.7 & 20{,}742.3 & 28{,}505\\
WAL-GC & 1   & 124     & 9{,}701.9 & 10{,}282.3 & 10{,}630.3 & 15{,}538\\
WAL-GC & 8   & 715     & 10{,}537.1 & 11{,}145.3 & 12{,}013.6 & 18{,}001\\
WAL-GC & 32  & 2{,}397 & 10{,}825.0 & 11{,}622.6 & 14{,}060.8 & 22{,}630\\
WAL-GC & 128 & 11{,}984 & 6{,}944.8 & 13{,}151.6 & 17{,}176.9 & 35{,}051\\
\midrule
Raft-3 & 1   & 1{,}019  & 415.6    & 4{,}435.7  & 6{,}987.0   & 9{,}417\\
Raft-3 & 8   & 1{,}829  & 974.0    & 18{,}843.2 & 27{,}976.3  & 39{,}310\\
Raft-3 & 32  & 1{,}630  & 2{,}090.5 & 93{,}495.6 & 157{,}513.2 & 256{,}130\\
\bottomrule
\end{tabular}
\end{table}

Three readings. \emph{First}, throughput scales with cores until the mutex
and thread oversubscription bite: ${\sim}205$k adm/s at $C{=}8$, a peak of
${\sim}375$k at $C{=}32$, and ${\sim}326$k at $C{=}128$ (past the machine's
16 hardware threads). The $C{=}1$ point (${\sim}17$k adm/s, p50 $51\,\mu$s)
is bounded by cross-core wakeup latency, not the gate---it matches the
sequential \texttt{bench\_socket} median ($\SOCKETMED$) on the same machine,
cross-validating the two harnesses. Since every admitted effect corresponds
to an irreversible external action already costing milliseconds, a single
gate instance is not the bottleneck until effect rates no single agent
deployment reaches---and runs shard across instances with zero coordination
(Section~\ref{sec:failure-model}). \emph{Second}, durable mode is flush-bound, and whether concurrency helps
the \emph{unbatched} server is an accident of the storage stack: on a
deliberately adversarial single-vCPU container the unbatched server shows a
pure inversion (clients only queue behind the lock-per-fsync, so added
concurrency does not raise throughput), while the workstation's NVMe/ext4
coalesces commits below \texttt{fsync} ($130{\to}1{,}295$ adm/s at
$C{=}1{\to}32$)---batching the application neither requested nor controls.
The reference server's default WAL writer makes the batching deliberate
(group commit~\cite{dewitt84}: one writer thread, ${\leq}512$-event batches,
each reply strictly after the fsync covering its event; the unbatched build
behind Table~\ref{tab:load}'s WAL rows is pinned separately): throughput then
scales monotonically with concurrency in \emph{both} environments---on the
workstation $124/715/2{,}397/11{,}984$ adm/s at $C{=}1/8/32/128$
(Table~\ref{tab:load}), and monotonically on the container floor as
well---with WAL-line accounting confirming every acknowledged event
persisted. A batch-size histogram (an env-gated stderr
counter; the write/fsync/ack discipline is untouched) shows the coalescing
is demand-driven rather than cap-driven: the mean achieved batch tracks
offered concurrency---$5.3$ events per flush at $C{=}8$, $26.6$ at $C{=}32$,
$116.4$ at $C{=}128$---and the 512-event cap is never reached, so what
bounds batching is the reply-after-covering-fsync discipline itself
(artifact, \texttt{evidence/walgc\_batch\_hist.txt}). Absolute rates track
the device's flush cost, which varied ${\sim}20{\times}$ across our
environments and sessions; the portable claims are that group commit
dominates the unbatched path at every $C{\geq}8$, scales monotonically, and
preserves fsync-before-acknowledge exactly.
\emph{Third}, the same closed-loop harness driven against the three-voter
replicated deployment (\texttt{Raft-3} in Table~\ref{tab:load}) sustains
$1{,}019/1{,}829/1{,}630$ adm/s at $C{=}1/8/32$ with p50
$0.42/0.97/2.1$\,ms, and a pipelined generator reaches ${\sim}4{,}500$ adm/s
at 16 in-flight connections; throughout, both followers tracked the leader's
log index with no replication lag and the log stayed compacted (mid-load
metrics: \texttt{last\_log\_index} advancing, \texttt{purged} trailing it by
the snapshot window). Notably the replicated path \emph{exceeds} the durable
single-node WAL-GC mode at low concurrency ($1{,}019$ vs.\ $124$ adm/s at
$C{=}1$; p50 $0.42$ vs.\ $9.7$\,ms), because committing to two in-memory
followers over loopback is cheaper than a local disk \texttt{fsync}; the two
converge by $C{=}32$. The caveat is that this is a single-machine cluster
(sub-millisecond peer round-trips), so these figures upper-bound throughput
and lower-bound latency for the replicated path---a geographically
distributed cluster pays real inter-node latency per commit---but they
establish that replication for availability runs at throughput comparable to
the gate's own durable single-node mode, not orders below it. The geographic
trade is then \emph{measured} under emulated WAN latency rather than argued:
a \texttt{netem} sweep (artifact, \texttt{scripts/netem\_raft\_sweep.sh};
receipt, \texttt{evidence/netem\_raft.txt}) injects symmetric delay on the
cluster's shared interface and re-runs the identical closed-loop bench
against the same three-voter cluster, with leader and quorum evidence
recorded at each point. At injected RTT${\approx}0$ the sweep reproduces the
loopback regime ($2{,}244/2{,}295/1{,}681$~adm/s at $C{=}1/8/32$, p50
$0.32/0.94/2.2$\,ms); at RTT${\approx}10$\,ms---intra-continental datacenter
spacing---single-client admission collapses to $26$~adm/s at p50
$21.5$\,ms, almost exactly \emph{two} injected round-trips per admission
(client${\to}$leader plus the leader${\to}$follower quorum round): the
consensus arithmetic executed, not estimated. $C{=}8$ sustains $146$~adm/s
at p50 $31.7$\,ms at the same delay. Two scope notes keep the receipt
honest: \texttt{netem} on the shared interface delays the client-to-leader
hop as well as inter-node RPC, so these latencies \emph{upper-bound} the
pure replication cost; and the committed receipt covers the $0$ and
$10$\,ms points, ending before the $C{=}32$ run at $10$\,ms completed. The
deployment reading is unchanged in direction and now quantified: what
governs a wide-area gate is per-effect commit latency---tens of
milliseconds, small against the multi-second model inference preceding each
tool call---not serial admission rate.
The replicated tier buys crash-safety for one consensus round-trip of added
latency; single-node WAL therefore remains the low-latency default, and
three-voter replication is reserved for deployments that require cross-region
availability and can absorb that per-effect cost. We do not claim wide-area
throughput parity with the single-node mode---the trade is latency for
availability, and the sub-Hz per-agent tool-call rate is what would keep it
acceptable.
With the
executed failover receipt (Section~\ref{sec:limitations}), availability is
measured rather than asserted---and one boundary is restated here so the
replicated numbers are not over-read: the verification tiers cover the
single-node admission core; the Raft tier reuses that \emph{unmodified}
verified \texttt{Gate} as its state machine, but the consensus machinery
around it is operational, not verified (Section~\ref{sec:limitations}).

\subsection{Limitations}
\label{sec:limitations}
We foreground the limitations a deployer must weigh.
\emph{Complete mediation is an assumption, not a guarantee}
(Section~\ref{sec:failure-model}): the repair holds only for effects that
traverse the gate, and one unwrapped tool bypasses it \emph{at the wrapper
layer}---and, under the structural routes below, that same unwrapped tool's
external action is \emph{refused by the kernel}, converting the failure mode
from fail-open (a silent leak) to fail-closed (a loud availability error),
which is the design's answer to the objection that one forgotten wrap
collapses the guarantee (Section~\ref{sec:design}). At the wrapper layer alone this is intrinsic to any placement-based enforcement point; we mitigate it with choke-point placement and the best-effort static mediation linter, and where a deployment can afford OS-level enforcement, the two implemented structural routes (loopback-only namespace; \texttt{cgroup} eBPF hooks) make network-egress bypass kernel-refused rather than merely discouraged---exercised with a real framework: the full \textsf{FW-A} integration passes all three repaired scenarios inside a loopback-only namespace, the kernel refusing the unwrapped path throughout (\texttt{evidence/e2e\_structural\_langgraph.txt}). The residual cost is operational (a cgroup or namespace per tool process), with one channel caveat repeated here because a limitations section is where it belongs: the structural routes close \emph{network-externalized}
effects (the namespace for all egress; the \texttt{cgroup} eBPF hooks
per-syscall across \texttt{connect4/6} and \texttt{sendmsg4/6}), and effects that
externalize through a shared filesystem, local IPC, or shared memory remain
on the placement contract or an analogous seccomp/LSM policy
(Table~\ref{tab:mediation}). \emph{Both structural routes are Linux
mechanisms} (network namespaces; \texttt{cgroup} eBPF \texttt{sock\_addr}
hooks on a BTF-enabled kernel): on other operating systems the discharge
degrades to placement discipline plus the linter, or to that platform's own
confinement machinery (sandbox profiles, AppContainer), which we have not
implemented or evaluated. \emph{Effects are modeled at a single commit
point} (Section~\ref{sec:background}): the identity
$(\mathit{run},\mathit{key})$ admits one externalization. Real APIs are not
always so shaped---a payment authorization followed by a capture, a
multipart upload followed by its completion call, a resource created and
then activated each externalize at more than one point. The mapping is
phase-grained mediation: each externally visible phase is its own keyed
admission (its own hold, decision, dedup, and fence), so the barrier holds
per phase; what the gate does not provide is atomicity \emph{across} phases
of one logical action---an approved authorization whose capture is later
rejected leaves the authorization standing, and undoing it is the
compensation territory of Section~\ref{sec:related}, not admission. A
deployment gating multi-phase APIs therefore chooses the phase at which
human authority attaches (typically the irrevocable one) and keys it
explicitly; the key-generation taxonomy of Section~\ref{sec:design} is where
that choice is expressed. \emph{A hold fences actions, not knowledge}: sibling
branches keep computing during a hold (by design; Section~\ref{sec:design}),
so a sibling's \emph{read} may observe state that predates the pending
decision, and the agent may plan against it. Nothing stale is ever
\emph{released}---any consequential action that planning produces still
queues behind its own admission---but read-set freshness under holds is
state coordination, the territory of concurrency-anomaly analysis for
multi-agent systems~\cite{khan26}, not of an admission gate; we scope it out
explicitly rather than leave it implied. \emph{Occurrence evidence is indirect}: the exposure study
measures plan-shape \emph{emission} (Section~\ref{sec:exposure}) and the
corpus measures independent reports (Section~\ref{sec:prevalence}); neither is
a census of production incidents, and the end-to-end repair
(Section~\ref{sec:eval}) is demonstrated on all six frameworks (all
four execution models, both language runtimes)---the five gate-shipping
frameworks on their full violated set and the post-hoc-review framework
(\textsf{FW-E}) on its two approval-independent axes. \emph{The measurement is a snapshot} of current releases; we check
version stability only for \textsf{FW-A} (two releases), and upstream fixes
may change specific cells---which is a goal, not a threat, and argues for
frameworks shipping barrier-semantics regression tests. \emph{The corpus is
small and tracker-biased} (Section~\ref{sec:prevalence}). \emph{Exposure uses
five models and ten authored tasks} at $N{=}100$, with native-API anchors for
GPT-4o and Claude and a second serving path for all three near-zero models
(native for Gemini and DeepSeek, provider-direct via Together for Llama);
wider task coverage and naturalistic, multi-turn prompts would
tighten the estimates further, and on the Llama arm $117/1000$ runs were
excluded as unclassifiable (schema-nonadherent tool calls), leaving one
compound task thinly sampled.
\emph{Naturalistic tasks and prevalence:} to test whether the exposure shape
is an artifact of our ten authored tasks, we ran the same measurement over a
third-party benchmark we did not write---$\tau$-bench's~\cite{taubench24}
retail and airline customer-service episodes---with two frontier models
(GPT-4o, Claude) each driving $\tau$-bench's own tool-calling agent, recording
every emitted tool-call batch before the framework truncates it
(\texttt{results/taubench\_exposure\_*.jsonl}; $431$ tool turns). The dominant
behavior is serialization: parallel batches do occur (about one tool-turn in
seven), but they are overwhelmingly read-only fan-out, and across the $71$
gated batches observed not one co-emitted a benign read---the exposure
shape of Table~\ref{tab:exposure} did not arise on these tasks. A more severe
\emph{consequential}-sibling variant---two or more write calls emitted in one
batch, so that approving one does not fence the others---did appear, but only
for GPT-4o (nine of its $52$ gated batches, about one in six, spanning both
domains) and never for Claude, which serialized writes as well. We read this plainly: the benign-sibling exposure
is a controlled measurement of a real barrier failure, not a claim that the
shape is frequent in the wild, where current frontier models tend to gather
information in parallel and then act sequentially. That serialization, however,
is an emergent model habit, not a property the framework enforces: the
primitive supplies \emph{no} barrier, so any model or configuration that does
batch gated calls---GPT-4o's writes here, or an explicit parallel-tool
setting---executes the siblings unchecked. Tellingly, $\tau$-bench's shipped
agent hardcodes a truncation of each turn to its first tool call---a framework
suppressing parallel calls by fiat rather than gating them, itself evidence
both that models emit such batches and that the barrier is absent by default.
Stated as conditions rather than a mood, the shape matters where (i) input is
adversarial (Section~\ref{sec:injection} induces it deterministically),
(ii) a deployment enables, or a model favors, parallel tool calls (GPT-4o's
consequential-write batches here; explicit parallel-call settings), or
(iii) the orchestration fans out by construction (the multi-branch graphs of
Section~\ref{sec:measurement})---and the barrier costs nothing where models
happen to serialize, because it is only exercised when a sibling exists.
The gate's value is therefore in closing a severe, latent gap, and does not
rest on a high base rate. Three considerations make a latent gap nonetheless
worth a barrier, stated as the justification rather than a plea from
incidence. First, \emph{severity, not frequency}, is the criterion
functional-safety engineering applies to an irreversible actuation
(\cite{iec61508,iso26262}, Section~\ref{sec:related}): an interlock on a
catastrophic, unrecoverable action is justified by the cost of the action and
the cheapness of the interlock, not by how often it fires---and here the
interlock is free where models serialize. Second, the benign base rate bounds
only the benign regime: under adversarial input the shape is induced
deterministically (Section~\ref{sec:injection}), so an attacker does not
inherit the model's serialization habit, and the operative rate under threat
is not the naturalistic one. Third, the measured direction is toward
activation, not away from it---the models that do batch already emit
consequential-write siblings (GPT-4o, nine of $52$ gated $\tau$-bench
batches), and genuinely multi-effect tasks batch near-certainly ($500/500$,
Section~\ref{sec:exposure})---so the gap widens as
parallel tool-calling spreads, and the barrier is adopted before that
transition, while adoption is cheap. A production census of multi-effect
tasks remains future work; the third-party signal we do have is
$\tau$-bench's own reference solutions, of which 39\%/30\% (retail/airline)
require at least two consequential writes (Section~\ref{sec:exposure}).
\emph{Availability and scale}: because the gate
fails closed, its downtime pauses effects---the designed posture for
irreversible actions, and bounded where a deployment prefers it by the
executed TTL-watchdog policy, which converts silent starvation into a
visible, sticky rejection (Section~\ref{sec:design},
Remark~\ref{rem:watchdog}); we evaluate a single node under
concurrent load (Section~\ref{sec:load})---closed-loop, one machine---but not
replication of the admission log, which a production deployment needs; per-run
sharding (Section~\ref{sec:failure-model}) addresses throughput, not
availability. This is no longer only a design argument: the artifact now
carries a feature-gated replicated deployment (\texttt{soundgate\_raft}) in
which the Raft log entry is the admission operation and the state machine is
the unmodified \texttt{Gate}, committed across a three-voter openraft
(Raft~\cite{raft14}) cluster before any verdict is acknowledged---the
replicated analog of the single-node fsync-before-reply discipline. The
executed failover receipt (\texttt{raft\_failover.txt}) shows the availability
and safety halves together: after \texttt{kill -9} of the leader, a new leader
is observed within 1.9\,s (0.5--1\,s election window plus detection), the
identity released \emph{before} the crash refuses as a duplicate on the
\emph{new} leader---the replicated fence survives leader death, so no
double-release across failover---and fresh effects admit normally. What we still do not claim is a
\emph{distributed} throughput characterization: the replicated figures in
Table~\ref{tab:load} come from a single-machine three-process cluster with
loopback peer links, so they upper-bound throughput and lower-bound latency
relative to a geographically distributed deployment, whose per-commit cost is
dominated by inter-node round-trip time---a cost now quantified under
emulated WAN delay by the \texttt{netem} sweep of Section~\ref{sec:load}
(${\sim}2{\times}$RTT per admission at RTT${\approx}10$\,ms), an emulation on
one machine rather than a multi-region deployment; one executed failover is a
demonstration, not a fault-injection study. Nor is the replication layer
itself inside the verified scope: the Raft tier reuses the \emph{unmodified}
verified \texttt{Gate} as its state machine, but the consensus machinery
around it---\texttt{openraft} election, log storage, snapshot
transfer---is operational, exercised by the failover and no-lag receipts
rather than verified; a deployment adopting it trusts openraft's Raft the
way single-node mode trusts \texttt{fsync}. The five-model $N{=}100$ study of Section~\ref{sec:exposure}
executes the first arm of a wider protocol; its remaining arms (30 tasks,
native-API replication for the OpenRouter-only models, unconditional and
execution-rate endpoints) are future work. These do not undermine the
core measured result---the barrier does not hold across the frameworks
tested---but they bound the strength of the ``in practice'' and ``repaired in
general'' claims.

\subsection{Threats to validity}
\label{sec:threats-validity}
\emph{External validity:} we measure six frameworks, one (\textsf{FW-A}) in
most depth. The selection is not by popularity alone but by \emph{execution
model}: the five gate-shipping frameworks instantiate the four architectural
families a production agent runtime is built on---Pregel/BSP supersteps, an
event-driven bus, message-passing fan-out, and parallel tool calls within a
single model turn---plus a second language runtime (Node) of one of them, so
the recurrence is demonstrated across the dominant execution models rather
than across near-duplicates of one design; a leak that survives all four
families is unlikely to be an artifact of any single scheduler. Closed,
hosted agent platforms expose no local runtime to instrument and are the
stated boundary of this open-runtime scope. One axis remains unprobed for a
documented reason (\textsf{FW-E}
checkpoint-resume, where the framework cannot serialize a tool-bearing crew's
state), and the a-priori
kill condition (Section~\ref{sec:scope}) states that if every other framework
proves sound on every axis and the observed behaviors are patched upstream before
submission, the contribution narrows toward a two-framework case study. The
cross-framework recurrence of the sibling leak argues against that outcome.
\emph{Construct validity:} representing effects as event-log appends could
understate framework-internal mitigations; we mitigate by also observing
caller-visible signals (the exception a caller sees, the number of effect
occurrences). \emph{Internal validity:} probes are model-free and deterministic,
removing LLM nondeterminism as a confound.
\emph{Reliance assumptions:} the repair's guarantees inherit the mediation
contract and the mechanisms that discharge it---an incorrectly written
wrapper is an unmediated path by definition (the linter's target, and the
structural routes' reason to exist), and kernel-, filesystem-, and
consensus-layer misbehavior sits inside the stated trusted computing base
(Section~\ref{sec:failure-model}); the fuzz receipt tests the protocol
boundary's fail-closed posture, not those layers' internals.

% =====================================================================
\section{Related Work}
\label{sec:related}
\textbf{Transactional settlement for tool calls.} The closest mechanism to
ours is Atomix~\cite{atomix26}, which wraps tool calls in progress-aware
transactions: effects are tagged with epochs, per-resource frontiers gate
commit, bufferable effects are delayed, and externalized effects are
compensated on abort. Atomix's epoch tags and frontiers are themselves a
fencing discipline against stale effects, and it answers \emph{when} it is
safe to settle an effect under speculation, contention, and
faults---presuming its transactional runtime and resource model are adopted;
we answer a prior question:
whether the stop primitives frameworks \emph{already ship} deliver their
implied barrier contract at all---and repair the measured gaps with a
deliberately minimal admission gate that needs no resource model, no
compensation logic, and no rollback assumption (many gated actions, an e-mail
or a payment capture, admit none). The two compose: our measurement motivates
precisely the settlement layer Atomix builds, and \sys{}'s
hold/dedup/fence core is the fragment deployable without adopting a
transaction model; where Atomix's resource model applies, its transactions
subsume our dedup and fence, and the gate is the conservative fallback when no
compensation or resource model exists. SagaLLM~\cite{sagallm25} brings saga-style
compensation~\cite{sagas87} to multi-agent planning and likewise presumes
compensability; our gate targets the admission side of the same problem.

\textbf{Policy and privilege enforcement.} AgentSpec~\cite{agentspec25}
enforces user-written runtime rules over agent actions; Progent~\cite{progent25}
provides programmable privilege control over tool calls;
CaMeL~\cite{camel25} isolates control flow from untrusted data by design.
These systems decide \emph{which} calls are permissible, largely against
adversarial inputs; prompt-injection benchmarks such as
AgentDojo~\cite{agentdojo24} evaluate that attack axis directly, and
execution-isolation architectures such as IsolateGPT~\cite{isolategpt25}
confine LLM-app components from one another. Our axis is orthogonal: even a
fully permitted,
human-gated call lacks barrier semantics under parallelism, replay, and
cancellation---the sibling leak persists under a perfect policy layer, because
it is a concurrency-control failure, not an authorization failure. \sys{}
composes beneath such layers at the admission point; embedded in the
framework, a policy layer inherits the framework's control flow---and with it
exactly the parallel-step races of Table~\ref{tab:matrix}.

\textbf{OS-level sandboxing and isolation.} A different lineage confines what
an untrusted process \emph{can reach}---beginning at the instruction level
with software-based fault isolation~\cite{wahbe93}: \texttt{seccomp-bpf}
syscall filters~\cite{seccomp}, Linux namespaces, gVisor-style user-space
kernels~\cite{gvisor19}, and WebAssembly/WASI capability sandboxes~\cite{wasi}, whose component model makes an import a component was never granted simply not exist for it---a viable mediation-discharge route for tools compiled to that target. Type-and-effect systems~\cite{lucassen88} make the same distinction statically, but no static discipline expresses the run-time, per-identity hold/decide state the barrier needs. These are mediation \emph{enforcers}, not admission \emph{deciders}: a sandbox constrains destinations and syscalls before execution, whereas the barrier's properties are stateful, per-identity, and decided at run time. The two compose rather than compete---our structural-mediation routes (Section~\ref{sec:design}) use exactly this machinery to force every effect path to the gate, which supplies the admission semantics the sandbox cannot.

\textbf{Durable execution and idempotency.} Hosted workflow services expose
the approval-gate shape as callback \emph{task tokens}, and exactly-once
messaging scopes idempotence by producer identity and sequence number (e.g.,
Kafka's idempotent producer~\cite{kafka-kip98}) or by per-record keys checked against durable
state, as in MillWheel and its descendants~\cite{millwheel13}, and Flink
obtains the analogous exactly-once guarantee by aligned checkpointing over
operator state~\cite{flink17}---\emph{(run,\,key)} is the same discipline at
the agent-effect boundary. The transactional-outbox pattern~\cite{outbox},
which commits an effect's intent to the same store as its state change and
relays it once, is the same hold-then-release split with the store as
arbiter rather than an external gate. Dedup-on-replay is the idempotent-consumer discipline of event-sourced systems~\cite{eventsourcing,helland12} transplanted to the agent-effect boundary rather than an invention of ours, with ancestry from Gray's transaction notes~\cite{gray78} to deterministic databases in the Calvin lineage~\cite{calvin12}.
Durable-execution platforms such as Temporal~\cite{temporal} re-architect orchestration into deterministic workflows with journaled activities; Section~\ref{sec:temporal} executes that audit rather than presuming it: journaled activities close the replay axis by construction, while the sibling, cancellation, and timeout predicates carry the same behavioral bits as the measured frameworks---documented as cooperative rather than implied as barriers, exactly the contract difference this paper isolates---and the gate composes at its activity boundary unchanged. Step Functions exposes no local runtime to instrument and remains
outside the open-runtime scope; its documented Parallel-state semantics make
the task-token pause branch-local by specification, and we assert nothing
beyond that documentation. 
Workflow engines in the Airflow/Argo/Cadence lineage, hosted durable-execution services (AWS Step Functions, Azure Durable Functions; Amazon reports specifying such services in TLA\textsuperscript{+} before shipping them~\cite{newcombe15}), and the BPMN~2.0/WS-BPEL engines whose compensation handlers and cancellation scopes standardized this machinery two decades ago~\cite{bpel07} all ship human-approval tasks under the same implicit barrier assumption we audit; we measure agent frameworks because that is where the assumption now carries irreversible tool calls proposed by a model. Why not simply adopt one of these engines? For greenfield systems one should: deterministic replay solves this by construction. But re-platforming an existing agent is a rewrite of its control flow, whereas the minimal retrofit preserving each framework's programming model is an external admission point at the tool boundary---the only point whose guarantee does not depend on which engine was chosen. \sys{} transplants classical idempotence and fencing
discipline~\cite{helland12} to that boundary, with effect identity scoped per
run. A sharper form of the same question is whether the frameworks' \emph{own}
persistence layers---checkpointers and durable-state stores now shipping in
the very frameworks we measure~\cite{lgdocs26}---already close the gap
natively. They do not close the axis this paper isolates: a checkpointer makes
\emph{state} recoverable across a resume, which is exactly the machinery whose
ordering race produces the replay double-execution of
Section~\ref{sec:measurement} (the maintainer-filed persistence-ordering issue
of Section~\ref{sec:prevalence} is a checkpointer bug), and it says nothing
about a \emph{sibling} effect executing during an approval pause---that leak is
present with checkpointing fully enabled, because the sibling commits inside
the same superstep regardless of whether state is later journaled. Durable
persistence is complementary to, not a substitute for, a decided admission
point.

\textbf{Temporal-mismatch vulnerabilities.} TOCTOU studies in LLM
agents~\cite{toctou-mindgap25} and browser-use agents~\cite{toctou-browser26}
exploit the window between an agent's check and its use of external state. The
sibling leak is the same structure turned inward: a framework-internal TOCTOU
whose ``check'' is the approval gate and whose ``use'' is the sibling's effect,
the time-of-check-to-time-of-use window being the pause itself. Our
axes are the control-plane dual: windows between a \emph{human decision
point} and an effect, created by the framework's own pause, resume, cancel,
and timeout machinery---and they open with no adversary present. Where those
studies mitigate by fusing check and use into one atomic tool call, the gate
instead makes the use \emph{wait} on an external admission of the check.

\textbf{Agent operating layers.} A recent line frames the agent runtime
itself as an operating system---MemGPT's interrupt-driven control flow over
tiered memory~\cite{memgpt23} and AIOS's agent-level scheduling, context
switching, and resource isolation~\cite{aios24}. These layers govern
\emph{when agents run}; neither interposes a decided admission point between
an emitted tool effect and the world, so the barrier question this paper
measures arises unchanged inside them, and the gate composes at their tool
boundary exactly as it does beneath the frameworks of
Section~\ref{sec:measurement}.

\textbf{Failure studies and practitioner evidence.} MAST~\cite{mast25}
taxonomizes multi-agent failures from execution traces at the
reasoning/coordination level; recent work ports concurrency-anomaly detection
and verification to multi-agent systems~\cite{khan26}. An older
agent-verification line model-checks BDI agent \emph{programs}
directly~\cite{ajpf12}; those are properties of agent decision logic, where
ours are properties of the surrounding orchestrator's control plane. Practitioner
guidance~\cite{abstractalg26} and vendor documentation~\cite{lgdocs26}
acknowledge re-execution and HITL pitfalls individually. Our contribution is
the executable differential predicates, the cross-framework matrix, and the
occurrence measurements that turn scattered acknowledgements into a
characterized, mechanically repaired gap.

\textbf{Enforcement, capabilities, and distributed idioms.} As a mechanism,
\sys{} sits in the reference-monitor lineage of Anderson~\cite{anderson72}
(complete mediation being that lineage's defining obligation, and ours) and is
an execution monitor in Schneider's sense~\cite{schneider00} with one
refinement: Schneider's truncation automata \emph{halt} the program at a bad
action, whereas the gate never halts the agent---it \emph{suppresses the single
external action} and lets computation continue, placing it in Ligatti's
edit/suppression automata~\cite{ligatti05} rather than the truncation class.
Inlined reference monitors weave such checks into the untrusted program
itself~\cite{erlingsson00}; the gate deliberately keeps its checks in a
separate process, because here the untrusted control flow is precisely what
fails.
The runtime-verification community~\cite{leucker09} has built general
monitor-synthesis machinery of this kind (JavaMOP~\cite{javamop12}); our gate
is a fixed-policy, effect-boundary instance. Its protocol boundary sits in the
language-theoretic-security tradition~\cite{langsec16}---a minimal input
language, a single typed decode whose rejection is the fail-closed default, a
fixed reply alphabet---and the fuzz receipt of Section~\ref{sec:failure-model}
executes that discipline. Its per-run, per-key admission is a capability-style check in the object-capability tradition~\cite{miller06}; OS/hardware capability systems (Capsicum~\cite{capsicum10}, CHERI~\cite{cheri15}) enforce \emph{which resources} a component may name, orthogonal to the \emph{when} the gate governs---a process may hold a capability to a payment endpoint yet submit while the run is paused.

Our cancellation orphan is the forty-year-old orphaned-computation problem of
RPC~\cite{birrell84} resurfacing under async runtimes. Erlang/OTP's supervision
trees~\cite{armstrong03} solve crash \emph{propagation} inside one trusted
runtime; structured-concurrency disciplines with explicit cancellation
scopes---Trio nurseries, Project Loom task scopes, Kotlin and
Swift~\cite{structconc}---make cancellation lexical, and Rust documents
\emph{cancellation safety} primitive-by-primitive because a dropped future
orphans what it had in flight~\cite{tokio-cancel} (machine-checked foundations
for Rust's ownership exist at the language level, RustBelt~\cite{rustbelt18};
verifying an async \emph{runtime} against them repairs the substrate the gate
deliberately distrusts), while formal Node event-loop
semantics~\cite{asyncjs17} make our \textsf{FW-F} orphan the substrate's
expected outcome. All of these address \emph{cancellation} only---silent on the approval-barrier and replay axes---and adopting them rewrites a framework's concurrency substrate: the frameworks we measure are unstructured \texttt{asyncio} and bare promises, which is why their cancellation soundness is host-dependent (Table~\ref{tab:matrix}). ``Use a better concurrency model''
helps a new framework, not the deployed ones; the external gate retrofits the
guarantee onto the substrate that exists and covers the axes a cancellation
discipline does not.

The cancellation/close fence is the agent analogue of the \emph{fencing tokens}
and epoch counters of Chubby, ZooKeeper, and etcd
leases~\cite{burrows06,fencing-kleppmann}; the hold/decide split mirrors
two-phase commit's prepare/commit barrier~\cite{2pc-gray} without a
compensation phase; a hold is, in classical terms, an exclusive intention on
one effect identity---two-phase locking's blocking~\cite{eswaran76} without its
contention, since a held effect locks no shared data. 2PC's classical objection (a prepared participant blocks, holding locks) does not bite: a held effect locks nothing and blocks no other run, so coordinator failure yields one un-externalized action, not a stalled cluster---fail-closed, with the bounded-hold watchdog converting the wait to a timed rejection where preferred. Operationally the gate resembles a service-mesh or
RPC interceptor applied to agent effects---kin to Kubernetes admission
controllers~\cite{k8s-admission}, Envoy's external-authorization
filter~\cite{envoy-extauthz}, the sidecar interception meshes such as Istio and
Linkerd~\cite{linkerd} generalize, and payment idempotency-key
middleware~\cite{stripe-idem}---but the mesh marks the boundary: a sidecar
interposes \emph{transport-level} allow/deny keyed by connection or route,
whereas the gate's verdict is \emph{stateful and per-effect-identity}, which no
transport proxy carries; distributed-tracing interceptors in the
OpenTelemetry lineage~\cite{otel} propagate per-request \emph{context}
through the same choke points, and context propagation can carry the gate's
identities but does not decide them. Why not a container network policy, or Envoy's \texttt{ext\_authz} filter backed by OPA? A deny-by-default egress policy is the \emph{mediation} half---our structural routes are precisely that mechanism---but a firewall rule embodies no hold-until-decided, no per-identity dedup, and no run fence: those are verdict \emph{state}, not reachability. Make OPA stateful enough to hold one effect for a minutes-long human decision, release exactly that effect once, refuse its replay, and fence its cancelled run, and one has re-implemented the gate's admission state machine at the proxy---a legitimate topology for the same core, whose correctness burden is then exactly the state machine we verify. Authorization systems (OPA/XACML, the Zanzibar
lineage~\cite{zanzibar19}) govern \emph{who} may decide; the gate governs what
a decision \emph{means}, and the two compose. Credential-attached policy (capability tokens with caveats, as in Macaroons~\cite{macaroons14}) bounds \emph{what} an effect may be; the barrier question is temporal and stateful---whether \emph{this} effect may externalize \emph{now}---which a stateless token cannot answer without rebuilding the gate's state machine inside the credential. The sibling leak is, classically, a
control-flow anomaly the workflow-patterns catalog~\cite{wfpatterns03} names,
workflow-net verification models as pause/resume/cancel regions (first-class in
YAWL~\cite{yawl05}), and reachability analysis~\cite{vdaalst98} would
classify---the news being not the anomaly class but that shipping agent
frameworks reintroduce it and nobody had measured them. We claim no novelty in
these mechanisms; the contribution is identifying the agent control-plane gap,
measuring it, and showing this machinery closes it at the tool boundary. The
Model Context Protocol~\cite{mcp} standardizes that boundary---a gate deployed
as an MCP-side proxy is exactly the choke point placement wants, and MCP's
transport/discovery/authorization layer supplies the who/what axis but not
stop-state admission, so hold/dedup/fence compose on it unchanged. Hosted products move the confirmation prompt into the provider (OpenAI's ChatGPT agent requires explicit confirmation before consequential actions~\cite{chatgptagent25}); a provider-side pause relocates the gate upstream and inherits the same audit question---whether it fences parallel siblings inside a runtime no local probe can instrument---the closed-platform boundary our scope states.

\textbf{Interruptibility in AI safety.} The alignment literature has long
studied stopping intelligent agents: safe interruptibility~\cite{orseau16}
and corrigibility~\cite{soares15} ask that an agent not learn to resist or
manipulate an off-switch. That work concerns an agent's \emph{incentives} to
avoid interruption; ours concerns whether the surrounding framework
\emph{mechanically} halts effects when a human does intervene. The problems
are complementary---a corrigible policy still leaks a sibling effect if the
orchestrator lacks a barrier---and to our knowledge this mechanical,
control-plane view of agent stopping has not been measured across frameworks.
AI-control protocols designed to stay safe against subversive models~\cite{aicontrol23} presuppose, whenever they review actions before they land, a mechanism that actually withholds the action pending review---the admission layer the gate supplies and the frameworks, as measured, do not. ToolEmu~\cite{toolemu24} estimates tool-action risk in emulated sandboxes and AgentBench~\cite{agentbench24} measures tool-use competence; we address the dual question of whether the machinery can \emph{enforce} a human decision once an action is flagged.

\textbf{Fail-safe control in safety engineering.} The principle that a
safety-critical system must either perform its function correctly or fail into a
predictable safe state is long codified in functional-safety
engineering---IEC~61508 for programmable electronic
safety-related systems~\cite{iec61508}, its automotive adaptation ISO~26262 for
road vehicles~\cite{iso26262}, and the emergency-stop and protective-stop
requirements of ISO~10218 for industrial robots~\cite{iso10218}---where a stop command is a hard guarantee, not a request, and hazardous actuation defaults to the de-energized state. We neither certify a safety integrity level nor model physical hazards, but the discipline transfers: the gate treats an externally visible tool effect as the actuation to be interlocked, failing closed on cancel, reject, or timeout rather than racing to completion---and the measurement half of this paper is in effect a conformance test of whether shipping frameworks honor that stop-means-stop contract for their software actuators.

\textbf{Positioning.} A natural reviewer question is why, if the gap is real,
framework maintainers have not already closed it. The corpus answers empirically: incremental fixes are landing (a merged single-interrupt fix; a January~2026 cluster of routing bugs closed together; an open enhancement for the general barrier)---the signature of a recognized-but-unsolved concurrency gap. The general fix is hard for the structural reason
above: closing it inside a framework means either serializing gated steps
(surrendering the parallelism the execution model exists to provide) or
building an in-control-flow admission point that inherits the very scheduler
that races. Patching each framework upstream is complementary, and
our regression probes exist to support it; the gate exists because operators
run several frameworks and versions at once and need one enforcement point
whose guarantee does not depend on each upstream's fix landing and staying
fixed. Our primary contribution is control-plane
\emph{correctness}; the security implications (an unwrapped or spoofed effect)
are real but secondary; we \emph{implement} authenticated decisions
(HMAC-SHA256, attack-tested; Section~\ref{sec:failure-model}) and treat
complete mediation as a deployment obligation dischargeable structurally
for network egress (Section~\ref{sec:design}) rather than a fully solved
part of the design. The paper's value is deliberately empirical: measurement, plus a repair whose admission core is verified as a model and connected to the deployed Rust by differential conformance---refinement evidence, not a refinement proof (Section~\ref{sec:failure-model}).

% =====================================================================
\section{Conclusion}
Stop primitives are the controls practitioner guidance treats as load-bearing
for agent systems with real tool
access~\cite{abstractalg26,lgdocs26}, yet on the frameworks and releases we
measured they do not provide barrier semantics: an approval gate
can be bypassed by a sibling branch during the pause it creates, rejections can
arrive too late, effects can double under replay, and cancellations and timeouts
can leave work running. These gaps recur across independently designed
frameworks, four execution models, and two language runtimes; whether the cause
is a shared design pattern or the shared constraints of the underlying async
substrates, they are not one implementation's accident. We are equally plain
about prevalence: on naturalistic tasks current frontier models mostly
serialize their consequential calls, so the everyday gap is \emph{latent}
rather than prevalent---a window in which barriers can be adopted before
parallel-call defaults, model behavior, multi-effect tasks, or adversarial
input close it---and because the barrier is exercised only when a sibling
exists, it costs nothing where models happen to serialize. Enforcing stop
semantics at the tool boundary---outside the framework's own control flow, and
under an explicit complete-mediation contract---closes every measured
gap for mediated effects in end-to-end replay on every evaluated
framework---all four execution models, both language runtimes---while
preserving legitimate actions, at overhead
bounded by the transport for in-memory admission and by \texttt{fsync} for
durable admission. We argue that
environment-external effect gating is one robust default for how agent
systems with irreversible tool access are operated.

\section*{Acknowledgment}
The author thanks the maintainers of the open-source agent frameworks whose
public issue discussions informed the probe design.

\clearpage
\appendices
\section{Exposure Study: Serving-Path and Multiplicity Detail}
\label{app:exposure}
This appendix carries the serving-path mechanics and forensic detail behind
Section~\ref{sec:exposure}; every number is recomputed from the committed
raw logs.

\smallskip\noindent\textbf{Transport mechanics.} Transient transport errors
on the OpenRouter path were retried, and records are deduplicated per
$(\mathit{task},\mathit{run})$ keeping the final attempt (raw and
deduplicated counts in the artifact). We do not send a decoding seed on that
path: doing so would additionally require a backend supporting
\texttt{tools+seed} together, which several do not, and seed determinism
never held across backends regardless; temperature-1.0 sampling averaged
over $N{=}100$ does not depend on it.

\smallskip\noindent\textbf{Multiplicity.} The two headline task effects are
large enough that no correction for the ten implicit per-task comparisons
reclassifies them (Bonferroni-adjusted $99.5\%$ Wilson intervals:
$[0.61,0.85]$ for GPT-4o's $75/100$, $[0.22,0.48]$ for Claude's $34/100$).
Extending the correction to the most conservative family the study
admits---every task$\times$model cell in \emph{both} serving-path panels, 50
primary plus 30 second-path, $m{=}80$---changes nothing: at a joint 95\%
level (per-cell Wilson at $1{-}0.05/80$), every reported cell at or above
$0.10$ still excludes zero---GPT-4o \texttt{compound\_cleanup} $75/100$
$[0.58,0.87]$ and \texttt{single\_offboard} $40/100$ $[0.25,0.57]$, Claude
\texttt{compound\_transfer} $34/100$ $[0.20,0.51]$, Llama-Together
\texttt{compound\_cleanup} $57/63$ $[0.71,0.97]$, and Gemini-native
\texttt{compound\_invoice} $17/100$ $[0.08,0.33]$---so no nonzero finding
depends on the number of looks. All remaining cells are zeros or singletons
reported with their intervals, per the estimation framing of
Section~\ref{sec:exposure}; the complete per-task table is committed to the
artifact and nothing is selectively reported.

\smallskip\noindent\textbf{Positive-control diagnostic.} Stripping the
consequential tool from \texttt{compound\_transfer}, leaving only its two
benign siblings so there is zero hazard on the table (artifact,
\texttt{gemini\_diagnostic.py}), GPT-4o still bundles both tools in one turn
on $12/15$ trials ($80\%$). Gemini, DeepSeek, and Llama bundle on $0/30$
trials \emph{each}---ninety combined trials, zero bundled turns---though
they differ in how: Llama engages every trial and calls exactly one tool
each time; DeepSeek is similar ($29/30$ one-tool); Gemini mostly disengages
from the stripped task entirely ($28/30$ neither tool called). We cannot
fully separate a schema-translation artifact of the shared integration from
a genuine training-time bias toward strictly sequential tool use; either way
the identical near-zero appears when there is nothing to be cautious about,
which is why the OpenRouter rates for these three models are reported as
pathway-confounded rather than as exposure or safety.

\smallskip\noindent\textbf{Second-path forensics.} \emph{Gemini:}
\textsf{gemini-3.5-flash} on Google's native surface
(\texttt{generativelanguage.googleapis.com}), same $N{=}100$ over the same
ten tasks with no OpenRouter layer, called the gated tool in $1000/1000$
runs and emitted it in parallel with a benign sibling in $17/1000$,
concentrated on \texttt{compound\_invoice} ($17/100$) and zero on the other
nine; a newer generation than the 2.5 OpenRouter row, so the de-confound is
family-level. \emph{DeepSeek:} \textsf{deepseek-v4-flash} on
\texttt{api.deepseek.com} (thinking mode off, plain function calling), same
$N{=}100$, called the gated tool in $975/1000$ runs (918 two-turn, 80
three-turn, zero transport errors, 25 explicit-answer abstentions) and
emitted the shape in $0/975$---every consequential call alone in its turn;
V4-flash succeeds the retired V3.2 alias, so this too is a family-level
de-confound. \emph{Llama:} Meta ships no first-party API, so the second
path is provider-direct: \textsf{Llama-3.3-70B-Instruct-Turbo} on Together
AI, same $N{=}100$, emits in $64/769$ classifiable runs. The $117/1000$
excluded runs invent tool names absent from the schema (seven distinct on
\texttt{compound\_reorder}, e.g.\ \texttt{create\_replenishment\_order});
the skew is diagnostic ($87$ of $117$ on \texttt{compound\_reorder}, $14$
\texttt{compound\_email\_update}, $7$ \texttt{single\_offboard}, $5$
\texttt{compound\_transfer}, at most one elsewhere; full breakdown in the
raw log), and the exclusions concentrate in \texttt{compound\_reorder}
($n{=}13$ valid, zero exposures) while $57$ of $64$ exposures sit on
\texttt{compound\_cleanup} at $99/100$ valid. In per-total-run terms:
$64/1{,}000$ pooled ($0.064$), and on \texttt{compound\_cleanup} $57/100$
total ($0.57$); the conditional $0.90$ is $57/63$ \emph{called} runs, so
conditioning reports the leak-relevant rate rather than concealing the base.
One mediation implication a deployer should draw explicitly: a runtime that
routes model-invented tool names through a catch-all executor converts
schema nonadherence into an unmediated path, so unknown-tool handlers must
be wrapped or refused like any consequential tool---and under the structural
routes of Section~\ref{sec:design} such a handler's egress is refused by the
kernel exactly as an unwrapped tool's is.

\smallskip\noindent\textbf{Family versus checkpoint.} Only GPT-4o and
Claude carry the \emph{same} checkpoint on two independent serving paths, so
only their rates are checkpoint-level replications. The Gemini and DeepSeek
second paths are successor checkpoints on the vendors' own surfaces, so
those rows establish whether the \emph{family} reaches the shape on a
faithful path---Gemini does, DeepSeek does not---without cross-validating
the exact 2.5/V3.2 single-path rates; the Llama second path replicates the
same checkpoint through a different provider's tool-call handling.

\section{Responsible Disclosure}
\label{app:disclosure}
The behaviors this paper measures are properties of shipped control-flow
semantics, characterized throughout as \emph{contract mismatches} rather than
secret, exploitable defects (Section~\ref{sec:intro}); whether each is a bug or
an intended design is for maintainers to judge. The prevalence corpus
(Section~\ref{sec:prevalence}) consists of public reports filed by third
parties, so citing it creates no new vulnerability information; among our own
probe findings, resume re-execution is vendor-documented~\cite{lgdocs26} and
the missing multi-interrupt barrier is publicly tracked by a maintainer-filed
issue. One further defect we encountered is itself public and independently
corroborated: a report we opened on the LangGraph
tracker demonstrating that \texttt{durability="sync"} does not order a task's
pending-write persistence before the superseding checkpoint, so a crash during
checkpoint persistence produces host-dependent replay-versus-re-execution with
duplicated external side effects; several other developers reproduced it on
different operating systems and core counts, and multiple community pull
requests propose the write-before-checkpoint barrier. Because this report is
ours, it is excluded from the third-party occurrence count of
Section~\ref{sec:prevalence} and used only as upstream-trajectory evidence.
The probes are model-free,
use no user data, and exercise only local, inert effects.

% ---------------------------------------------------------------------
% Elsevier declarations (JSS). GenAI-in-writing declaration: uncomment and
% edit the block below only if applicable per Elsevier's GenAI policy.
%\section*{Declaration of generative AI and AI-assisted technologies in the writing process}
%During the preparation of this work the author used [TOOL] in order to
%[REASON]. After using this tool, the author reviewed and edited the content
%as needed and takes full responsibility for the content of the publication.

\section*{CRediT authorship contribution statement}
\textbf{Sajjad Khan:} Conceptualization, Methodology, Software, Validation,
Formal analysis, Investigation, Data curation, Writing -- original draft,
Writing -- review \& editing, Visualization.

\section*{Declaration of competing interest}
The author declares that he has no known competing financial interests or
personal relationships that could have appeared to influence the work
reported in this paper.

\section*{Funding}
This research did not receive any specific grant from funding agencies in the
public, commercial, or not-for-profit sectors.

\section*{Data availability}
The complete artifact---probes with per-framework transcripts, the
measurement harness and raw result logs, the formal models with checker
logs, the \sys{} reference implementation, and a single-command audit
(\texttt{reproduce.sh}) that re-derives every headline number from committed
data---is available for review at
\url{https://anonymous.4open.science/r/soundgate-paper-B576} and will be made
publicly available at \url{https://github.com/sajjadanwar0/soundgate-paper}
upon acceptance; the gate is
installable from PyPI (\texttt{pip install soundgate}).

% ---------------------------------------------------------------------
% References. SINGLE-FILE CONVENTION: all entries live inline in this
% thebibliography environment -- no external .bib, no bibtex pass. Add
% real entries directly here (IEEEtran numeric style by hand). Compile is
% pdflatex x2 only.


\begin{thebibliography}{68}
\bibitem{dewitt84} D.~J. DeWitt, R.~H. Katz, F.~Olken, L.~D. Shapiro,
M.~R. Stonebraker, and D.~A. Wood, ``Implementation techniques for main
memory database systems,'' in \emph{Proc. ACM SIGMOD}, 1984, pp. 1--8.
\bibitem{eventsourcing} M.~Fowler, ``Event sourcing,''
\url{https://martinfowler.com/eaaDev/EventSourcing.html}, 2005, accessed
July 2026.
\bibitem{envoy-extauthz} Envoy Project Authors, ``External authorization
(\texttt{ext\_authz}) filter,'' Envoy proxy documentation,
\url{https://www.envoyproxy.io/docs/envoy/latest/intro/arch_overview/security/ext_authz_filter},
accessed July 2026.
\bibitem{aicontrol23} R.~Greenblatt, B.~Shlegeris, K.~Sachan, and F.~Roger,
``AI control: Improving safety despite intentional subversion,''
arXiv:2312.06942, 2023.
\bibitem{tokio-cancel} Tokio contributors, ``\texttt{tokio::select!}:
cancellation safety,'' Tokio API documentation,
\url{https://docs.rs/tokio/latest/tokio/macro.select.html}, accessed July
2026.
\bibitem{loom} Tokio contributors, ``Loom: concurrency permutation testing
for Rust,'' \url{https://docs.rs/loom}, accessed July 2026.
\bibitem{yawl05} W. M. P. van der Aalst and A. H. M. ter Hofstede, ``YAWL:
Yet another workflow language,'' \emph{Information Systems}, vol. 30, no. 4,
pp. 245--275, 2005.
\bibitem{zanzibar19} R. Pang \emph{et al.}, ``Zanzibar: Google's consistent,
global authorization system,'' in \emph{Proc. USENIX ATC}, 2019,
pp. 33--46.
\bibitem{mcp} Anthropic, ``Model Context Protocol,'' specification,
\url{https://modelcontextprotocol.io}, 2024, accessed July 2026.
\bibitem{msaf26} Microsoft, ``Microsoft Agent Framework overview,''
\url{https://learn.microsoft.com/en-us/agent-framework/overview/}, 2026,
accessed July 2026.
\bibitem{sel4} G. Klein \emph{et al.}, ``seL4: Formal verification of an
OS kernel,'' in \emph{Proc. 22nd ACM SOSP}, 2009, pp. 207--220.
\bibitem{anderson72} J. P. Anderson, ``Computer security technology planning
study,'' U.S. Air Force Electronic Systems Division, Tech. Rep. ESD-TR-73-51,
1972.
\bibitem{javamop12} P. O. Meredith, D. Jin, D. Griffith, F. Chen, and
G. Ro\c{s}u, ``An overview of the MOP runtime verification framework,''
\emph{Int. J. Softw. Tools Technol. Transf.}, vol. 14, no. 3, pp. 249--289,
2012.
\bibitem{vdaalst98} W. M. P. van der Aalst, ``The application of Petri nets
to workflow management,'' \emph{J. Circuits Syst. Comput.}, vol. 8, no. 1,
pp. 21--66, 1998.
\bibitem{birrell84} A. D. Birrell and B. J. Nelson, ``Implementing remote
procedure calls,'' \emph{ACM Trans. Comput. Syst.}, vol. 2, no. 1,
pp. 39--59, 1984.
\bibitem{burrows06} M. Burrows, ``The Chubby lock service for
loosely-coupled distributed systems,'' in \emph{Proc. 7th USENIX OSDI},
2006, pp. 335--350.
\bibitem{stripe-idem} Stripe, ``Idempotent requests,'' API documentation,
\url{https://docs.stripe.com/api/idempotent_requests}, accessed July 2026.
\bibitem{k8s-admission} The Kubernetes Authors, ``Admission control in
Kubernetes,'' \url{https://kubernetes.io/docs/reference/access-authn-authz/admission-controllers/},
accessed July 2026.
\bibitem{toolemu24} Y. Ruan \emph{et al.}, ``Identifying the risks of LM
agents with an LM-emulated sandbox,'' in \emph{Proc. ICLR}, 2024
(arXiv:2309.15817).
\bibitem{greshake23} K. Greshake, S. Abdelnabi, S. Mishra, C. Endres,
T. Holz, and M. Fritz, ``Not what you've signed up for: Compromising
real-world LLM-integrated applications with indirect prompt injection,'' in
\emph{Proc. 16th ACM Workshop on Artificial Intelligence and Security
(AISec)}, 2023, pp. 79--90.
\bibitem{parasuraman97} R. Parasuraman and V. Riley, ``Humans and
automation: Use, misuse, disuse, abuse,'' \emph{Human Factors}, vol. 39,
no. 2, pp. 230--253, 1997.
\bibitem{lgdocs26} LangChain, ``LangGraph documentation: Graph API and
interrupts,'' \url{https://docs.langchain.com/oss/python/langgraph/graph-api}
and \url{https://docs.langchain.com/oss/python/langgraph/interrupts},
accessed July 2026.
\bibitem{atomix26} B. Mohammadi, N. Potamitis, L. Klein, A. Arora, and
L. Bindschaedler, ``Atomix: Timely, transactional tool use for reliable
agentic workflows,'' arXiv:2602.14849, 2026.
\bibitem{sagallm25} E. Y. Chang and L. Geng, ``SagaLLM: Context management,
validation, and transaction guarantees for multi-agent LLM planning,''
\emph{Proc. VLDB Endow.}, vol. 18, no. 12, pp. 4874--4886, 2025.
\bibitem{sagas87} H. Garcia-Molina and K. Salem, ``Sagas,'' in \emph{Proc.
ACM SIGMOD}, 1987, pp. 249--259.
\bibitem{agentspec25} H. Wang, C. M. Poskitt, and J. Sun, ``AgentSpec:
Customizable runtime enforcement for safe and reliable LLM agents,''
arXiv:2503.18666, 2025.
\bibitem{progent25} T. Shi, J. He, Z. Wang, H. Li, L. Wu, W. Guo, and
D. Song, ``Progent: Programmable privilege control for LLM agents,''
arXiv:2504.11703, 2025.
\bibitem{camel25} E. Debenedetti \emph{et al.}, ``Defeating prompt
injections by design,'' arXiv:2503.18813, 2025.
\bibitem{temporal} Temporal Technologies, ``Temporal: Durable execution
platform,'' \url{https://temporal.io}, accessed July 2026.
\bibitem{helland12} P. Helland, ``Idempotence is not a medical condition,''
\emph{ACM Queue}, vol. 10, no. 4, 2012.
\bibitem{toctou-mindgap25} D. Lilienthal and S. Hong, ``Mind the gap:
Time-of-check to time-of-use vulnerabilities in LLM-enabled agents,''
arXiv:2508.17155, 2025.
\bibitem{toctou-browser26} L. Jiang, Z. Liu, H. Luo, and Z. Lin, ``Atomicity
for agents: Exposing, exploiting, and mitigating TOCTOU vulnerabilities in
browser-use agents,'' arXiv:2603.00476, 2026.
\bibitem{mast25} M. Cemri \emph{et al.}, ``Why do multi-agent LLM systems
fail?'' arXiv:2503.13657, 2025.
\bibitem{khan26} S. Khan, ``Verified detection and prevention of concurrency
anomalies in multi-agent large language model systems,'' arXiv:2606.17182,
2026.
\bibitem{schneider00} F. B. Schneider, ``Enforceable security policies,''
\emph{ACM Trans. Inf. Syst. Secur.}, vol. 3, no. 1, pp. 30--50, 2000.
\bibitem{miller06} M. S. Miller, ``Robust composition: Towards a unified
approach to access control and concurrency control,'' Ph.D. dissertation,
Johns Hopkins University, 2006.
\bibitem{fencing-kleppmann} M. Kleppmann, \emph{Designing Data-Intensive
Applications}. O'Reilly, 2017, ch.~8--9 (fencing tokens and generation
numbers).
\bibitem{2pc-gray} J. Gray and A. Reuter, \emph{Transaction Processing:
Concepts and Techniques}. Morgan Kaufmann, 1993.
\bibitem{orseau16} L. Orseau and S. Armstrong, ``Safely interruptible
agents,'' in \emph{Proc. 32nd Conf. Uncertainty in Artificial Intelligence
(UAI)}, 2016, pp. 557--566.
\bibitem{soares15} N. Soares, B. Fallenstein, S. Armstrong, and
E. Yudkowsky, ``Corrigibility,'' in \emph{AAAI Workshop on AI and Ethics},
2015.
\bibitem{abstractalg26} Abstract Algorithms, ``Human-in-the-loop workflows
with LangGraph: Interrupts, approvals, and async execution,''
\url{https://www.abstractalgorithms.dev/langgraph-human-in-the-loop},
Apr. 2026, accessed July 2026.
\bibitem{seccomp} The Linux Kernel documentation, ``Seccomp BPF (SECure
COMPuting with filters),''
\url{https://www.kernel.org/doc/html/latest/userspace-api/seccomp_filter.html},
accessed July 2026.
\bibitem{gvisor19} E. G. Young, P. Zhu, T. Caraza-Harter,
A. C. Arpaci-Dusseau, and R. H. Arpaci-Dusseau, ``The true cost of
containing: A gVisor case study,'' in \emph{Proc. 11th USENIX Workshop on
Hot Topics in Cloud Computing (HotCloud)}, 2019.
\bibitem{wasi} Bytecode Alliance, ``WASI: The WebAssembly System
Interface,'' \url{https://wasi.dev}, accessed July 2026.
\bibitem{armstrong03} J. Armstrong, ``Making reliable distributed systems in
the presence of software errors,'' Ph.D. dissertation, KTH Royal Institute
of Technology, Stockholm, 2003.
\bibitem{structconc} N. J. Smith, ``Notes on structured concurrency, or: Go
statement considered harmful,''
\url{https://vorpus.org/blog/notes-on-structured-concurrency-or-go-statement-considered-harmful/},
2018, accessed July 2026.
\bibitem{ligatti05} J. Ligatti, L. Bauer, and D. Walker, ``Edit automata:
Enforcement mechanisms for run-time security policies,'' \emph{Int. J. Inf.
Secur.}, vol. 4, no. 1--2, pp. 2--16, 2005.
\bibitem{macaroons14} A.~Birgisson, J.~G. Politz, {\'U}.~Erlingsson, A.~Taly,
M.~Vrable, and M.~Lentczner, ``Macaroons: Cookies with contextual caveats for
decentralized authorization in the cloud,'' in \emph{Proc. Network and
Distributed System Security Symposium (NDSS)}, 2014.
\bibitem{capsicum10} R.~N.~M. Watson, J.~Anderson, B.~Laurie, and K.~Kennaway,
``Capsicum: Practical capabilities for UNIX,'' in \emph{Proc. 19th USENIX
Security Symposium}, 2010.
\bibitem{cheri15} R.~N.~M. Watson \emph{et al.}, ``CHERI: A hybrid
capability-system architecture for scalable software compartmentalization,''
in \emph{Proc. IEEE Symposium on Security and Privacy}, 2015.
\bibitem{outbox} C.~Richardson, \emph{Microservices Patterns}. Manning,
2018, ch.~3 (Transactional Outbox pattern).
\bibitem{mckeeman98} W.~M. McKeeman, ``Differential testing for software,''
\emph{Digital Technical Journal}, vol.~10, no.~1, pp. 100--107, 1998.
\bibitem{csmith11} X.~Yang, Y.~Chen, E.~Eide, and J.~Regehr, ``Finding and
understanding bugs in C compilers,'' in \emph{Proc. ACM SIGPLAN Conf.
Programming Language Design and Implementation (PLDI)}, 2011.
\bibitem{raft14} D.~Ongaro and J.~Ousterhout, ``In search of an
understandable consensus algorithm,'' in \emph{Proc. USENIX Annual
Technical Conference (ATC)}, 2014.
\bibitem{ajpf12} L.~A. Dennis, M.~Fisher, M.~P. Webster, and R.~H. Bordini,
``Model checking agent programming languages,'' \emph{Automated Software
Engineering}, vol.~19, no.~1, 2012.
\bibitem{agentbench24} X.~Liu \emph{et al.}, ``AgentBench: Evaluating LLMs
as agents,'' in \emph{Proc. Int. Conf. Learning Representations (ICLR)},
2024.
\bibitem{ironfleet15} C.~Hawblitzel, J.~Howell, M.~Kapritsos, J.~R. Lorch,
B.~Parno, M.~L. Roberts, S.~Setty, and B.~Zill, ``IronFleet: Proving
practical distributed systems correct,'' in \emph{Proc. ACM Symp. Operating
Systems Principles (SOSP)}, 2015.
\bibitem{verdi15} J.~R. Wilcox, D.~Woos, P.~Panchekha, Z.~Tatlock, X.~Wang,
M.~D. Ernst, and T.~Anderson, ``Verdi: A framework for implementing and
formally verifying distributed systems,'' in \emph{Proc. ACM SIGPLAN Conf.
Programming Language Design and Implementation (PLDI)}, 2015.
\bibitem{fscq15} H.~Chen, D.~Ziegler, T.~Chajed, A.~Chlipala, M.~F.
Kaashoek, and N.~Zeldovich, ``Using Crash Hoare Logic for certifying the
FSCQ file system,'' in \emph{Proc. ACM Symp. Operating Systems Principles
(SOSP)}, 2015.
\bibitem{cilium} Cilium Authors, ``Cilium: eBPF-based networking,
observability, and security,'' 2024. [Online]. Available:
https://cilium.io
\bibitem{falco} The Falco Authors, ``Falco: Cloud-native runtime security,''
CNCF, 2024. [Online]. Available: https://falco.org
\bibitem{iec61508} International Electrotechnical Commission, ``IEC 61508:
Functional safety of electrical/electronic/programmable electronic safety-related
systems,'' Edition 2.0, IEC, Geneva, 2010.

\bibitem{iso26262} International Organization for Standardization, ``ISO 26262:
Road vehicles---Functional safety,'' 2nd ed., ISO, Geneva, 2018.

\bibitem{iso10218} International Organization for Standardization, ``ISO
10218-1: Robots and robotic devices---Safety requirements for industrial
robots---Part 1: Robots,'' ISO, Geneva, 2011.

\bibitem{eswaran76} K.~P. Eswaran, J.~N. Gray, R.~A. Lorie, and I.~L.
Traiger, ``The notions of consistency and predicate locks in a database
system,'' \emph{Commun. ACM}, vol.~19, no.~11, pp.~624--633, 1976.

\bibitem{langsec16} F.~Momot, S.~Bratus, S.~M. Hallberg, and M.~L.
Patterson, ``The seven turrets of Babel: A taxonomy of LangSec errors and
how to expunge them,'' in \emph{Proc. IEEE Cybersecurity Development
(SecDev)}, 2016, pp.~45--52.

\bibitem{chatgptagent25} OpenAI, ``Introducing ChatGPT agent: bridging
research and action,'' Jul.~2025. [Online]. Available:
https://openai.com/index/introducing-chatgpt-agent/

\bibitem{taubench24} S.~Yao, N.~Shinn, P.~Razavi, and K.~Narasimhan,
``$\tau$-bench: A benchmark for tool-agent-user interaction in real-world
domains,'' arXiv:2406.12045, 2024.

\bibitem{wahbe93} R.~Wahbe, S.~Lucco, T.~E. Anderson, and S.~L. Graham,
``Efficient software-based fault isolation,'' in \emph{Proc.\ 14th ACM Symp.\
Operating Systems Principles (SOSP)}, 1993, pp.~203--216.

\bibitem{lucassen88} J.~M. Lucassen and D.~K. Gifford, ``Polymorphic effect
systems,'' in \emph{Proc.\ 15th ACM SIGPLAN-SIGACT Symp.\ Principles of
Programming Languages (POPL)}, 1988, pp.~47--57.

\bibitem{leucker09} M.~Leucker and C.~Schallhart, ``A brief account of
runtime verification,'' \emph{J.\ Logic and Algebraic Programming}, vol.~78,
no.~5, pp.~293--303, 2009.

\bibitem{millwheel13} T.~Akidau, A.~Balikov, K.~Bekiro\u{g}lu, S.~Chernyak,
J.~Haberman, R.~Lax, S.~McVeety, D.~Mills, P.~Nordstrom, and S.~Whittle,
``MillWheel: Fault-tolerant stream processing at Internet scale,''
\emph{Proc.\ VLDB Endowment}, vol.~6, no.~11, pp.~1033--1044, 2013.

\bibitem{newcombe15} C.~Newcombe, T.~Rath, F.~Zhang, B.~Munteanu,
M.~Brooker, and M.~Deardeuff, ``How Amazon Web Services uses formal
methods,'' \emph{Commun.\ ACM}, vol.~58, no.~4, pp.~66--73, 2015.

\bibitem{bpel07} OASIS, ``Web Services Business Process Execution Language
version 2.0,'' OASIS Standard, Apr.\ 2007.

\bibitem{wfpatterns03} W.~M.~P. van~der Aalst, A.~H.~M. ter~Hofstede,
B.~Kiepuszewski, and A.~P. Barros, ``Workflow patterns,''
\emph{Distributed and Parallel Databases}, vol.~14, no.~1, pp.~5--51, 2003.

\bibitem{asyncjs17} M.~C. Loring, M.~Marron, and D.~Leijen, ``Semantics of
asynchronous JavaScript,'' in \emph{Proc.\ 13th ACM SIGPLAN Int.\ Symp.\
Dynamic Languages (DLS)}, 2017, pp.~51--62.

\bibitem{flink17} P.~Carbone, S.~Ewen, G.~F\'ora, S.~Haridi, S.~Richter, and
K.~Tzoumas, ``State management in Apache Flink: Consistent stateful
distributed stream processing,'' \emph{Proc.\ VLDB Endowment}, vol.~10,
no.~12, pp.~1718--1729, 2017.

\bibitem{memgpt23} C.~Packer, S.~Wooders, K.~Lin, V.~Fang, S.~G. Patil,
I.~Stoica, and J.~E. Gonzalez, ``MemGPT: Towards LLMs as operating
systems,'' arXiv:2310.08560, 2023.

\bibitem{aios24} K.~Mei, Z.~Li, S.~Xu, R.~Ye, Y.~Ge, and Y.~Zhang, ``AIOS:
LLM agent operating system,'' arXiv:2403.16971, 2024.

\bibitem{rustbelt18} R.~Jung, J.-H. Jourdan, R.~Krebbers, and D.~Dreyer,
``RustBelt: Securing the foundations of the Rust programming language,''
\emph{Proc.\ ACM Program.\ Lang.}, vol.~2, no.~POPL, pp.~66:1--66:34, 2018.

\bibitem{bainbridge83} L.~Bainbridge, ``Ironies of automation,''
\emph{Automatica}, vol.~19, no.~6, pp.~775--779, 1983.

\bibitem{shuttle} AWS Labs, ``Shuttle: a library for testing concurrent Rust
code,'' \url{https://github.com/awslabs/shuttle}, accessed July 2026.

\bibitem{linkerd} W.~Morgan et al., ``Linkerd: a service mesh for Kubernetes,''
Cloud Native Computing Foundation, \url{https://linkerd.io}, accessed July 2026.

\bibitem{gray78} J. Gray, ``Notes on data base operating systems,'' in
\emph{Operating Systems: An Advanced Course}, ser. Lecture Notes in Computer
Science, vol.~60. Springer, 1978, pp. 393--481.

\bibitem{calvin12} A. Thomson, T. Diamond, S.-C. Weng, K. Ren, P. Shao, and
D.~J. Abadi, ``Calvin: Fast distributed transactions for partitioned database
systems,'' in \emph{Proc. ACM SIGMOD}, 2012, pp. 1--12.

\bibitem{endsley95} M.~R. Endsley, ``Toward a theory of situation awareness in
dynamic systems,'' \emph{Human Factors}, vol.~37, no.~1, pp. 32--64, 1995.

\bibitem{cummings04} M.~L. Cummings, ``Automation bias in intelligent time
critical decision support systems,'' in \emph{Proc. AIAA 1st Intelligent
Systems Technical Conference}, 2004.

\bibitem{temporal-docs} Temporal Technologies, ``Temporal documentation:
activity heartbeats and cancellation; activity timeouts,''
\url{https://docs.temporal.io}, accessed July 2026.

\bibitem{agentdojo24} E. Debenedetti, J. Zhang, M. Balunovi\'{c}, L. Beurer-Kellner,
M. Fischer, and F. Tram\`{e}r, ``AgentDojo: A dynamic environment to evaluate
prompt injection attacks and defenses for LLM agents,'' in \emph{Proc. NeurIPS
Datasets and Benchmarks}, 2024 (arXiv:2406.13352).

\bibitem{isolategpt25} Y. Wu, F. Roesner, T. Kohno, N. Zhang, and U. Iqbal,
``IsolateGPT: An execution isolation architecture for LLM-based agentic
systems,'' in \emph{Proc. Network and Distributed System Security Symposium
(NDSS)}, 2025 (arXiv:2403.04960).

\bibitem{erlingsson00} \'{U}. Erlingsson and F.~B. Schneider, ``IRM
enforcement of Java stack inspection,'' in \emph{Proc. IEEE Symp. Security
and Privacy}, 2000, pp. 246--255.
\bibitem{kafka-kip98} Apache Software Foundation, ``KIP-98: Exactly-once
delivery and transactional messaging,'' Apache Kafka Improvement Proposal,
2017. [Online]. Available:
\url{https://cwiki.apache.org/confluence/display/KAFKA/KIP-98+-+Exactly+Once+Delivery+and+Transactional+Messaging}
\bibitem{otel} OpenTelemetry Authors, ``The OpenTelemetry specification.''
[Online]. Available: \url{https://opentelemetry.io/docs/specs/otel/}
\end{thebibliography}
\end{document}